\documentclass[12pt,fleqn]{article}
\usepackage{lscape}
\usepackage{amsfonts,ulem,url}
\usepackage{xr-hyper}
\usepackage{amssymb}
\usepackage{amsmath}
\usepackage{amstext,verbatim,graphicx,ifthen,multirow,amsthm,algorithm2e, mathtools, array}
\usepackage{bm}
\usepackage[round]{natbib}
\usepackage[bf]{caption}
\usepackage{physics}
\usepackage{setspace}
\usepackage{enumitem}
\usepackage{placeins}
\usepackage{xcolor}
\usepackage{rotating}
\usepackage{adjustbox}
\usepackage{appendix}

\let\P\relax
\DeclareMathOperator{\P}{\mathbb{P}}
\DeclareMathOperator{\Pn}{\mathbb{P}_n}
\DeclareMathOperator{\E}{\mathbb{E}}

\usepackage{float}

\newcommand\independent{\protect\mathpalette{\protect\independenT}{\perp}}
\def\independenT#1#2{\mathrel{\rlap{$#1#2$}\mkern2mu{#1#2}}}

\newtheorem{assumption}{Assumption}[section]
\newtheorem{theorem}{Theorem}[section]
\newtheorem{lemma}{Lemma}[section]

\newtheorem{corollary}{Corollary}[section]

\theoremstyle{definition}

\newtheorem{remark}{Remark}[section]

\numberwithin{equation}{section}

\usepackage{apptools}
\AtAppendix{\counterwithin{lemma}{section}}
\AtAppendix{\counterwithin{theorem}{section}}
\AtAppendix{\counterwithin{corollary}{section}}

\DeclareMathOperator*{\argmin}{arg\,min}

\newcommand{\bias}{\overline{\text{bias}}}
\newcommand{\m}[1]{\mathcal{#1}}
\newcommand{\mt}[1]{\tilde{\mathcal{#1}}}
\let\norm\relax
\DeclarePairedDelimiter{\norm}{\lVert}{\rVert}

\newcommand{\lo}{\theta}
\newcommand{\muplo}{\tilde\theta^{\mu}}
\newcommand{\plo}{\tilde\theta}
\newcommand{\hlo}{\hat \theta}

\newcommand{\bestmu}{\mu^{\star}}
\newcommand{\pmu}{\tilde \mu}
\newcommand{\hmu}{\hat\mu}
\newcommand{\Tef}{T_{e}}
\newcommand{\Pnn}{\frac{1}{n}\sum_{i=1}^n}

\newcommand{\spn}[1]{\textbf{span}\{#1\}}

\graphicspath{{./figures/}}

\usepackage[colorlinks=true,unicode,bookmarksnumbered=false, hyperfootnotes=true]{hyperref}
\hypersetup{citecolor = blue}

\usepackage[colorinlistoftodos]{todonotes}

\begin{document}

%\listoftodos

\title{Large-Sample Properties of the Synthetic Control Method under Selection on Unobservables\thanks{{\small We are grateful for comments by seminar participants at  CREST, Stanford GSB, EUI, CEMFI, UCL, Cambridge, LSE, Oxford, 43 Meeting of the Brazilian Econometric Society, Synthetic Control conference at Princeton University, ICSDS in Florence,  African Meeting of Econometric Society, Alberto Abadie, Manuel Arellano, Stephane Bonhomme, Guido Imbens, and Stefan Wager. This research was generously supported by the research grant from XIX Concurso Nacional de Ayudas a la Investigacion en Ciencias Sociales, Fundacion Ramon Areces.}}} 
\author{Dmitry  Arkhangelsky \thanks{{\small  Associate Professor, CEMFI, CEPR, darkhangel@cemfi.es. }} \and David Hirshberg \thanks{{\small Assistant Professor, Emory University QTM, davidahirshberg@emory.edu.
}} }\date{\today}

\maketitle

\begin{abstract}
\singlespacing
We analyze the synthetic control (SC) method in panel data settings with many units. We assume the treatment assignment is based on unobserved heterogeneity and pre-treatment information, allowing for both strictly and sequentially exogenous assignment processes. We show that the critical property that determines the behavior of the SC method is the ability of input features to approximate the unobserved heterogeneity. Our results imply that the SC method delivers asymptotically normal estimators for a large class of linear panel data models as long as the number of pre-treatment periods is sufficiently large, making it a natural alternative to the Difference-in-Differences. 
\end{abstract}


\newpage
\section{Introduction}

Methods based on Difference in Differences (DiD) have had a tremendous impact on applied research in economics and social sciences more broadly. \cite{currie2020technology} document the prevalence of DiD in empirical practice and show that it is the dominant method for estimating treatment effects with observational data. This success can be attributed to the unique combination of transparency and flexibility of the DiD estimator. The flip side of this is that the validity of these methods relies on assumptions that severely restrict both the model for the counterfactual outcomes and the treatment assignment process. For the DiD estimator to work, the counterfactual outcomes for treated units should evolve in parallel to the control ones, which, outside of exceptional cases, requires the underlying outcomes to follow a two-way model and the treatment assignment to be based only on permanent characteristics (\citealp{ghanem2022selection}). This problem has been recognized for a long time: in one of the early applications of DiD, \cite{ashenfelter1985using} document that the data soundly reject the DiD assumptions. 

There are two different ways of resolving this tension. One possibility is restricting attention to applications and datasets where the DiD assumptions will likely hold. Empirical practice, particularly the reliance on tests for parallel trends, suggests this is not an uncommon solution. This process leads to well-understood inferential problems \citep{roth2022pretest,rambachan2023more} and more broadly contributes to publication bias (e.g., \citealp{andrews2019identification}). The alternative route is adopting methods that remain valid in environments where approaches based on DiD fail, which we focus on in this paper. 

%Any such method is bound to be more complex than the DiD, but it can retain its attractive features.

 Textbooks on panel data (e.g., \citealp{arellano2003panel,wooldridge2010econometric}) describe various models that substantially relax the DiD assumptions. These models can be estimated using the generalized method of moments (GMM), leading to estimators with well-understood statistical properties. At the same time, this process can be quite fragile and often requires multiple choices from the user (e.g., \citealp{blundell1998initial}). More importantly, each model leads to a different estimator -- a dramatic contrast with the simplicity and transparency of DiD. 

In this paper, we argue that there is a single estimation strategy that delivers valid answers in a large class of panel data models, including those where the DiD approach fails. This strategy is based on adapting the Synthetic Control (SC) method \citep{abadie2003,Abadie2010} to panel data applications. The SC method exploits the information available before a specific policy (treatment) was adopted to calculate weights for units not exposed to the policy. These weights are then utilized to construct a counterfactual path for the exposed units. The key input for the SC method is a set of features -- functions of the available past information -- used to calculate the weights. 
%This set commonly consists of levels of the pre-treatment outcomes, and we refer to the resulting implementation as the SC estimator.  

The SC method is already one of the key tools for policy evaluation in social sciences.\footnote{Some recent studies using this method include \citet{cavallo2013catastrophic,andersson2019carbon, mitze2020face,jones2022labor}.} At the same time, it was originally designed for comparative case studies, which have few available units and often a single treated unit. Most empirical applications that use the SC method have a similar structure. In contrast, we focus on environments where the number of units is large and possibly much larger than the number of observed pre-treatment periods. In this way, we can reinterpret the SC method as a general-purpose algorithm encompassing applications where researchers would normally rely on DiD. To justify this interpretation, we derive a set of new statistical results that establish the properties of the SC method in large samples. 

Our analysis is based on two structural conditions that describe the interplay between the treatment assignment and the outcome model. First, we assume that the treatment is independent of future counterfactual outcomes conditional on permanent unobserved characteristics and observed pre-treatment information. This sequential exogeneity restriction is natural for many economic applications and substantially generalizes the selection assumptions that underlie the DiD analysis. Second, we assume that the relevant unobservables can be recovered with precision when the number of pre-treatment periods is sufficiently large. This retrievability restriction is routinely made in panel data models (e.g., \citealp{bonhomme2022discretizing}) and is critical for our analysis. 

We show that the SC method has desirable statistical properties if the input features are rich enough to approximate the relevant unobservables. In particular, we prove that the SC method delivers asymptotically normal estimators if the approximation error decreases fast enough. As a direct application of this general result, we demonstrate consistency and asymptotic normality of the SC estimator in a model with two-way fixed effects where the treatment assignment is based on both permanent and time-varying components of the outcomes. It has been long-established that such selection patterns naturally arise in economic applications  (e.g., \citealp{ashenfelter1985using}). The DiD estimator is inconsistent in this environment, even though the baseline counterfactual outcomes follow a two-way model. 

Our results extend beyond the two-way setup and remain valid in models with interactive fixed effects (e.g., \citealp{bai2009panel,moon2015linear,moon2017dynamic}). In contrast to a significant portion of the literature on interactive fixed effects, we do not rely on fixed rank assumptions and prove the asymptotic normality of the SC estimator for models where the dimension of fixed effects increases with sample size. This analysis better describes the interplay between finite and large $T$ regimes, showing that the convergence rate of the estimator is slower in a more flexible model. Our conclusions are in line with the recent results established in \cite{freeman2023linear} for a semiparametric model with unobserved two-way heterogeneity (see also \citealp{beyhum2022factor}). A similar performance was shown in \cite{arkhangelsky2021synthetic} for the Synthetic DiD estimator. Notably, the last result was established for a model with strictly exogenous treatment assignment. 

Our findings also reveal potential shortcomings of the SC method. Motivated by our theoretical results, we use simulations to demonstrate that the SC estimator fails in two-way environments when there is sizable heterogeneity in the persistence of the time-varying shocks across units, which affects the assignment. The failure of the SC estimator in this environment can be viewed as a consequence of using an insufficiently rich set of features to construct the weights. Linear combinations of pre-treatment outcomes capture the differences in the means but cannot distinguish the differences in the persistence. By including unit-specific measures of persistence, researchers can potentially alleviate this problem. More broadly, our results show that for the SC method to be consistent, researchers must choose a set of features that approximate the underlying heterogeneity. 

Our formal results contribute to several strands of the literature. First, we generalize the available statistical guarantees for the SC method in several dimensions (see \citealp{abadie2021using} for a recent survey). We derive an asymptotic expansion
for the SC method under high-level assumptions on the input features
that holds when the number of treated units is asymptotically larger than the number of time periods. This result can be applied to build a variety of estimators, allowing researchers to tailor the general strategy to their specific applications. We provide conditions under which the SC method delivers unbiased and asymptotically normal estimators and thus can be used for inference. This result expands the range of applications where the SC method can be credibly used.
%\dmitry{For Skip: to write a discussion about the asymptotic regime.}

Our results have direct implications for the econometric panel data literature. We show that the SC method delivers asymptotically normal estimators in a large class of linear panel data models. Some of these models can be used directly to consistently estimate the causal parameter of interest using GMM, even if the number of periods is small. Our analysis shows that the same parameters can be consistently recovered using linear estimators if the number of periods is large. This insight is connected to the literature on the biases of fixed effects estimators in linear panel data models \citep{nickell1981biases,hahn2002asymptotically,alvarez2003time}. The critical difference is that the SC method has this property simultaneously for a large class of panel data models, whereas the fixed effect estimators target a particular one. 

The idea that functions of pre-treatment outcomes can be a reasonable alternative to quasi-differencing schemes estimated by GMM is not new in theoretical and empirical research. In \cite{blundell1999market}, the authors directly use averages of pre-treatment firm-level histories to account for unobserved heterogeneity. In \cite{blundell2002individual}, the authors show that the resulting estimator, which they call the pre-sample mean estimator, has attractive properties in simulations, especially when compared with the GMM procedures. In a much earlier work, \cite{chamberlain1982multivariate} suggested using long lags of outcomes to control for the unobserved heterogeneity. Our results demonstrate the connection between all these proposals and the SC method and provide statistical guarantees for a large class of similar estimators. 

The motivation for the pre-sample mean estimator is straightforward. To the extent that unobservables have any meaning, they have to be part of some observable variables, and past outcomes are the most natural candidates for such connections. For instance, in the two-way model that forms the basis of the DiD estimation, the relevant unobservables -- unit fixed effects -- are part of the outcomes. As a result, the average of the pre-treatment outcomes is a good proxy for the fixed effects in this model. In models with more complex structures, such as interactive fixed effects, one cannot rely on a single average and needs to construct multiple proxies. We show that the SC method accomplishes this automatically by examining all possible linear combinations of the input features. This interpretation of the SC method suggests that other time-varying covariates that contain information on relevant unobservables should also be used as input features. In the paper, we discuss several examples in which such covariates dramatically improve the performance of the SC method.

We also contribute to the balancing literature (e.g., \citealp{graham2012inverse,imai2014covariate,zubizarreta2015stable,athey2018approximate,tan2020model,wang2020minimal,armstrong2018finite, hirshberg2021augmented}) by deriving the properties of a particular balancing estimator in environments where important confounders are unobserved. We do this by arguing that unobservables create a misspecification problem that becomes negligible in a specific limit. This interpretation leads to a high-dimensional problem, which cannot be addressed by relying on sparsity assumptions (e.g., \citealp{belloni2014jep,chernozhukov2018double}). Instead, we use a built-in independence property of balancing estimators: the weights do not directly depend on the outcomes and thus are conditionally mean-independent from them. We use this fact to derive the asymptotic expansion for the SC method under relatively mild conditions. 

Our analysis has several limitations, the biggest being our focus on block designs where all units are treated simultaneously. A natural next step would be extending our results to environments with staggered designs, where units adopt the treatment sequentially. We discuss an adaptation of the SC method that can be applied to estimate contemporaneous effects in such applications. A complete treatment of dynamic effects in such settings is challenging even without unobservables; see \cite{viviano2021dynamic} for a modern approach and references. 

The paper proceeds as follows. In Section \ref{sec:setup}, we introduce the SC method and our key assumptions, discuss examples, and present numerical experiments that demonstrate the performance of the SC method in empirically relevant contexts. In Section \ref{sec:results}, we establish the asymptotic properties of the SC method and discuss the underlying statistical assumptions. In Section \ref{sec:exm}, we discuss our results in the context of linear panel data models. In Section \ref{sec:discussion}, we discuss an adaptation of our results to staggered designs. Finally, Section \ref{sec:conc} concludes.

\paragraph{Notation:} For a vector $x \in \mathbb{R}^p$ we use $\|x\|_2$ to denote its $l^2$ norm; for a random variable $X$ we use $\|X\|_2$ to denots its $L^2$ norm. For an arbitrary matrix $X$, we use $\| X\|_{op}$ to denote its operator norm. For a random vector $X$, we write $\mathbb{E}[X]$ for its expectation and $\mathbb{V}[X]$ for its covariance matrix. We use $\mathbf{1}$ to denote a constant function and $\mathcal{I}_d$ to denote the $d\times d$ identity matrix. For two determinsitic sequences $a_n$ and $b_n$ we write $a_n \sim b_n$ if sequences $\frac{a_n}{b_n}$  and $\frac{b_n}{a_n}$ are well-defined and bounded. We write $a_n\gtrsim b_n$ if $\frac{b_n}{a_n}$ is bounded, and $a_n \gg b_n$ if $\frac{b_n}{a_n}$ converges to zero, possibly up to log factors.

\section{SC method}\label{sec:setup}
This section introduces the SC method and our key conceptual assumptions. We then discuss several examples, demonstrating the scope of our assumptions. We close this section with a Monte Carlo study that shows the advantages of the SC method over the DiD estimator. The goal of this section is to convince the applied reader that the SC method is a natural alternative to the DiD in relatively simple environments. Our formal analysis in Section \ref{sec:results} demonstrates that this is also the case in a much larger class of models.

\subsection{Estimation Approach}\label{sec:estimator}
\begin{figure}[t!]
\includegraphics[width=.6\textwidth]{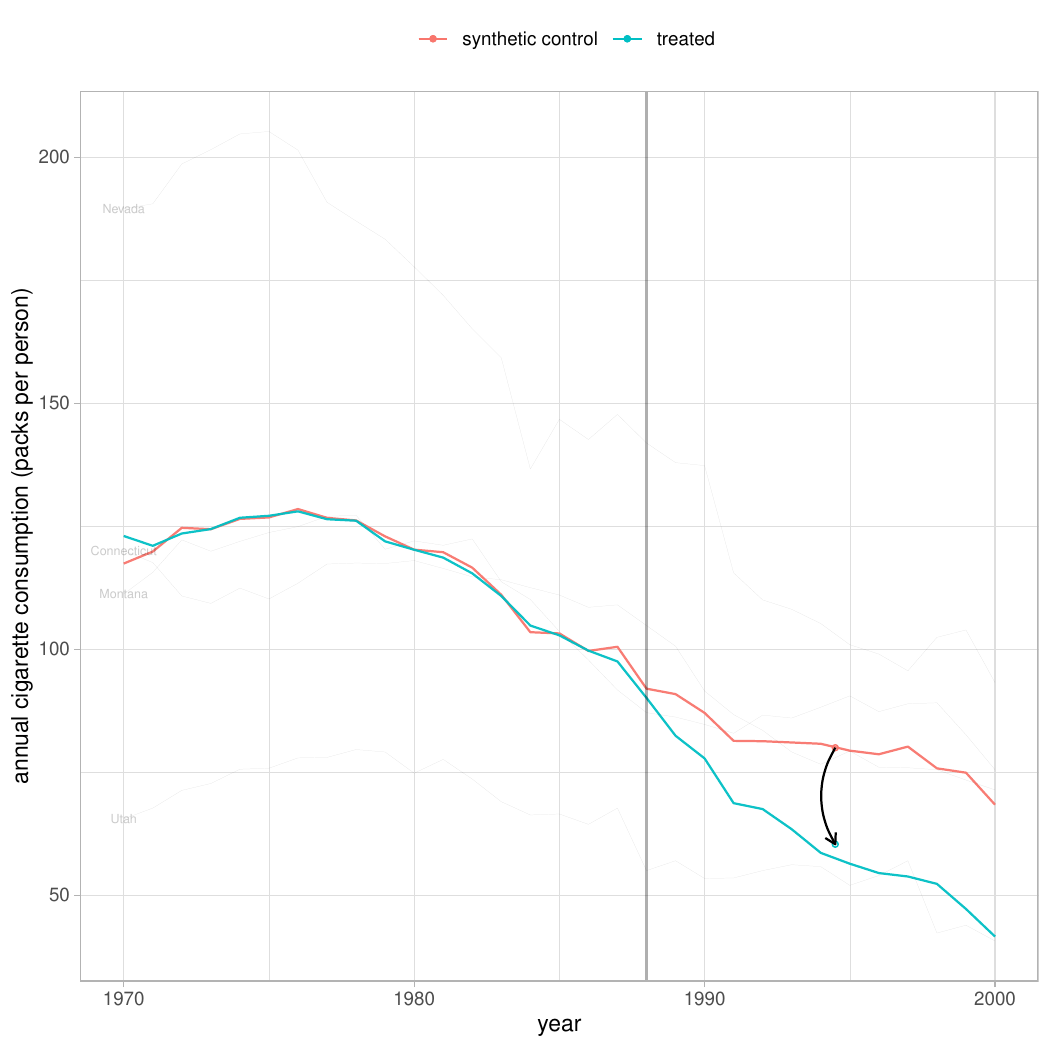}
\caption{\footnotesize Replication of the results in \cite{Abadie2010}.}\label{fig:sc}
\end{figure}

%% CODE
%library(synthdid)
%setup = panel.matrices(california_prop99)
%estimate = sc_estimate(setup$Y, setup$N0, setup$T0)
%top.controls = synthdid_controls(estimate)[1:4, , drop=FALSE]
%plot(estimate, spaghetti.units=rownames(top.controls), trajectory.alpha=.9,
%	       spaghetti.line.alpha=.1, spaghetti.label.alpha=.2, spaghetti.line.width=.05) +
%          xlab('year') + ylab('annual cigarette consumption (packs per person)')
%ggsave('synthetic-control-california-smoking.pdf')

We consider settings in which the researcher has access to a dataset $\m{D}:= \{X_i, D_i, Y_i\}_{i=1}^n$, with $n$ units in total. Here $Y_i$ is an outcome of interest, $D_i\in \{0,1\}$ is a binary treatment, and $X_i$ describes data available for unit $i$ from $T_0$ pre-treatment periods. The SC estimators we consider involve a post-treatment comparison of an average of treated units to a weighted average of control units with similar pre-treatment outcomes. For weights $\hat\omega_i$, it has this form:
\begin{equation}\label{eq:estimator}
    \hat \tau := \frac{1}{n\overline{\pi}} \sum_{i=1}^n D_i Y_{i} - \Pnn \hat \omega_i(1-D_i) Y_{i} \qfor \bar \pi: = \frac{1}{n}\sum_{i=1}^n D_i.
\end{equation}

In Figure \ref{fig:sc}, we see a typical comparison. We re-estimate the effect of California's 1988 cigarette tax on per-capita cigarette consumption, as discussed in \citet{Abadie2010}, by comparing the post-treatment cigarette consumption in California to that of an average of states which, prior to treatment, had a similar rate of cigarette consumption. In this case, the components $X_{i1},\dots, X_{i,T_0}$ of the vector $X_i$ are cigarette consumption rates in years prior to 1988. As a similarity criterion, we have used mean-squared-error; in particular, we have chosen the weights $\{\hat \omega_{i}\}_{i \le n}$ by solving the following entropy-regularized least squares problem:
\begin{equation}\label{eq:classic_synth}
\begin{aligned}
   \hat \omega := \argmin_{\omega\ge 0}&\left\{\frac{\zeta^2}{n^2}\sum_{i=1}^n \omega_i \log(\omega_i) + \sum_{t=1}^{T_0} \qty(\frac{1}{n\overline{\pi}}\sum_{i=1}^n D_i X_{it} - \Pnn \omega_i (1-D_i) X_{it} )^2 \right\}\\
\text{subject to: } &\Pnn \omega_i(1-D_i) = 1.
\end{aligned}
\end{equation}

There is a natural interpretation of this optimization problem as protecting us against bias when, absent treatment, future outcomes would be predicted linearly by past ones. In particular, using the dual characterization of the Euclidian norm, $\norm{x}_2 = \max_{u : \norm{v}_2 \le 1} v^T x$, we can rewrite \eqref{eq:classic_synth} as follows.
\begin{equation}\label{eq:classic_synth_dual}
\begin{aligned}
   \hat \omega := \argmin_{\omega\ge 0}\max_{v : \norm{v} \le 1} &\left\{\frac{\zeta^2}{n^2}\sum_{i=1}^n \omega_i \log(\omega_i) + \qty( \frac{1}{n\overline{\pi}}\sum_{i=1}^n D_i (v^T X_i) - \Pnn \omega_i (1-D_i) (v^T X_i) )^2 \right\}\\
\text{subject to: } &\Pnn \omega_i(1-D_i) = 1.
\end{aligned}
\end{equation}
As a result, by averaging the outcomes with the weights $\hat \omega_i$ and taking the difference as in \eqref{eq:estimator}, we essentially eliminate any systematic variation in the future outcomes that can be predicted by linear combinations $v^T X_i$.

More generally, when we expect future outcomes to be predicted by a function $f(X_i)$ in some set $\m{F}$, we might instead consider this problem:
\begin{equation}\label{eq:general_synth}
\begin{aligned}
   \hat \omega := \argmin_{\omega\ge 0}\max_{f \in \m{F}} &\left\{\frac{\zeta^2}{n^2}\sum_{i=1}^n \omega_i \log(\omega_i) + \qty( \frac{1}{n\overline{\pi}}\sum_{i=1}^n D_i f(X_i) - \Pnn \omega_i (1-D_i) f(X_i) )^2 \right\}\\
\text{subject to: } &\Pnn \omega_i(1-D_i) = 1.
\end{aligned}
\end{equation}
This interpretation, borrowing from the literature on covariate balance \citep[e.g][]{ben2021balancing}, is discussed in \citet{ben2018augmented}.
Here we will focus on a set $\m{F}$ of predictors that are linear in features $\phi_1(X_i) \ldots \phi_p(X_i)$ of our pre-treatment observations,
\begin{equation}\label{eq:main_text_set}
  \m{F} = \left\{f: \sum_{k=1}^p \beta_k \phi_k(x),\, \sum_{k=1}^p\beta_k^2 \le 1 \right\}.
\end{equation}
As a result, we will be working with weights chosen via the following special case of \eqref{eq:general_synth}.
\begin{equation}\label{eq:pr_problem}
\begin{aligned}
   \hat \omega := \argmin_{\omega\ge 0}&\left\{\frac{\zeta^2}{n^2}\sum_{i=1}^n \omega_i\log(\omega_i) + \sum_{l=1}^p \qty(\frac{1}{n\overline{\pi}}\sum_{i=1}^n D_i \phi_l(X_i) - \Pnn \omega_i (1-D_i) \phi_l(X_i) )^2 \right\}\\
\text{subject to: } &\Pnn \omega_i(1-D_i) = 1.
\end{aligned}
\end{equation}
In the California example above, as is typical, each feature is one of $T_0$ pre-treatment outcomes, i.e., $\phi_l(X_i)=X_{il}$ for $l=1 \ldots T_0$. However, our formulation also allows for additional time-varying covariates, and in the coming sections we will discuss several such examples. 
%Time-invariant covariates that represent observed cross-sectional heterogeneity, e.g., region-level indicators, can be incorporated into the constraint in \eqref{eq:pr_problem}. 

\subsection{Estimation Target}
To define our target estimand, we interpret the observed data using the potential outcomes framework (\citealp{neyman1923,rubin1974estimating}). We assume that for the units not exposed to the treatment, we observe the baseline outcomes $Y_i(0)$, while for the treated ones, we observe $Y_{i}(1)$. We assume that the counterfactual pre-treatment outcomes are not affected by the treatment, $X_i = X_{i}(0) = X_{i}(1)$. This restriction is usually implicit in the standard cross-sectional settings where $X_i$ contains fixed attributes (e.g., \citealp{imbens2015causal}). In our environment, $X_i$ contains pre-treatment outcomes, and this requirement should be interpreted as a no-anticipation assumption (e.g., \citealp{abbring2003nonparametric}). 
\begin{assumption}\label{as:pot_outcomes}\textsc{(Potential Outcomes)}\\
    There exist potential outcomes $X_i(0),\, X_{i}(1),\, Y_i(1),\, Y_i(0)$, where    $Y_i= D_i Y_i(1) + (1-D_i) Y_i(0)$ and $X_i = X_{i}(0) = X_{i}(1)$.
\end{assumption}

With this assumption we can decompose $\hat \tau$ into two parts:
\begin{equation*}
    \hat \tau = \Pnn \frac{D_i}{\overline{\pi}}(Y_{i}(1) - Y_{i}(0)) + 
    \Pnn \left(\frac{D_i}{\overline{\pi}} - \hat \omega_i(1-D_i)\right)Y_{i}(0).
\end{equation*}
The first term is our target estimand,
\begin{equation*}
  \tau :=\Pnn \frac{D_i}{\overline{\pi}}(Y_{i}(1) - Y_{i}(0)).
\end{equation*}
By definition, $\tau$ is the in-sample average effect on the treated, a natural target in many applications. Our theoretical results describe the behavior of the error  $\hat \tau - \tau$ in large samples. 

Our next assumption describes the features of the data-generating process (DGP) that restrict the underlying sampling and assignment processes. The first part of this assumption assumes that the potential outcomes, treatment indicators, and an unobserved unit-level characteristic $\eta_i$ are sampled randomly from some population. This restriction is typical in econometric panel data analysis, going back at least to \citet{chamberlain1984panel}, and is commonly made in the recent literature on the DiD estimators (e.g., \citealp{abadie2005semiparametric,callaway2021difference}). In light of this assumption, we will often drop the subscript $i$ when discussing the properties of a generic observation. The second part of the assumption allows for rich selection patterns based on unobserved heterogeneity $\eta$ and information on the past. This latent unconfoundedness assumption is commonly imposed in causal models for panel data (e.g., \citealp{arkhangelsky2022doubly}). The key difference between that setting and our setup is that $X$ contains information on past outcomes, which implies that it is not a strictly exogenous covariate.\footnote{See \cite{arellano2003panel} for a textbook discussion of strict exogeneity.} Finally, we also impose a weak overlap restriction on the treatment probabilities. This restriction implies that $\eta_i \ne D_i$. It is an identification assumption that guarantees that if $\eta_i$ were observed, it would have been possible to solve the selection problem by appropriately reweighting the control units. 
\begin{assumption}\label{as:sel}\textsc{(Sampling and Selection)}\\
    (a) unit-level outcomes $\left\{\qty(X_i(0), X_{i}(1), Y_{i}(0), Y_{i}(1), D_i, \eta_i)\right\}_{i=1}^n$ are i.i.d.; (b) $ D_i \independent Y_i(0)\Bigl|\, \eta_i, X_i$, and  $\pi_i := \E[D_i| \eta_i, X_i]$ belongs to $(0,1)$ with probability 1.
\end{assumption}

Without further restrictions, the second part of Assumption \ref{as:sel} has no empirical content.
%\footnote{One can test the overlap part of Assumption \ref{as:sel} by investigating the properties of $\E[D|X]$, but there are no other testable restrictions.} 
In particular, it is trivially satisfied by defining $\eta := Y(0)$, which is partially unobserved. To attach meaning to $\eta$, we need to connect it to observables, which we do with our next assumption. First, we define the conditional expectation of the outcome and the corresponding error:
\begin{equation}
    \mu := \E[Y(0)| \eta,X], \quad \epsilon := Y(0) - \mu.
\end{equation}
We use this to define the effective number of pre-treatment periods:
\begin{equation}
    \frac{1}{\Tef(\m{F})} := \min_{f \in\spn{\mathbf{1},\m{F}}}\E\left[(f(X) - \mu)^2\right],
\end{equation}
with the convention that $ \Tef(\m{F}) = \infty$ if $ \min_{f \in\spn{\mathbf{1},\m{F}}}\E\left[(f(X) - \mu)^2\right] = 0$. This quantity measures how useful the pre-treatment information in $\mathcal{F}$ is for predicting the relevant conditional mean. 
\begin{assumption}\label{as:emp_set}\textsc{(Identifiability)}\\
$\Tef(\mathcal{F}) \rightarrow \infty$ as $T_0$ increases to infinity.
\end{assumption}
This assumption guarantees that in the limit where $T_0$ is infinite, it is possible to recover $\mu$ using $\mathcal{F}$. The validity of this restriction depends on the underlying probability model that connects $Y,\,X$, $\eta$, and the set of features $\m{F}$. In the next section, we discuss three examples that illustrate the scope of this assumption. We will substantially generalize them in Section \ref{sec:exm}. 

\begin{remark}
    Weights that minimize (\ref{eq:pr_problem}) were proposed in \cite{hainmueller2012entropy} to produce a balanced sample. \citet{imbens1998information} and \citet{schennach2007point} analyzed a related exponential tiling estimator in the context of models defined by moment conditions, see also \cite{graham2012inverse}. What makes our setup special is that $X_i$ describes pre-treatment characteristics of unit $i$, particularly pre-treatment outcomes. For the case with a single treated unit, the problem \eqref{eq:pr_problem} with  $\zeta = 0$ corresponds to the SC method proposed in \cite{Abadie2010}. The same procedure, with $p= T_0$  and $\phi_{k}(X_i)$ equal to the levels of the pre-treatment outcomes, is often called the SC estimator (e.g., \citealp{doudchenko2016balancing,ben2018augmented,ferman2021synthetic}). 
\end{remark}
\begin{remark}
    In the literature on the SC method, there are various ways of dealing with multiple treated units, with \eqref{eq:pr_problem} being one option. Another prominent choice is to construct a synthetic control unit separately for each treated unit and aggregate afterward. See \cite{abadie2021penalized} for a discussion of this approach and \cite{ben2022synthetic} for a synthesis. As explained in \cite{cattaneo2022uncertainty}, the difference between the two approaches lies in how they measure the distance between the units. 
    %We choose to analyze the estimator \eqref{eq:pr_problem} because it connects naturally with other alternatives, such as the DiD and balancing.
    %For example, \cite{armstrong2018finite} show that for a particular functional class, $\mathcal{F}$ the weights that solve an analog of \eqref{eq:main_text_set} (with $l^2$ regularization instead of negative entropy) correspond to matching.
\end{remark}

\begin{remark}
    By definition we have a lower bound $\frac{1}{\Tef(\mathcal{F})} =   \min_{f \in\spn{\mathbf{1},\m{F}}}\E\left[(f(X) - \mu)^2\right] \ge \mathbb{E}[(\mu - \mathbb{E}[\mu|X])^2]$. The reciprocal of the latter quantity measures the optimal effective number of periods. In a typical panel data model, we expect $\mathbb{E}[(\mu - \mathbb{E}[\mu|X])^2] = O\qty(\frac{1}{T_0})$.\footnote{In the first example of Section \ref{subsec:exm}, if $(\epsilon_{1}, \dots, \epsilon_{T_0})$ and $\eta$ are jointly normal then $ \min_{f \in\spn{\mathbf{1},\m{F}}}\E\left[(f(X) - \mu)^2\right] = \mathbb{E}[(\mu - \mathbb{E}[\mu|X])^2]$.}  For a fixed $T_0$, even the best function of the past is insufficient to account for all aspects of selection. As a result, one cannot guarantee that the approximation error is arbitrarily small by expanding $\m{F}$. This contrasts our setup with the usual analysis under unconfoundedness, where one typically introduces enough conditions on $\m{F}$ for the approximation error to be negligible (e.g., \citealp{hirano2003efficient,wang2020minimal}).
\end{remark}

\subsection{Examples}\label{subsec:exm}
In the three examples below, we maintain Assumptions \ref{as:pot_outcomes} - \ref{as:sel} and interpret the observed pre-treatment variables $X$ as realizations of the underlying potential outcomes $X(0)$. We consider different types of $X$ and different models for $X(0)$. The main goal of these examples is to convince the reader that the concept of effective pre-treatment periods  $\Tef(\m{F})$ is useful.  Our first example describes a two-way model in which $\Tef(\m{F})$ behaves as $T_0$. We then demonstrate that this behavior dramatically deteriorates in the presence of unobserved policy shocks. Finally, we show that the initial behavior of $\Tef(\m{F})$ can be restored if additional information is available. 

\paragraph{Two-way model:} Suppose $X = (Y_{1},\dots, Y_{T_0})$ and $Y = Y_{T_0+1}$, i.e., the pre-treatment information we observe are outcomes in the pre-treatment periods. Also, suppose that $p = T_0$ and $\phi_t(X) = Y_{t}$.
In addition, suppose that the baseline potential outcomes $Y_{t}(0)$ follow a two-way model:
\begin{equation*}
    Y_{t}(0) = \eta + \lambda_t + \epsilon_{t}, \quad \mathbb{E}[\eta] = 0, \quad \mathbb{E}[\epsilon_{t}|\eta] = 0,
\end{equation*}
where $\lambda_t$ is a fixed constant. We assume that $\epsilon_{t}$ are uncorrelated and have equal variance $\sigma^2$. In this case $\mu = \eta + \lambda_{T_0+1}$ and we have:
\begin{equation*}
    \mathbb{E}\qty[\qty(\mu - c_0 - \sum_{t = 1}^{T_0}c_t Y_{t})^2] = 
    \qty(\lambda_{T_0+1} - c_0 - \sum_{t = 1}^{T_0}c_t \lambda_t)^2 + \mathbb{V}[\eta]\qty(1- \sum_{t = 1}^{T_0}c_t)^2  + \sigma^2\sum_{t=1}^{T_0}c_{t}^2. 
\end{equation*}
Minimizing the last expression over $(c_0, c_1,\dots, c_{T_0})$ we get
\begin{equation*}
      \frac{1}{\Tef(\mathcal{F})} = \min_{f \in\spn{\mathbf{1},\m{F}}}\E\left[(f(X) - \mu)^2\right] = \frac{ \sigma^2\mathbb{V}[\eta]}{\mathbb{V}[\eta]T_0 +  \sigma^2}. 
\end{equation*}
It follows that $\Tef(\mathcal{F}) \sim T_0 $ and Assumption \ref{as:emp_set} holds.  This behavior is natural: each period provides new information about $\eta$, and thus the effective number of periods is equal to $T_0$. 

\paragraph{Unobserved policy shock:} We continue assuming that $X = (Y_{1},\dots, Y_{T_0})$, and use the same set $\mathcal{F}$ as in the previous example. However, suppose now that the baseline outcomes follow a model with interactive fixed effects (e.g., \citealp{holtz1988estimating}):
\begin{equation}
    Y_{t}(0) = \eta^{(1)} + \eta^{(2)} \psi_t + \lambda_t + \epsilon_{t}, \quad \mathbb{E}[\eta] = 0, \quad \mathbb{E}[\epsilon_{t}|\eta] = 0,
\end{equation}
where $\eta := \qty(\eta^{(1)}, \eta^{(2)})$, and $\psi_t = 0$ for $t< T_0$ and $\psi_t = 1$ for $t\ge T_0$. We can interpret $\psi_t$ as an aggregate policy shock that affects the relevant outcomes, with $\eta^{(2)}$ measuring its heterogeneous impact on different units.

Assuming that the coordinates of $\eta$ are uncorrelated and making the same assumptions on $\epsilon_{t}$ as in the previous example, we have $\mu = \eta^{(1)} + \eta^{(2)} + \lambda_{T_0+1}$ which leads to the following bound:
\begin{multline*}
    \mathbb{E}\qty[\qty(\mu - c_0 - \sum_{t = 1}^{T_0}c_t Y_{t})^2] = 
    \qty(\lambda_{T_0+1} - c_0 - \sum_{t = 1}^{T_0}c_t \lambda_t)^2 + \mathbb{V}[\eta^{(1)}]\qty(1- \sum_{t = 1}^{T_0}c_t)^2  + \\
    \mathbb{V}[\eta^{(2)}]\qty(1- c_{T_0})^2 + \sigma^2\sum_{t=1}^{T_0}c_{t}^2 \ge \frac{ \sigma^2\mathbb{V}[\eta^{(2)}]}{\mathbb{V}[\eta^{(2)}] +  \sigma^2}. 
\end{multline*}
This implies that $\Tef(\m{F})\sim1$, and thus Assumption \ref{as:emp_set} is violated. Again, this behavior should not be surprising: in this example, we have a single pre-treatment period that provides information about the relevant unobserved heterogeneity.

\paragraph{Addressing policy shocks:} As a next example, suppose that $X = (Y_{1},Z_{1},\dots, Y_{T_0}, Z_{T_0})$ and the variables $Y_{t}(0)$ and $Z_{t}(0)$ evolve according to the following model:
\begin{align*}
    Y_{t}(0) &= \eta^{Y} + \lambda_t^{Y} + \psi_t\eta^{Z} + \epsilon^Y_{t},\\
    Z_{t}(0) & = \eta^{Z} + \lambda_t^{Z} + \epsilon^Z_{t},
\end{align*}
where $\psi_t$ is the same as in the previous example, and
\begin{align*}
    \mathbb{E}\qty[\qty(\eta^{Y},\eta^{Z})] = 0,\quad  \mathbb{V}\qty[\qty(\eta^{Y},\eta^{Z})] = \mathcal{I}_2, \quad \mathbb{E}\qty[\left(\epsilon^Y_{t}, \epsilon^Z_{t}\right)|\eta^{Y},\eta^{Z},Y_{t-1}(0),Z_{t-1}(0),\dots  ] = 0.
\end{align*}
This example is a generalization of the previous one because now we have access to a noisy measurement of $\eta^{Z}$. Suppose that $p = 2\times T_0$ and $\phi_t(X) = Y_{t}$, $\phi_{T_0 + t}(X) = Z_{t}$ for $t \le T_0$. Computation analogous to the one for the two-way model demonstrates that $\Tef(\m{F}) \sim T_0$. As a result, the access to the additional variable that captures relevant unobserved heterogeneity can restore Assumption \ref{as:emp_set} even in situations with unobserved policy shock. 

How natural is it to assume that researchers have access to a variable like $Z_{t}$? The model described above is quite specific and is unlikely to be directly applicable. At the same time, an alternative way of interpreting this structure is to see that $Y_{t}(0)$ and $Z_t(0)$ describe outcome variables that depend on the same unit-level heterogeneity $(\eta^{Y}, \eta^{Z})$ but react differently to aggregate shocks. From this perspective, the model in this section is a simple example of a more flexible setup that can be used in a large class of applications. In Section \ref{subsec:var}, we further develop this interpretation by considering a joint dynamic model for $Y_{t}(0)$ and $Z_{t}(0)$.
\\

The first two examples describe extreme cases for the behavior of $\Tef(\m{F})$, which ranges from $1$ to $T_0$. The behavior in the second example is akin to identification failure and may be considered too pessimistic. On the other hand, the behavior of $\Tef(\m{F})$ in the two-way model is quite optimistic because all information from the past is directly applicable. In Section \ref{sec:exm}, we show that $\Tef(\m{F})\sim T_0$ in models with interactive fixed effects, as long as the underlying factors $\psi_t$ are ``strong'', but is smaller in more complicated models.

\subsection{Numerical experiments}\label{sec:sim}
In this section, we discuss several Monte-Carlo experiments in which we compare the performance of the SC method versus the event study specification with two-way fixed effects (TWFE). The outcome models in these experiments are designed to give no prior advantage to the more complicated method. The goal of this exercise is to show that the SC method is a competitive alternative to the DiD, but its performance is not always perfect. Our theoretical results in Section \ref{sec:results} provide formal statistical guarantees that explain this performance.

For each $t \in \{1,\dots,  T_0 + K\}$, we assume that the baseline outcomes are described by a two-way model:
\begin{equation*}
Y_{i,t}(0) = \eta_i + \lambda_t + \epsilon^{(d)}_{i,t},
\end{equation*}
with shocks following either a stationary autoregressive process ($d= AR$) or a random walk process ($d = RW$). We also specify a treatment effect for $t> T_0$ as a function of time:
\begin{equation*}
    Y_{i,t}(1) - Y_{i,t}(0) = \tau(t - T_0 - 1),
\end{equation*}
so that the treatment effect is zero in the first treatment period $T_0+1$ and then grows linearly. We use this specification for presentation purposes; the estimators we consider are invariant with respect to any treatment effect specification. Appendix \ref{ap:pars} contains the values of all parameters and additional details.

We consider two different models for $\epsilon^{(d)}_{it}$. The first one is a stationary AR(1) model:
\begin{equation*}
    \epsilon_{i,t}^{AR} = \rho \epsilon^{AR}_{i,t-1} + u^{AR}_{i,t}.
\end{equation*}
The underlying selection process is based on past errors in periods $T_0$ and $T_0-1$ and unobserved heterogeneity:
\begin{equation*}
    \pi _i = \mathbb{E}\left[\frac{\exp\left(\eta_i + \beta^{AR}_{T_0}\epsilon^{AR}_{i,T_0} + \beta^{AR}_{T_0-1}\epsilon^{AR}_{i,T_0-1} + \nu_{i}\right)}{1+ \exp\left(\eta_i + \beta^{AR}_{T_0}\epsilon^{AR}_{i,T_0} + \beta^{AR}_{T_0-1}\epsilon^{AR}_{i,T_0-1} + \nu_{i}\right)}| \eta_i, Y_{i,T_0},\dots\right].
\end{equation*}
where $\nu_i$ is a random coefficient unrelated to all other variables.  With this model, we want to capture the underlying dynamics in the outcomes and their connection with the selection process.\footnote{We introduce $\nu_i$ to guarantee that the performance of the estimators is not driven by the linearity of $\log\qty(\frac{\pi_i}{1-\pi_i})$ in $\eta_i$ and past shocks.}

The second model for $\epsilon_{it}^{(d)}$ is a random walk:
\begin{equation*}
    \epsilon^{RW}_{it} = \epsilon^{RW}_{it-1} + u^{RW}_{it}.
\end{equation*}
The corresponding selection process based on the error in period $T_0$:
\begin{equation*}
       \pi _i = \mathbb{E}\left[\frac{\exp\left( \beta^{RW}_{T_0}\epsilon_{i,T_0}+ \nu_{i}\right)}{1+ \exp\left(\beta^{RW}_{T_0}\epsilon_{i,T_0} + \nu_{i}\right)}| \eta_i, Y_{i,T_0},\dots\right].
\end{equation*}
Since $\epsilon_{i,T_0}$ is not directly observed, the assignment model implicitly depends on $\eta_i$ and past outcomes. The distinguishing feature of this model is that the outcomes have a strong stochastic trend, and the best-performing units have a higher chance of adopting the treatment. As a result, we expect big differences between treated and control units.

We compare the performance of the SC method to the standard two-way fixed effects estimator because the latter dominates the empirical practice. In particular, we consider an event-study specification:
\begin{equation}\label{eq:event_study}
    Y_{i,t} = \eta_{i} + \lambda_t + \sum_{k \ne -1} \tau_{k} D_i \{t - T_0 - 1 = k\} + \varepsilon_{i,t},
\end{equation}
which we estimate by the ordinary least squares (OLS) with two-way fixed effects. As an alternative to $\hat \tau_{k}^{TWFE}$, we consider the SC method described in Section \ref{sec:estimator}. We use $(Y_{i1},\dots, Y_{i,T_0})$ as features to construct the weights and set $\zeta^2 = 1$. We then apply these weights separately for $K+1$ post-treatment outcomes to construct $\hat \tau_{k}^{SC}$. To mimic the event study plot, we also construct $\hat \tau_{k}^{SC}$ for $k <0$ by applying the SC weights to the pre-treatment outcomes.

Regardless of the choice of model for $\epsilon_{i,t}^{(d)}$, the underlying DGP has the form \eqref{eq:event_study}, with $\tau_{k} = 0$ for $k <0$. As a result, the problematic performance of $\hat\tau_{k}^{TFWE}$ that we observe in some cases should not be attributed to any misspecification errors, e.g., heterogeneity in treatment effects emphasized in the recent work on the DiD-based estimators (e.g., \citealp{de2020two,callaway2021difference,goodman2021difference, sun2021estimating,borusyak2021revisiting}). Instead, the strict exogeneity does not hold in the models we consider, i.e., the adoption decisions are correlated with the time-varying parts of the outcomes conditional on the permanent unobserved heterogeneity.

\subsubsection{Comparison}

\begin{figure}[t!]
    \centering
    \includegraphics[scale = 0.4]{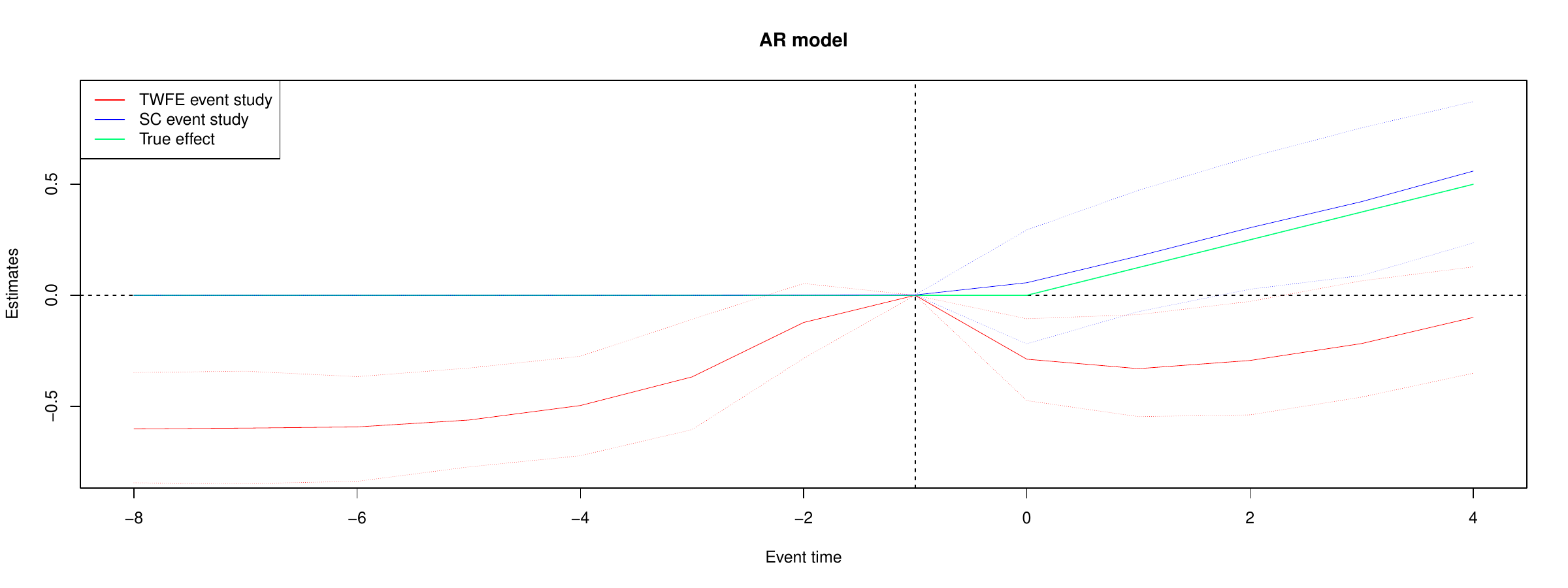}
    \caption{\footnotesize  AR design; the computation is based on $B = 200$ simulations, with each simulation having $n = 400$ units, $T_0 = 8$ pre-treatment periods, $K = 5$ treatment periods. The simulation parameters are reported in Appendix \ref{ap:pars}. The solid lines correspond to average results over the simulations. The dotted lines correspond to $5\%$ and $95\%$ quantiles of the distribution of the corresponding estimator in the simulations.}
    \label{fig:ar_model}
\end{figure}

\begin{figure}[t!]
    \centering
    \includegraphics[scale = 0.4]{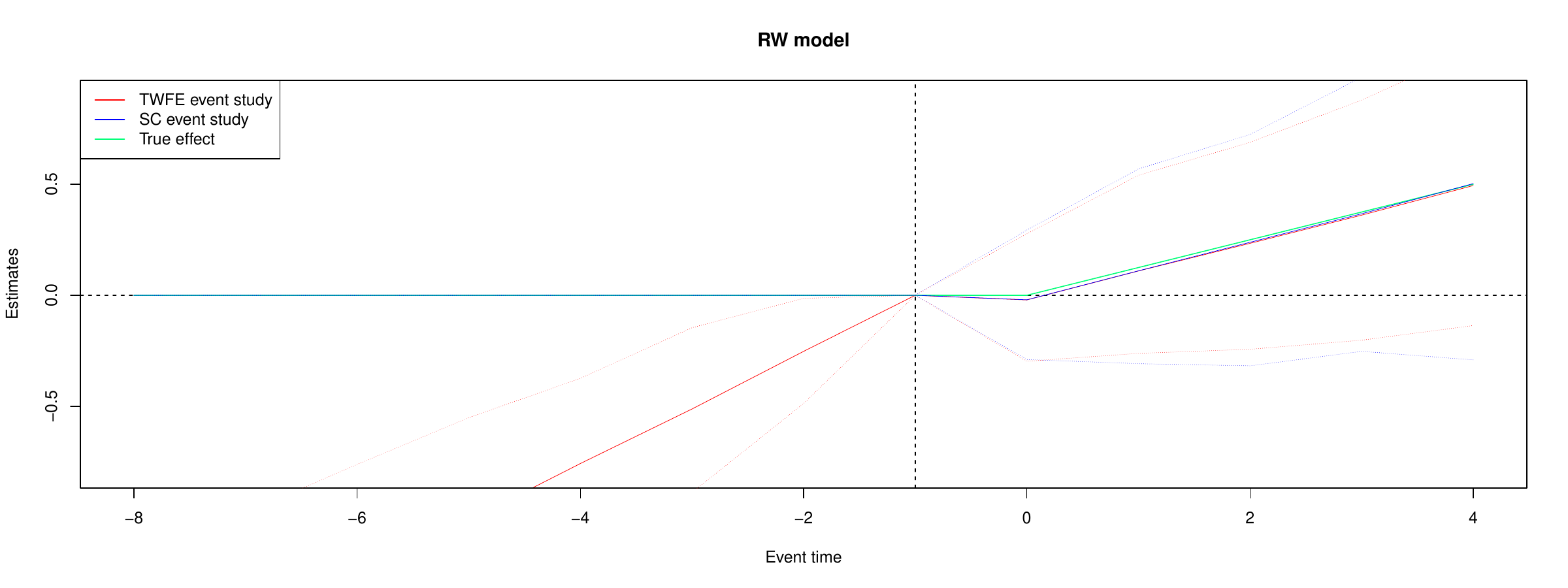}
    \caption{\footnotesize  RW design; the computation is based on $B = 200$ simulations, with each simulation having $n = 400$ units, $T_0 = 8$ pre-treatment periods, $K = 5$ treatment periods. The simulation parameters are reported in Appendix \ref{ap:pars}. The solid lines correspond to average results over the simulations. The dotted lines correspond to $5\%$ and $95\%$ quantiles of the distribution of the corresponding estimator in the simulations. }
    \label{fig:rw_model}
\end{figure}

We present the visual results in Figures \ref{fig:ar_model}-\ref{fig:rw_model}. Expectedly, the TWFE estimator is severely biased in the first simulation, with the bias being much larger than the estimation noise. At the same time, the SC estimator performs well in this simulation. The estimator is biased, which is in line with the theoretical results we present in Section \ref{sec:results}.  However, this bias is negligible compared to the noise, which itself is comparable to the noise of the TWFE estimator. In the second case, both estimators are nearly unbiased for the true treatment effect.  The lack of bias confirms the results in Section \ref{sec:results}: the conditional mean in the random walk model is a linear function of the past outcomes, and one of the relevant approximation errors is equal to zero. 

These two examples confirm our theoretical results and paint a positive picture for the SC estimator. It works well in environments where the DID estimator fails and remains competitive in environments where the DiD estimator is optimal. However, our theoretical results indicate that this behavior should depend on the ability of the set of features -- levels of past outcomes -- to approximate the conditional mean of the counterfactual outcome (Assumption \ref{as:emp_set}). We investigate this hypothesis by considering a straightforward generalization of the two models for which this assumption does not hold. In particular, we consider a simulation where $50\%$ of the observations are generated with the AR(1) model, and the rest are generated with the random walk model.   

\begin{figure}[t!]
    \centering
    \includegraphics[scale = 0.40]{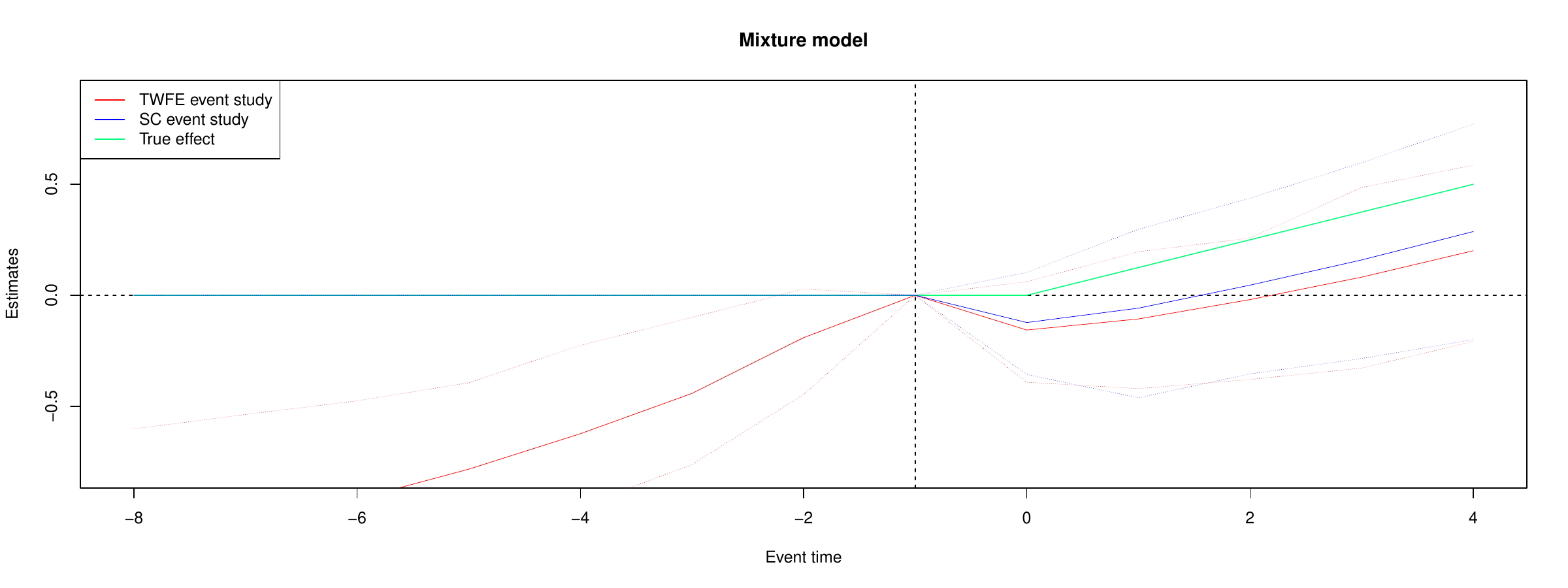}
    \caption{\footnotesize  Mixture design; the computation is based on $B = 200$ simulations, with each simulation having $n = 400$ units, $T_0 = 8$ pre-treatment periods, $K = 5$ treatment periods. The simulation parameters are reported in Appendix \ref{ap:pars}. The solid lines correspond to average results over the simulations. The dotted lines correspond to $5\%$ and $95\%$ quantiles of the distribution of the corresponding estimator in the simulations. }
    \label{fig:mix_model}
\end{figure}

The results for this simulation are visualized by Figure \ref{fig:mix_model}, and they are less favorable for the SC estimator. Its bias is comparable to that of the DiD estimator, and the estimator is more noisy. This failure might be surprising: a mixture of two two-way models remains a two-way model. The crucial difference between this simulation and the previous ones thus lies not in the specification of the levels but rather in the persistence of the time-varying shocks. Indeed, we can write down the model in the following form:
\begin{equation*}
    Y_{i,t}(0) = \eta_i + \lambda_t + \{\text{$i$ is from the  AR(1) model}\} \epsilon_{i,t}^{AR} + \{\text{$i$ is from the RW model}\} \epsilon_{i,t}^{RW}.
\end{equation*}
As a result, this model has an additional dimension of permanent heterogeneity. This dimension is relevant for the selection model and for the conditional mean of the counterfactual outcomes. The linear combination of past outcomes cannot distinguish the units that come from the AR(1) model from those that come from the random walk model, and Assumption \ref{as:emp_set} fails. 

Based on the results in this section and the formal results we derive in the following section, we recommend using the SC method in applications where the assignment process depends on the unobserved heterogeneity and past outcomes. The TWFE estimator is likely to fail in such applications, leaving applied researchers without a default option. SC method appears to be a natural alternative because, in contrast to conventional panel data estimators, it does not require users to specify a particular model. At the same time, the performance of the SC method is not always perfect, and our theoretical results in Section \ref{sec:results} outline the key reasons for its failure.

\subsubsection{Inference}

In general, inference based on SC estimators is challenging. In \cite{Abadie2010}, the authors focused on a permutation-based procedure, while subsequent research proposed alternative strategies (e.g., \citealp{chernozhukov2021exact, cattaneo2021prediction}). However, these challenges are mostly driven by the focus on applications with a single treated unit. For applications with many treated units \cite{arkhangelsky2021synthetic} show that conventional inference procedures based on unit-level bootstrap and $t$-statistic are asymptotically valid in models with a strictly exogenous assignment mechanism. The same results extend to our setup. 

In particular, we suggest that users conduct inference in two steps. First, they create bootstrap samples by randomly drawing $n$ units with replacements from the original dataset. In each bootstrap sample $s$ researchers  construct  $\hat \tau_{k}^{SC,s}$ and then compute the standard deviation of these estimates:
\begin{equation*}
    \hat \sigma\qty(\hat \tau_{k}^{SC,b}) := \sqrt{\frac{1}{S}\sum_{s=1}^S \qty (\hat \tau_{k}^{SC,s} - \frac{1}{S}\sum_{j=1}^S\hat \tau_{k}^{SC,j})^2}.
\end{equation*}
As a second step, researchers can use $  \hat \sigma\qty(\hat \tau_{k}^{SC,b})$ either to construct the $t$-statistic for a given null hypothesis
\begin{equation*}
    t_{stat} := \frac{\hat \tau_{k}^{SC} - \tau_{k} }{ \hat \sigma\qty(\hat \tau_{k}^{SC,b})},
\end{equation*}
and compare it to quantiles of the standard normal distribution or to construct a confidence interval (CI):
\begin{equation}\label{eq:conf_int}
    \tau_{k}  \in  \hat \tau_{k}^{SC} \pm q_{\frac{\alpha}{2}}  \hat \sigma\qty(\hat \tau_{k}^{SC,b}).
\end{equation}
If we use $\hat \tau_{k}^{TWFE}$ instead of $\hat \tau_{k}^{SC}$, then the described procedure corresponds to the conventional asymptotic inference for the TWFE estimator. We compare the inference procedures based on the two estimators using simulations.

\begin{figure}[t!]
    \centering
    \includegraphics[scale = 0.40]{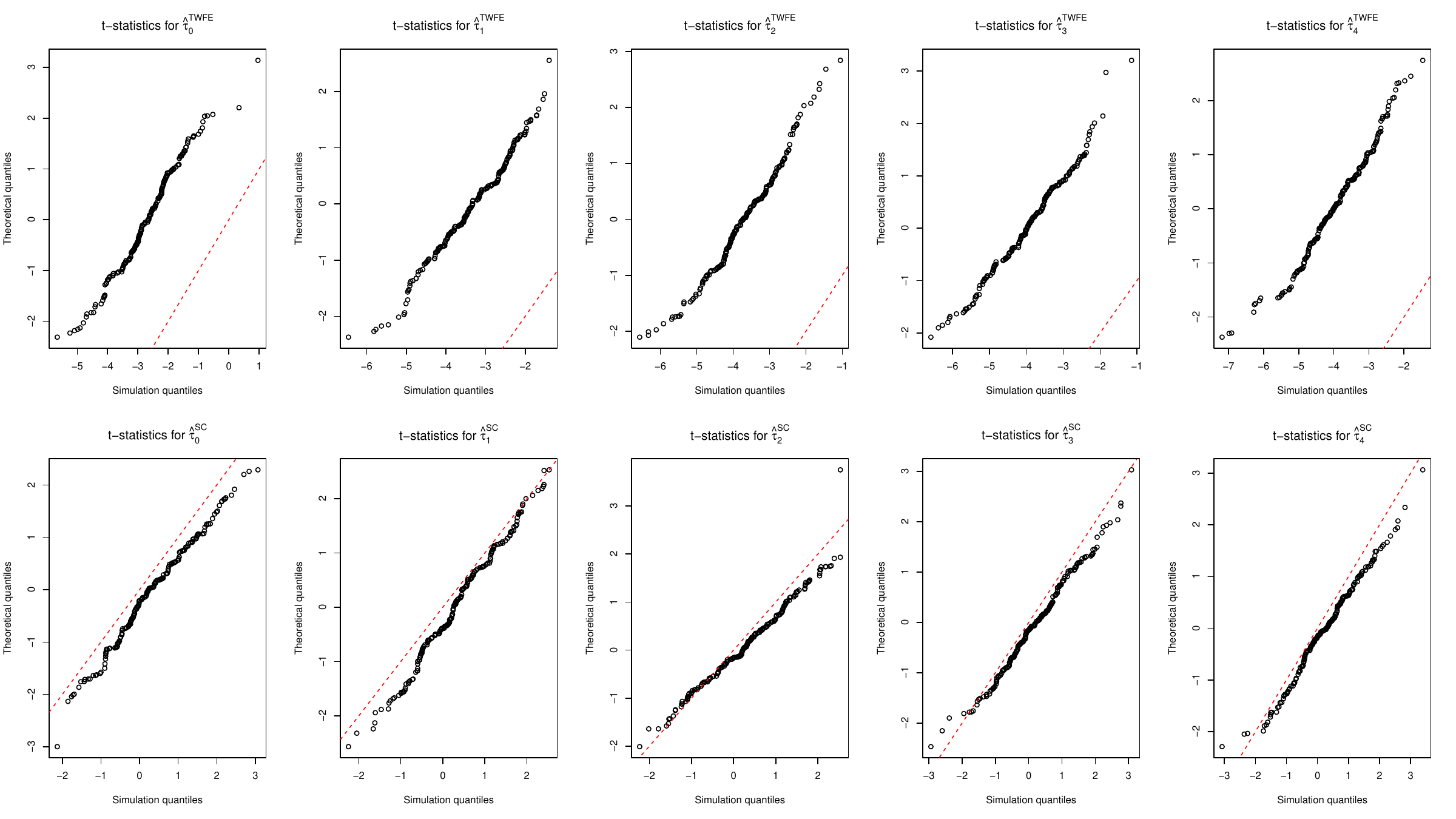}
    \caption{\footnotesize  Computations based on $B = 200$ simulations with AR design. Each simulation has $n = 400$ units, $T_0 = 8$ pre-treatment periods, $K = 5$ treatment periods. The simulation parameters are reported in Appendix \ref{ap:pars}. First row: QQ plots for $t$-statistics based on the TWFE estimator; second row: QQ plots for $t$-statistics based on the SC estimator. Variance for each estimator is computed using $100$ bootstrap samples. }
    \label{fig:inf_ar_model}
\end{figure}

Figure \ref{fig:inf_ar_model} describes the inference results for the TWFE estimator and SC estimator in AR design. Each graph corresponds to a quantile-quantile (QQ) plot that connects the distribution of the corresponding $t$-statistic in simulations with the quantiles of the standard normal distribution. As expected from the results in the previous section, $t$-statistic based on $\hat \tau_{k}^{TWFE}$ cannot be used for inference, with all its quantiles being shifted by the bias. Results are much more positive for the $\hat \tau_{k}^{SC}$. The distributions of the corresponding $t$-statistics are closely aligned with the standard normal distribution, albeit not perfectly, which is expected given a small bias evident from Figure \ref{fig:ar_model}. We report the coverage rates for the corresponding $95\%$ confidence intervals (CI) in Table \ref{table:coverage}. The results tell the same story as Figure \ref{fig:inf_ar_model}, with coverage for CI based on $\hat\tau_{k}^{TWFE}$ being below $20\%$ because of the bias, and the coverage for CI based on $\hat\tau_{k}^{SC}$ being close to its nominal $95\%$ level. The results for other designs (RW and mixture) confirm the visual evidence from Figures \ref{fig:rw_model} - \ref{fig:mix_model} and we report them in Appendix \ref{ap:pars}.

\begin{table}[t!]
\centering
\setlength{\extrarowheight}{0.2cm}
\begin{tabular}{|l|r|r|r|r|r|}
\hline\hline
\multirow{2}{*}{} & \multicolumn{5}{c|}{Effect in treatment period}\\\cline{2-6}
& 0 & 1 & 2 & 3 & 4 \\
\hline\hline
CI based on $\hat \tau_{k}^{TWFE}$ & 0.19 & 0.04 & 0.03 & 0.01 & 0.01 \\
CI based on $\hat \tau_{k}^{SC}$ & 0.93 & 0.95 & 0.93 & 0.93 & 0.94 \\
\hline\hline
\end{tabular}
\caption{\footnotesize  Coverage rates for $95\%$ confidence intervals based on $B = 200$ simulations with AR design. Each simulation has $n = 400$ units, $T_0 = 8$ pre-treatment periods, and $5$ treatment periods. The simulation parameters are reported in Appendix \ref{ap:pars}. First row: coverage rates based on $\hat \tau_{k}^{TWFE}$; second row: coverage rate based on  $\hat \tau_{k}^{SC}$. Confidence intervals are constructed using \eqref{eq:conf_int}.}
\label{table:coverage}
\end{table}

\subsubsection{Validation and placebo analysis}

One of the distinguishing features of the TWFE analysis is its reliance on the pretrends to test the underlying model. This practice has many potential problems (see \citealp{roth2022pretest,rambachan2023more}) but remains common in applied work. The results presented in Figure \ref{fig:rw_model} demonstrate that the lack of pretrends is not necessary for the event-study estimator to be unbiased, but the underlying DGP is quite specific. Figures \ref{fig:ar_model} and \ref{fig:mix_model} demonstrate that pretrends provide useful information in other cases. The situation with the SC estimator is different; by construction, the pre-treatment outcomes are balanced, and we do not observe any pretrends in Figures \ref{fig:ar_model} -- \ref{fig:mix_model}. 

\begin{figure}[t!]
    \centering
    \includegraphics[scale = 0.40]{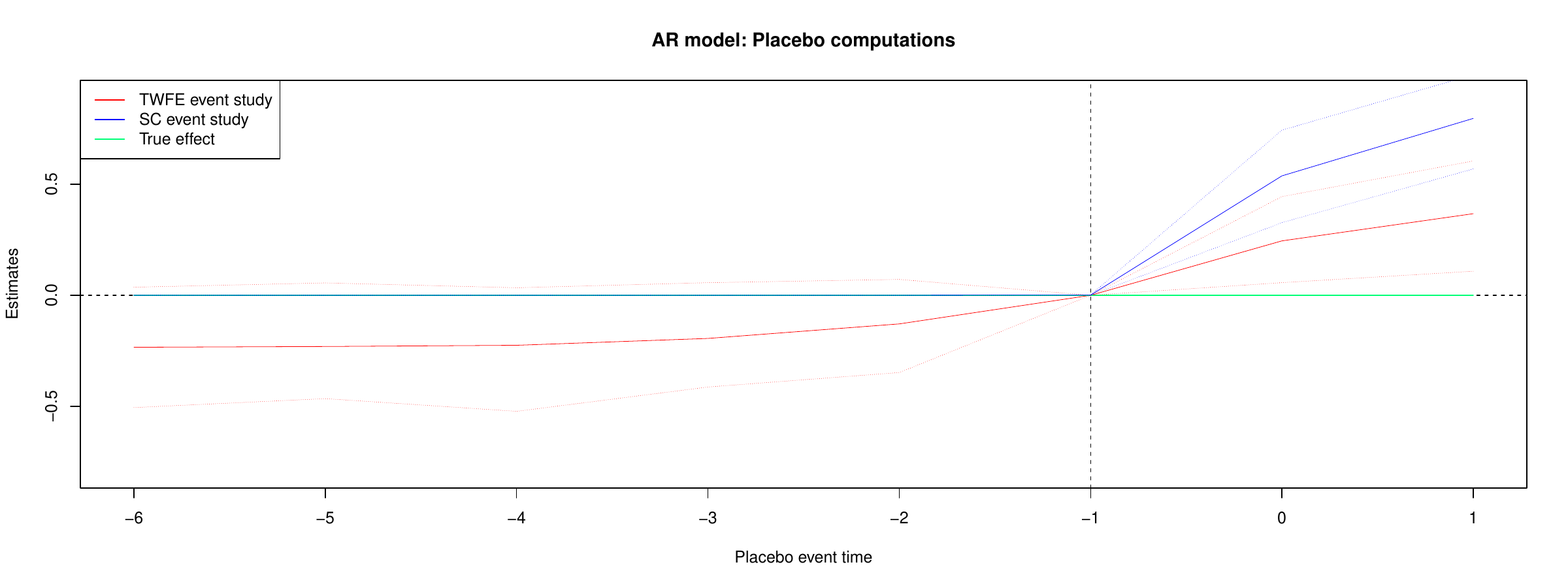}
    \caption{\footnotesize  The computation is based on $B = 200$ simulations, with each simulation having $n = 400$, $T_0 = 8$, $K = 5$. The simulation parameters are reported in Appendix \ref{ap:pars}. The solid lines correspond to average results over the simulations. The dotted lines correspond to $5\%$ and $95\%$ quantiles of the distribution of the corresponding estimator over the simulations. }
    \label{fig:plac}
\end{figure}

A different way of validating the model is to use a placebo analysis. In the context of the SC method, placebo evaluations take two forms that play distinct roles. First, as suggested in \cite{Abadie2010}, one can assign placebo treatments to control units, construct new estimates, and use them to quantify the uncertainty in the original estimator.  As we discussed in the previous section, this is a common inference technique in SC applications, but in the environments that we consider, one can use more conventional inference methods based on units-level bootstrap.
%\dmitry{Reference inference section} 

Alternatively, one can conduct a different placebo analysis by shifting the adoption period and checking if the SC estimator is close to zero for the placebo treatment periods. For the TWFE estimator, this exercise is equivalent to the standard analysis of the pretrends, but for the SC estimator, it delivers new information.  If the SC estimates are far away from zero in the placebo periods, one can interpret this as a failure of the underlying assumptions.  Unfortunately, such placebo exercises can deliver misleading results in the environments that we consider. 

To demonstrate this, we shift the (placebo) adoption time to period $T_0 - 1$ and use data from periods $1$ to $T_0-2$ to construct the SC weights and periods $T_0 -1$ and $T_0$ to conduct the placebo analysis. We present the results of this exercise for AR(1) simulation in Figure \ref{fig:plac}. Both estimators fail the placebo evaluation, delivering positive effects in the placebo periods in at least $95\%$ of simulations. This behavior is encouraging for the TWFE estimator, which is biased. However, the SC estimator's performance might appear puzzling, given the positive results presented in Figure \ref{fig:ar_model}. The failure is due to selection: the treated units are partly selected on the value of time-varying shocks in periods $T_0$ and $T_0-1$. As a result, these units tend to have larger outcomes in these periods, leading to biased estimates. Formally, Assumption \ref{as:sel} is not satisfied in this simulation. We do not recommend using such placebo evaluations for the validation of the SC estimator whenever one suspects this type of selection.

The results from Figure \ref{fig:plac} demonstrate that placebo exercises based on shifting the adoption periods can be misleading. This does not mean, though, that it is impossible to validate the performance of the SC estimator. Below, we describe one possible option, leaving the full theoretical analysis to future research. In particular, for each unit $i$, we compute the first-order autocorrelation coefficient using $T_0$ pre-treatment periods, which we denote by $\hat \rho_i$. We then calculate the difference between these coefficients among treated and control units using two different weights: uniform weights and the SC weights; we call the resulting coefficient $\hat \rho^{(k)}$, where $k   \in \{\text{Uniform; SC}\}$. We conducted this exercise for the three DGPs described before (AR, RW, and mixture) and for three different designs. The first design has the same number of units and periods as all our previous computations, and with the other two, we gradually increase the number of units and periods. To make comparisons across DGPs meaningful, in each case, we normalize $\hat \rho^{(k)}$ by its standard deviation (over simulations).

\begin{table}[t!]
    \centering
    \begin{tabular}{| m{4em} |c|c|c|c|c|c|}\hline
     & \multicolumn{2}{|c|}{Design 1} & \multicolumn{2}{|c|}{Design 2} &  \multicolumn{2}{|c|}{Design 3}  \\
     \hline
        DGP & Uniform & SC & Uniform & SC & Uniform & SC \\
        \hline\hline
       \multirow{2}{2em}{AR} & 0.05  & 0.10  &  0.01  & 0.02 & -0.06 & 0.08\\
       &  [-1.51, 1.72] & [-1.36, 1.79]  &[-1.77, 1.62] &  [-1.55, 1.47]  &[-1.74, 1.49] &  [-1.59,  1.65] \\
       \hline\hline
         \multirow{2}{2em}{Mixture} & 0.03 & -0.25 & -0.04 & -1.17 & 0.04 & -1.96 \\
         &  [-1.49, 1.65] & [-1.84, 1.32]  &  [-1.82, 1.44] & [-2.87, 0.37] & [-1.62, 1.63]  & [-3.65, -0.49]\\
         \hline\hline
         \multirow{2}{2em}{RW} & -0.14  & -0.06  & -0.06   &-0.06  & -0.01 & 0.07 \\
         & [-1.68, 1.61] &[-1.62, 1.89] &[-1.59, 1.61]  &[-1.61, 1.62] & [-1.78, 1.55] & [-1.54,  1.88]\\
    \hline\hline
    \end{tabular}
    \caption{\footnotesize Each sell reports the mean of $\hat \rho^{(k)}$ over $200$ simulations for $k \in \{\text{Uniform; SC}\}$. $5\%$ and $95\%$ quantiles (over simulations) are reported in the parenthesis. AR corresponds to AR(1) DGP, RW -- to the random walk DGP, and Mixture to the mixture DGP. Design 1 has $n = 400$ and $T_0 = 8$; Design 2 has $n = 3200$ and $T_0 = 12$; Design 3 has $n = 6000$ and $T_0 = 24$. }
    \label{tab:testing}
\end{table}

The results of this exercise are presented in Table \ref{tab:testing}. We can see that in all cases, the average (over stimulations) difference in the autocorrelation coefficients with uniform weights is close to zero, and the corresponding quantiles show that the distribution of $\hat \rho^{uniform}$ is dominated by noise. The situation is similar for the SC-based estimator for all designs and DGPs except the mixture one. In the latter case, we can already see a weak signal for the first design with an average of $-0.25$, which is dominated by the noise. The signal gets much stronger for the second design, with the average of $-1.17$ and $95\%$ quantile being closer to zero. Finally, the signal is overwhelming for the third design, with its $95\%$ quantile being equal to $-0.49$. A researcher who would have used this procedure to test the validity of the SC estimator would have rightly concluded that it has problems in the case of the mixture DGP.

\begin{remark}
    The failure of parallel trends for the AR(1) model is expected in light of the theoretical results in \cite{ghanem2022selection} for the DiD estimator. As shown in  \cite{ghanem2022selection}, the random walk is essentially a unique model under which the DiD estimator is consistent, even if the assignment process is not strictly exogenous, which is the reason for the success of the TWFE estimator in the RW simulation.
\end{remark}

\section{Theoretical results}\label{sec:results}
This section presents our abstract theoretical results for two asymptotic regimes. The first regime assumes that the share of treated units is asymptotically vanishing and is close in spirit to the traditional analysis for the SC method. In the second regime, the share of treated units is constant, which is more natural for economic applications where researchers currently use the DiD estimator. We show that the asymptotic behavior of the estimator differs across the two scenarios. To streamline the presentation, we describe our results first and then discuss the technical statistical assumptions we need to prove them in addition to those described in the previous section.

To state our results, we introduce additional notation. First, we define the log-odds, $\lo:= \log\qty(\frac{\pi}{1-\pi})$, and a loss function $\ell(x) := \exp(-x) + x -1$.\footnote{By definition $\ell(0) = \ell^{\prime}(0) = 0$ and $\ell^{\prime\prime}(0) = 1$, as a result for $x$ close to zero we have $\ell(x) \approx \frac{x^2}{2}$.} We also associate any function $f\in \mathcal{F}$ with the corresponding random variable $f(X)$. We use these objects to define the bias term:\footnote{Here $\| \cdot \|_{\m{F}}$ is the gauge of $\m{F}$ extended to $\spn{ \mathbf{1},\m{F},\mu}$, see Appendix \ref{ap:defs} for the formal definition.}
\begin{equation}\label{eq:bias}
\begin{aligned}
\bias
&:= \Pnn\frac{\pi_i \exp(\muplo_i -\lo_i)}{\E[\pi]}(\plo_i - \muplo_i)(\mu_i-\pmu_i), \qqtext { where } \\
\plo &:=  \argmin_{f \in \spn{ \mathbf{1},\m{F}}} \E\left[\ell(\lo - f)| D = 1\right] + \frac{\zeta^2}{2n}\|f\|^2_{\m{F}}, \\ 
\muplo &:=\argmin_{f \in \spn{ \mathbf{1},\m{F},\mu}}  \E\left[\ell(\lo - f)| D = 1\right] +\frac{\zeta^2}{2n}\|f\|^2_{\m{F}}, \\
\pmu &:= \argmin_{f\in \spn{\mathbf{1},\m{F}}} \E\qty[\exp(\muplo -\lo)(\mu - f)^2|D=1].
\end{aligned}
\end{equation}
In general, $\bias\ne 0$ because we cannot control perfectly for the unobservables. It has a familiar product structure typical for estimators of average effects, with one part coming from the outcome model and the second coming from the assignment model. In particular, the outcome part of the bias quantifies the difference between the conditional mean $\mu$ and its weighted projection $\pmu$. Assumption \ref{as:emp_set} suggests that this error behaves as $\frac{1}{\sqrt{\Tef}}$, which is indeed the case under additional technical assumptions we describe in the next section.\footnote{For brevity we suppress the dependence of $\Tef$ on $\m{F}$ whenever it does not cause confusion.} 

The assignment part of the bias is different and describes the discrepancy between two different projections of the log-odds $\lo$. The first projection, $\plo$, uses functions in $\m{F}$ to predict $\lo$. The second projection, $\muplo$, also uses $\mu$ to predict the log-odds. Neither $\plo$ nor $\muplo$ are assumed to be close to true log-odds $\lo$ even as $T_0$ gets large. Similarly to the first error, we show that the difference $\plo-\muplo$ behaves as $\frac{1}{\sqrt{\Tef}}$ implying that the overall bias term behaves as $\frac{1}{\Tef}$.

In our discussions below, we routinely use the fact that $\bias = O_p\qty(\frac{1}{\Tef})$. This is a worst-case bound, which we guarantee under weak assumptions. In practice, the bias can be smaller for several reasons. For example, the error $\plo_i - \muplo_i$ can be small because $\mu$ is not particularly relevant for predicting log-odds (in addition to functions in $\mt{F}$). Also, it might be the case that the two errors $\plo_i - \muplo_i$ and $\mu_i-\pmu_i$ are uncorrelated, making $\bias$ negligible.

The theoretical results we present next show that $\bias$ indeed describes the asymptotic bias of the SC method in both asymptotic regimes. This shows that $\Tef$ is a crucial parameter for the performance of the SC method. Depending on how large $\Tef$ is compared to the sample size, the SC method is either asymptotically unbiased and thus can be used for inference or is dominated by the bias. In Section \ref{sec:exm}, we show how $\Tef$ depends on the complexity of the underlying model.

\subsection{Vanishing treatment share}
Our first result describes the behavior of the SC method in the regime where the share of treated units is small.
\begin{theorem}\label{th:van_share}
    Suppose $\mathcal{F}$ is given by \eqref{eq:main_text_set},  Assumptions \ref{as:pot_outcomes} - \ref{as:emp_set}, and \ref{as:subg} - \ref{as:out_mom} hold, $1 \gg \E[\pi]  \gg  \frac{1}{\sqrt{n}}$, and $\zeta = O(1)$. Then we have:
    \begin{equation*}
          \hat \tau - \tau = \bias  + \Pnn \frac{D_i - \pi_i}{1-\pi_i}\frac{\epsilon_i}{\E[\pi]}+ o_p\left(\frac{1}{\sqrt{\E[\pi]n}}\right) + o_p\left(\frac{1}{\Tef}\right),
          \qqtext{ with } \bias = O_p\left(\frac{1}{\Tef}\right).
    \end{equation*}
\end{theorem}
This result describes the behavior of the SC estimator in settings where the share of the treated units is vanishingly small. It is close in spirit to the results on the SC method that have a single treated unit (e.g., \citealp{Abadie2010, ferman2021synthetic,chernozhukov2021exact}). However, the lower bound on $\E[\pi]$ implies that the total number of treated units is increasing to infinity, similarly to  \cite{arkhangelsky2021synthetic}. We expect this result to be useful in applications where the share of treated units is very small, e.g., $1\%$ of the sample is treated. 

Theorem \ref{th:van_share} shows that the estimation error has two dominant terms. The first term is the bias discussed above, which, under the assumptions of the theorem, behaves as $O_p\qty(\frac{1}{\Tef})$. The second term of the error is the noise term, which also has a product structure with errors coming from randomness in the treatment assignment and unpredictable noise in the outcomes. The standard deviation of this term is on the order of $\frac{1}{\sqrt{\E[\pi]n}}$. This behavior has a straightforward implication for inference. In particular, as long as $\Tef$ is larger than $\sqrt{\E[\pi]n}$, the SC estimator is dominated by the noise, guaranteeing its asymptotic normality. Formally, we have the following corollary.
\begin{corollary}\label{cor:as_norm_van}
    Suppose conditions of Theorem \ref{th:van_share} hold, and $\Tef^2 \gg \E[\pi]n$. Then we have
    \begin{equation*}
        \hat \tau - \tau = \Pnn \frac{D_i - \pi_i}{1-\pi_i}\frac{\epsilon_i}{\E[\pi]}+ o_p\left(\frac{1}{\sqrt{\E[\pi]n}}\right),
    \end{equation*}
    and $\sqrt{\overline \pi n}(\hat \tau - \tau)\rightharpoonup_{d} \mathcal{N}(0, \sigma^2_{van})$, where $\sigma^2_{van} =     \E\left[\epsilon_{i}^2 | D= 1\right]$.
\end{corollary}
This result implies that the estimator is asymptotically linear and unbiased. Standard tools, such as unit-level bootstrap, can be used to estimate the variance and conduct asymptotically valid inferences. In the extreme case where $\E[\pi]$ approaches $\frac{1}{\sqrt{n}}$, Corollary \ref{cor:as_norm_van} requires $\Tef^2 \gg \sqrt{n}$. It implies that the SC estimator can be asymptotically unbiased even when $\Tef$ is essentially of the order of $n^{\frac14}$ as long as the share of treated units is minimal.  This justifies using this estimator for inference in applications where $\Tef$ is not very large and $\E[\pi]$ is very small.

%To understand the asymptotic variance, it is useful to compare $\hat \tau$ to an infeasible estimator which uses the inverse-probability weights:
%\begin{equation*}
%    \hat \tau_{inf} = \Pnn \left(\frac{D_i}{\overline{\pi}} - \frac{(1-D_i)}{1-\pi_i}\frac{\pi_i}{\overline{\pi}}\right)Y_{i} = \tau + \Pnn \frac{D_i -\pi_i}{1-\pi_i}\frac{Y_{i}(0)}{\overline \pi}.
%\end{equation*}
 %One can show that $\sqrt{\overline \pi n}(\hat \tau_{inf} - \tau)$ converges to a centered normal distribution with a variance equal to $\E\left[\frac{Y^2(0)}{1-\pi}| D= 1\right] = \E\left[\frac{\mu^2}{1-\pi}| D= 1\right] + \sigma^2_{van}$. The fact that the asymptotic variance of $\hat \tau$ is smaller is expected since the estimators based on inverse probability weights are inefficient in general. We achieve this performance despite the model for log-odds being globally misspecified. 

\subsection{Non-vanishing treatment share}

Our next result describes the asymptotic behavior of the estimator in the regime where the share of treated units remains constant. To state this result, we introduce an additional error term:
\begin{equation}\label{eq:u_def}
    u:= \exp(\muplo - \lo) - 1,
\end{equation}
which is a population residual in the optimization problem that defines $\muplo$. Our next result shows that it affects the asymptotic variance of the SC method. 
\begin{theorem}\label{th:const_share}
     Suppose $\mathcal{F}$ is given by \eqref{eq:main_text_set}, Assumptions \ref{as:pot_outcomes} - \ref{as:emp_set}, and \ref{as:subg} - \ref{as:out_mom} hold; $\zeta = O(1)$ and $1 \gtrsim \E[\pi]  \gg  \frac{1}{\sqrt{n}}$. Then, 
\begin{equation*}
\begin{aligned}
        \hat \tau - \tau &= \bias +
        \Pnn \frac{D_i - \pi_i}{1-\pi_i}\frac{(\pi_iu_i + 1)}{\E[\pi]}\epsilon_i+\Pnn\frac{\pi_iu_i}{\E[\pi]} \epsilon_i +
        o_p\left(\frac{1}{\Tef}\right) + o_p\left(\frac{1}{\sqrt{\E[\pi]n}}\right), \\
        &\qqtext{ with } \bias = O_p\left(\frac{1}{\Tef}\right).
\end{aligned}
\end{equation*}
\end{theorem}
This result shows that the behavior of the SC method is more complicated in situations where the share of treated units is constant. While the bias term remains the same, the noise part has two additional terms proportional to $u_i$. These terms are negligible in environments with a vanishing treatment share, which is a manifestation of the built-in ``undersmoothing'' and, for this reason, do not appear in Theorem \ref{th:van_share}. Such behavior is typical from the point of the SC literature, where the error from the single treated unit dominates the estimation error.

In contrast, when the share of the treated units is constant, we need to consider the noise from all observations, and Theorem \ref{th:const_share}  captures that. Importantly, the ``price'' for having more treated units does not come in terms of the increased bias but rather in terms of additional noise terms. In contrast to Theorem \ref{th:van_share}, these noise terms do not depend on the ``design-based'' error $D_i - \pi_i$, and thus capture a different type of uncertainty.  They appear because of the misspecification error $u_i$ that quantifies the error between $\lo_i$ and $\muplo_i$. As a result, from a design-based perspective that only considers randomness from randomization (e.g., \citealp{abadie2020sampling, rambachan2020design}), these terms are part of the bias, which is negligible in the vanishing share regime. However, from the perspective of sampling-based uncertainty, these terms are part of the noise. We are not aware of other asymptotic results in a similar context that reflect the two types of uncertainty.\footnote{In \cite{abadie2020sampling} the authors discuss both sampling and design-based uncertainty, but there different perspectives matter for the size of the noise and do not affect the interpretation of the bias.}

Similar to Corollary \ref{cor:as_norm_van} in the situation where $\mathbb{E}[\pi] \sim 1$ and $\Tef^2 \gg n$, Theorem \ref{th:const_share} implies asymptotic linearity and unbiasedness of the SC method. As a result, in that regime, one can estimate the variance using unit-level bootstrap and use the conventional confidence intervals to conduct asymptotically valid inference. 

\begin{remark}
    The restriction $\Tef^2 \gg n$ should be familiar from the literature on large factor models. Theorem 1 in \cite{bai2003inferential} requires $T_0$ to guarantee asymptotic normality of estimated factors for each $t$. As we explain in Section \ref{sec:exm} in models considered by \cite{bai2003inferential} $\Tef \sim T_0$, and thus this is the same restriction.  This connection is expected: while the SC method does not explicitly estimate the parameters of the underlying statistical model, this happens implicitly through the construction of weights. The restriction on $T_0$ is relatively mild and allows the cross-sectional dimension of the problem to be much larger than the time-series one. This arrangement is natural in applications where researchers currently use  DiD-type strategies.
\end{remark}

\subsection{Statistical assumptions}
Our first assumption describes the statistical behavior of the subspace of random variables we use as inputs for the SC method and the population objects $\mu$ and $\lo$. 
\begin{assumption}\label{as:subg}\textsc{(Sub-gaussian class)}\\
    (a) $\spn{\m{F},\mu,\lo} \subseteq L^2$; (b) there exists an absolute constant $L_{\psi_2} <\infty$ such that for any $f \in \spn{\m{F},\mu,\lo} $ we have $\|f\|_{\psi_2} \le L_{\psi_2} \|f\|_2$.
\end{assumption}
This assumption guarantees that the $\spn{\m{F},\mu,\lo} $ is a sub-gaussian class. Such classes are well-understood objects in learning theory (e.g., \citealp{lecue2013learning}) and cover a wide variety of empirical problems. Moreover, the restriction to distributions with relatively light tails is almost necessary for our analysis. As we explain in Appendix \ref{ap:disc}, the search for the SC weights is analogous to the search over $\exp(f)$, where $f \in \spn{\mathbf{1},\m{F}}$. For this problem to be well-behaved, one has to assume the existence of exponential moments, making Assumption \ref{as:subg} particularly convenient. 

Our next assumption restricts the degree of misspecification in log-odds, particularly its asymptotic behavior. To introduce it, we consider a decomposition of $\muplo$ into two parts:
\begin{equation*}
     \muplo = f_{\muplo} + \beta_{\mu} \mu,
\end{equation*}
where $f_{\muplo} \in \spn{ \mathbf{1},\m{F}}$. This decomposition is unique as long as $\mu \not \in \spn{ \mathbf{1},\m{F}}$ and if $\mu \in\spn{ \mathbf{1},\m{F}}$ then we set $\beta_{\mu} = 0$.
\begin{assumption}\label{as:miss_lo}\textsc{(Degree of Misspecification)}\\
    (a) There exists an absolute constant $L_{\ell}>0$ such that  $\E\left[\ell(\lo - \muplo)| D = 1\right] \le L_{\ell}$; (b) there exists an absolute constant $L_{\mu}>0$ such that $|\beta_{\mu}| \le L_{\mu}$.
\end{assumption}
The first part of Assumption \ref{as:miss_lo} requires that the average difference between $\lo$ and $\muplo$ measured using the loss function $\ell(\cdot)$ does not become unbounded. We interpret this restriction as a uniform bound on the misspecification error, which still allows for global misspecification. It trivially holds if $\E\left[\ell(\lo - \E[\lo])| D = 1\right]$ is bounded, which is a minor integrability requirement for the models where $\|\lo-\E[\lo]\|_2$ does not change with $n$ and $T_0$. However, in the regime of Theorem \ref{th:van_share}  $\|\lo-\E[\lo]\|_2$ can increase and Assumption \ref{as:miss_lo} guarantees that there is always a function in $ \spn{\mathbf{1},\mu, \m{F}}$ that is close.

In contrast, the second part of Assumption \ref{as:miss_lo} describes local behavior. It restricts the population partial regression coefficient of $\muplo$ on $\mu$. It is restrictive, because Assumption \ref{as:emp_set} guarantees that as $T_0$ increases the distance between $\spn{ \mathbf{1},\m{F}}$ and $\mu$ decreases. In particular, the variation in $\mu$ that cannot be explained by $\spn{ \mathbf{1},\m{F}}$ vanishes. As a result, the population coefficient in the regression of $\muplo$ on this residual variation can become unbounded, and Assumption \ref{as:miss_lo} does not allow that.  

%A model where this assumption does not hold might look like this:
%\begin{equation*}
%    \muplo = f_{\muplo} + \tilde\beta_{\mu}\left(\frac{\mu - \bestmu}{ \|\mu -\bestmu \|_2}\right),
%\end{equation*}
%where $\tilde\beta_{\mu}$ is bounded, but $\beta_{\mu} = \frac{\tilde\beta_{\mu}}{\|\mu -\bestmu \|_2}$ is not. In particular, it is not $\mu$ that is important for the selection, but rather its normalized version $\left(\frac{\mu - \bestmu}{ \|\mu -\bestmu \|_2}\right)$. We believe that in the majority of applications, such behavior is unreasonable. 

%Next, we restrict the behavior of the distribution of the treatment probability $\pi$ in the population. Typical restrictions on $\pi$ include a strict overlap assumption, i.e., the existence of constants $\pi_0$ and $\pi_1$ such that $0 < \pi_0 \le \pi \le \pi_1 <1$ (e.g., \citealp{hirano2003efficient}). This assumption could have significantly simplified our analysis. However, we do not make it for two reasons. First, we want our analysis to be relevant for environments where the number of treated units is much smaller than the number of control units. This situation is common in applications, and the strict overlap together with Assumption \ref{as:sampling} guarantees that in the limit, the number of treated units is of the same order as the number of control units.  Second, strict overlap becomes increasingly more restrictive as the dimension of the problem increases (e.g., \citealp{d2021overlap}).

We restrict the tail behavior of $\pi$ with the following assumption.
\begin{assumption}\label{as:overlap} \textsc{(Degree of overlap)}\\
 (a) There exist an absolute constant $\epsilon_0>0$  such that for any $\epsilon\in(0,\epsilon_0]$ there exists an absolute constant $q(\epsilon) >0$ and $\E\left\{ \frac{\pi}{\E[\pi]} \ge q_{\epsilon}\right\} \ge 1- \epsilon$; (b) there exists an absolute constant $\pi_{\max}>0$ such that for any $\lambda \in [1,10]$ we have $\E[\exp(\lambda(\lo - \log(\E[\pi])))] \le \pi_{\max}$; (c) $\|\lo\|_2 = o(\sqrt{n})$.
\end{assumption}
The first part of this assumption restricts the left tail of the distribution of $\frac{\pi}{\E[\pi]}$. It prohibits $\frac{\pi}{\E[\pi]}$ from having a non-negligible mass at zero, even asymptotically. It is a very mild restriction, and we expect it to be satisfied in most applications where $\pi >0$ (as already required by Assumption \ref{as:sel}). To understand the other part of this assumption, observe that by definition, we have:
\begin{equation*}
    \pi = \frac{\exp(\lo)}{1 + \exp(\lo)} \le \exp(\lo) \Rightarrow \left(\frac{\pi}{\E[\pi]}\right)^{\lambda} \le \exp(\lambda(\lo - \log(\E[\pi]))).
\end{equation*}
As a result, the second part of Assumption \ref{as:overlap} puts restrictions on the right tail of the distribution of $\frac{\pi}{\E[\pi]}$, requiring it to have bounded moments.\footnote{We also have the opposite inequality: $\pi \ge \frac{\exp(\lo)}{2} \{\lo \le 0\} \Rightarrow \left(\frac{\pi}{\E[\pi]}\right)^{\lambda} \ge \frac{1}{2^{\lambda}}\exp(\lambda(\lo - \log(\E[\pi])))\{\lo \le 0\}$, which complements the upper bound in the relevant regime where $\lo$ goes to negative infinity.} Finally, the last restriction is a very weak bound on the magnitude of log-odds. Assumption \ref{as:overlap} is trivially satisfied if the strict overlap assumption (e.g., \citealp{hirano2003efficient}) holds. 

Our subsequent restriction controls the statistical complexity of the feature space. Given our focus on the finite-dimensional linear subspaces in the main text, we state it in terms of $p$, the dimension of $\spn{\m{F}}$ in \eqref{eq:main_text_set}. This assumption puts an upper bound on $p$ in terms of intrinsic parameters of the data: the number of treated units and the number of effective periods. 
\begin{assumption}\label{as:comp} \textsc{(Statistical complexity)}\\
$\frac{p}{\mathbb{E}[\pi]n} \ll 1$, and either (a) $\frac{p}{\mathbb{E}[\pi]n} \ll \sqrt{\frac{\Tef}{\mathbb{E}[\pi]n}}$ or (b) $\frac{p}{\mathbb{E}[\pi]n} \ll  \frac{1}{\sqrt{\Tef}}$ as $T_0$ and $n$ increase to infinity.
\end{assumption}
In the two-way example discussed in Section \ref{subsec:exm}, this reduces to the assumption that the number of pre-treatment periods is small relative to the expected number of treated units. That is, we have $p = T_0$ and $\Tef \sim T_0$, and it reduces to $T_0 \ll \E[\pi] n$. When we nonetheless have enough pre-treatment periods, i.e. when $\sqrt{\mathbb{E}[\pi]n} \ll T_0 \ll \mathbb{E}[\pi]n$, we are in the regime in which the SC estimator has asymptotically negligible bias. See Corollary~\ref{cor:as_norm_van}.

Our final assumption is a standard restriction on the behavior of the outcome and covariates, which we expect to hold in most applications. 
\begin{assumption}\label{as:out_mom} \textsc{(Outcome moments)}\\
    (a) There exist absolute constants $0 <\sigma_{\min}< \sigma_{\max} < \infty$ such that the conditional variance $\sigma^2 := \E[\epsilon^2|\eta, X]$ belongs to $[\sigma^2_{\min}, \sigma^2_{\max}]$ with probability 1; (b) there exists an absolute constant $\lambda_{\min}>0$ such that the minimal eigenvalue of the $p\times p$ covariance matrix $\mathbb{V}[\boldsymbol{\phi}(X)]$ is greater than $\lambda_{\min}$; (c) the variance of $\mu$ is bounded, $\mathbb{V}[\mu] = O(1)$.
\end{assumption}

\section{Linear panel models}\label{sec:exm}
In this section, we discuss a class of examples -- linear dynamic panel data models with fixed effects, thus expanding the two-way example from Section \ref{subsec:exm}. We show that it satisfies Assumption \ref{as:emp_set} for a set $\m{F}$, which consists of levels of observed variables, including pre-treatment outcomes.

\subsection{Setup}
Consider a vector of variables $(Y_{t}(0),X_{t}(0))$, where $Y_{t}(0)\in \mathbb{R}$ is a primary outcome of interest, and $X_{t}(0) = (X_{t}^{(1)},\dots, X_{t}^{(l)})\in \mathbb{R}^l$ is a vector of covariates, which can contain $Y_{t}(0)$ as one of its coordinates.  We specify a conditional mean model for $Y_{t}(0)$ given unobserved heterogeneity $\eta$ and past values of $X_{t}(0)$:
\begin{equation}\label{eq:he_model}
    Y_{t}(0) = \lambda_t + \eta^\top\psi_t + \sum_{k=1}^{K}X^\top_{t-k}(0)\beta_{t,k} + \epsilon_{t}, \quad \E[ \epsilon_t| \eta, X_{t-1}(0), X_{t-2}(0),\dots]= 0. 
\end{equation}
Here, $\lambda_t \in \mathbb{R}$, $\psi_t\in \mathbb{R}^{d}$ and $\beta_{t,k} \in \mathbb{R}^{l}$ are fixed parameters, while $\eta\in \mathbb{R}^{d}$ and $\epsilon_{t}\in \mathbb{R}$ are random variables. Without loss of generality we impose two normalizations and assume $\E[\eta] = 0$ and $\mathbb{V}[\eta] = \mathcal{I}_d$. Writing (\ref{eq:he_model}) for each unit,
\begin{equation*}
    Y_{i,t}(0) = \lambda_t + \eta_i^\top\psi_t + \sum_{k=1}^{K}X^\top_{i,t-k}(0)\beta_{t,k}  + \epsilon_{i,t},
\end{equation*}
one can see that this model allows for aggregate shifters $\lambda_t$ and $\psi_t$ (which we treat as fixed quantities), persistent unit-level heterogeneity $\eta_i$ in responses to these shifters and dynamic effects of past values of $X_{i,t}(0)$. 

The researchers observes $n$ units with $X_i := ((Y_{i,1},X_{i,1}),\dots, (Y_{i,T_0}, X_{i,T_0}))$, $Y_i := Y_{i,T_0+1}$, and a policy $D_i \in \{0,1\}$, which is implemented in period $T_0+1$. We impose Assumption \ref{as:pot_outcomes} and intepret $X_i$ as $X_i(0)$ satisfying model (\ref{eq:he_model}). We impose Assumption \ref{as:sel}, treating observations for each unit as an i.i.d. realization from the model (\ref{eq:he_model}) and allow $D_i$ to be correlated with $X_i$ and $\eta_i$. We assume that the researcher uses the estimator described in Section \ref{sec:estimator} with $\mathcal{F} := \{f: \sum_{t = 0}^{T_0}\sum_{j=1}^l \beta_{tj}X_{t}^{(j)}, \| \beta\|_2 \le 1\}$, i.e., the weights are constructed using levels of the pre-treatment covariates.  

\begin{remark}
    If $l=1$ and $X_{t}(0) = Y_{t}(0)$, then (\ref{eq:he_model}) describes a linear auto-regressive model for $Y_{t}(0)$. If $X_{t}(0)$ includes additional variables, then (\ref{eq:he_model}) describes a dynamic model for $Y_{t}(0)$, leaving the rest of the dynamic system unspecified. Additional variables in $X_{t}(0)$ can be strictly or sequentially exogenous. Versions of this model were extensively analyzed in the econometric literature. Under certain restrictions on $\psi_t$ and $\beta_{t,k}$ parameters of this model are identified for fixed $T_0$ and can be estimated at a usual parametric rate using an appropriate Generalized Method of Moments (GMM) estimator (see \cite{arellano2003panel} for a textbook treatment). In many cases, parameters are identified only weakly, making the resulting estimators unstable, and additional information that goes beyond (\ref{eq:he_model}) is needed  (e.g., \citealp{blundell1998initial}).  There is a vast literature on the estimation of (\ref{eq:he_model}) with growing $T_0$ (e.g., \citealp{bai2009panel,moon2015linear,moon2017dynamic}) using OLS with fixed effects. Importantly, most inference results for fixed and growing $T_0$ available in the literature assume that the dimension of $d$ is fixed, with \cite{freeman2023linear} being an important exception.
\end{remark}

\begin{remark}
Depending on applications, there are multiple reasons to expect the treatment indicator $D$ to be correlated with both $\eta$ and pre-treatment covariates. For example, in the context of labor market training programs, the individual earnings before enrollment tend to decrease (\citealp{ashenfelter1978estimating,ashenfelter1985using}), which suggests that adverse shocks to earnings are an important source of selection in addition to permanent differences in productivity captured by $\eta$. In the context of cross-country comparisons, it is well-known that political reforms, e.g., democratization, tend to be correlated with economic outcomes in previous periods (e.g., \citealp{acemoglu2019democracy}). More broadly, in policy evaluation exercises where units often represent geographic regions, it is natural to expect the adoption of the policy to be connected with the current state of the local economy summarized by functions of $X$.
\end{remark}

\subsection{Analysis}

Our first objective is to understand how the effective number of periods $\Tef$ depends on the structure of the model and the number of pre-treatment periods $T_0$. We do this for the setting where $X_t = Y_{t}$ leaving the discussion of additional variables to the next section. 

For $t\in \{K+1,\dots, T_0\}$ we define $\tilde Y_{t} := Y_t - \sum_{k=1}^{K}Y^\top_{t-k}(0)\beta_{t,k}-\lambda_t$, and the diagonal covariance matrix $\Sigma := \mathbb{V}[\{\epsilon_{K+1},\dots, \epsilon_{T_0}\}]$.  We have the following bound:
\begin{multline}\label{eq:gr_dim}
    \min_{f \in \spn{\mathbf{1},\m{F}}}\| f -\mu\|_2^2 \le \left\| \sum_{t = K+1}^{T_0} c_t \tilde Y_{t} - \eta^\top \psi_{T_0+1}\right\|_2^2 = \\
    \left\|\sum_{t = K+1}^{T_0} c_t\psi_{t}  -\psi_{T_0+1}\right\|_{2}^2 + \mathbf{c}^\top \Sigma \mathbf{c} \le \left\|\sum_{t = K+1}^{T_0} c_t\psi_{t}  -\psi_{T_0+1}\right\|_{2}^2 + \|\Sigma\|_{op}\sum_{t=K+1}^{T_0}c_{t}^2,
\end{multline}
where the first inequality follows from the fact that $ \lambda_{T_0+1}+ \sum_{t = K+1}^{T_0} c_t \tilde Y_{t}  + \sum_{k=1}^{K} Y^\top_{t-k}\beta_{t,k}$ belongs to $\spn{\mathbf{1},\m{F}}$ for all possible values of  $\mathbf{c}^\top := (c_{K+1},\dots, c_{T_0})$. 

We define $d\times(T_0-K)$ matrix $\Psi$, with columns equal to $\psi_t$, and consider its singular value decomposition: $\Psi = U \tilde D V^\top$, where $U$ is a $d\times d$ orthogonal matrix, $V$ is a $(T_0-K)\times (T_0 -K)$ matrix, and $\tilde D$ is a $d\times (T_0-K)$ matrix with zeros everywhere except the main diagonal. We use $\sigma(j)$ to denote the singular values (elements on the main diagonal of $\tilde D$), which we arrange in decreasing order. We also define a vector $\xi =(\xi_1,\dots, \xi_{d}) := U^\top \psi_{T_0+1}$, where each $\xi_j$ is the coefficient in projection of $\psi_{T_0+1}$ on the corresponding left singular vector of matrix $\Psi$. Using this notation and minimizing the bound \eqref{eq:gr_dim} over $\mathbf{c}$ we get:
\begin{equation}\label{eq:gen_bound}
     \min_{f \in \spn{\mathbf{1},\m{F}}}\| f -\mu\|_2^2 \le \|\Sigma\|_{op}\sum_{j = 1}^{\min\{d,T_0-K\}} \frac{\xi_j^2}{\sigma^2(j) + \|\Sigma\|_{op} }.
\end{equation}
As a result, the behavior of $ \min_{f \in \spn{\mathbf{1},\m{F}}}\| f -\mu\|_2^2=\frac{1}{\Tef}$ is governed by the decay of $\sigma^2(j)$, i.e., by how pronounced different components of $\psi_t$ are in the past, and by their alignment with $\xi_j$, i.e., how relevant different components of $\psi_t$ are for predicting the future. 

\subsubsection{Two-way model}
Suppose $l = d = 1$, and $\psi_t \equiv \psi$. This reduces our setup to a standard two-way model with an auto-regressive error structure:
\begin{equation}\label{eq:twfe_exm}
Y_{i,t}(0) = \lambda_t + \psi\eta_i + \sum_{k=1}^{K}\beta_{t,k}Y_{i,t-k}(0) + \epsilon_{i,t}, \quad \E[ \epsilon_{i,t}| \eta_i, Y_{i,t-1}(0), Y_{i,t-2}(0),\dots]= 0.
\end{equation}

This model generalizes the two-way example from Section \ref{subsec:exm} because it allows for the auto-regressive component. We also impose restrictions on $\lo$, describing its relationship with permanent heterogeneity and past shocks to the outcomes:
\begin{equation}\label{eq:assign_model}
    \lo_i = \alpha_{c} +\alpha_{\eta} \eta_i+ \sum_{j=0}^{k}\alpha_{j} \epsilon_{i,T_0-j},
\end{equation}
where $\alpha_{c}, \alpha_{\eta}, (\alpha_{0}, \dots, \alpha_k)$ are fixed constants. This specification is more restrictive than needed for our results to hold, and we use it to simplify the exposition. In particular, these conditions imply that $\muplo_i = \lo_i$ and thus the error defined in \eqref{eq:u_def} is equal to zero, $u_i = 0$. This condition affects the variance in Corollary \ref{cor:two-way} below. Despite its simplicity, the assignment process \eqref{eq:assign_model} implies $D_i$ is not strictly exogenous, and thus the DiD-based methods are not guaranteed to work. We return to this discussion in Section \ref{sec:sim} where we conduct numerical simulations for a similar model.

Applying our general bound \eqref{eq:gen_bound} with $d =1$ we get $\xi_{1}^2 = \psi^2$ and $\sigma^2(1) = T_0-K$ and thus we get:
\begin{equation*}
     \min_{f \in \spn{\mathbf{1},\m{F}}}\| f -\mu\|_2^2 \le \frac{\|\Sigma\|_{op}\psi^2}{\psi^2 (T_0-K) + \|\Sigma\|_{op}},
\end{equation*}
which is a minor generalization of the equality we had in Section \ref{subsec:exm}. We use this result to state a corollary of the general Theorem \ref{th:const_share}, with explicit assumptions on errors instead of high-level Assumptions \ref{as:subg} - \ref{as:out_mom}.
\begin{corollary}\label{cor:two-way}
    Suppose (a) Assumptions \ref{as:pot_outcomes} - \ref{as:sel} hold; (b) $Y_{t}(0)$ satisfies \eqref{eq:twfe_exm}; (c) $\lo$ satisfies \eqref{eq:assign_model}; (d) $\eta, \epsilon_{1},\dots, \epsilon_{T_0+1}$ are independent mean-zero random variables with the uniformly bounded subgaussian norms, and $ 0< \sigma^2_{\min} \le \mathbb{V}[\epsilon_{t}] \le \sigma^2_{\max} <\infty$;  (e) $T_0 \ll n$ and $\zeta = O(1)$. Then we have:
    \begin{equation*}
    \hat \tau - \tau = \bias + \Pnn \frac{D_i - \pi_i}{1-\pi_i}\frac{\epsilon_{i,T_0+1}}{\E[\pi]}+ o_p\left(\frac{1}{\sqrt{n}}\right) + o_p\left(\frac{1}{T_0}\right),
    \end{equation*}
    and if $T_0^2 \gg n$, then  $\sqrt{\overline \pi n}(\hat \tau - \tau)\rightharpoonup_{d} \mathcal{N}(0, \sigma^2_{as})$, where $\sigma^2_{as} =     \E\left[\frac{\mathbb{V}[\epsilon_{i,T_0+1}]}{(1-\pi)} | D= 1\right]$.
\end{corollary}
This result describes the behavior of the SC control estimator in applications where the underlying outcomes follow the two-way model. Crucially, it relaxes the strict exogeneity that underlies the DiD-based analysis. In particular, in applications where $T_0$ is large enough, the SC estimator is asymptotically unbiased and thus can be used for inference. This is the type of behavior we observed in simulations in Section \ref{sec:sim} with Corollary \ref{cor:two-way} providing formal support to our claim that the SC estimator is a reasonable alternative to the DiD estimator. 

\subsubsection{Interactive fixed effects}
We continue assuming that $l = 1$ and thus $X_t = Y_t$. But now, we set $d = f+1$ and write 
\begin{equation}\label{eq:int_fe}
Y_{i,t} = \lambda_t + \psi^{(1)}\eta^{(1)}_i + (\psi_t^{(2)})^\top\eta^{(2)}_i + \sum_{k=1}^{K}\beta_{t,k}Y_{i,t-k} + \epsilon_{i,t}, \quad \E[ \epsilon_{i,t}| \eta_i, Y_{i,t-1}(0), Y_{i,t-2}(0),\dots]= 0,
\end{equation}
where the dimension of $\eta^{(2)}$ is equal to $f$. The key difference between this model and the two-way model considered in the previous section is in the behavior of $\Tef$. Below we discuss two examples in which $\Tef$ increases as $T_0$, and increases at a rate slower than $T_0$.  

We extend the selection model \eqref{eq:assign_model} to allow for interactive effects:
\begin{equation}\label{eq:assign_model_if}
    \lo_i = \alpha_{c} +\alpha^{(1)}\eta^{(1)}_i + (\alpha^{(2)})^\top\eta^{(2)}_i + \sum_{j=0}^{k}\alpha_{j} \epsilon_{i,T_0-j},
\end{equation}
We also define the analog of $\xi$ for the selection model, $\xi^{(sel)} := U^{\top}[\alpha^{(1)},(\alpha^{(2)})^\top]^\top$.

\paragraph{Strong factors:} We start by assuming that $f$ is constant, i.e., it does not increase with $n$ and $T_0$, which is a standard assumption in the panel data literature that works with both finite $T_0$ (e.g., \citealp{holtz1988estimating}) and large $T_0$ (e.g., \citealp{bai2009panel}). First, we consider the environment in which $\min_{j} \sigma^2(j) \sim T_0$, i.e., the factors are strong (which is analogous to Assumption B in \citealp{bai2009panel}). In this case our general bound \eqref{eq:gen_bound} implies:
\begin{equation*}
 \min_{f \in \spn{\mathbf{1},\m{F}}}\| f -\mu\|_2^2 \le \|\Sigma\|_{op}\sum_{j = 1}^{\min\{d,T_0-K\}} \frac{\xi_j^2}{\sigma^2(j) + \|\Sigma\|_{op} }\sim \frac{\|\Sigma\|_{op}\| \psi_{T_0+1}\|_2^2}{T_0}. 
\end{equation*}
This guarantees that in the finite-dimensional factor model, as long as all factors are equally strong, the behavior of $\Tef$ is the same as in the two-way model we discussed in the previous section. As a result, the immediate analog of Corollary \ref{cor:two-way} holds under the same assumptions.

\paragraph{Growing number of factors:} Finally, we consider a situation where $f$ is growing with $T_0$. In this case, some factors are bound to be weak as long as the variance of the outcome is bounded. Formally, this means that $\sigma^2(j)$ has to decrease with $j$, and we assume that it decreases at polynomial rate $\sigma^2(j) \sim T_0 j^{-\kappa}$, where $\kappa>1$. We also assume that the factor $\psi_{T_0+1}$ is typical, in a sense that $\xi_j^2 \sim \frac{\sigma^2(j)}{ T_0}$.\footnote{Formally, $\sigma^2(j) = u_j^{\top}\Psi v_j$ and $\xi_j = u_j^{\top} \psi_{T_0+1}$, where $u_j$ and $v_j$ are the corresponding singular vectors. Suppose  $v_1 = \left(\frac{1}{\sqrt{T_0-K}},\dots, \frac{1}{\sqrt{T_0 - K}}\right)$ and thus $\frac{1}{\sqrt{T_0-K}}v_1$ corresponds to averaging over time. In this case,  $\frac{u_1^{\top}\Psi v_1}{\sqrt{T_0-K} } = u_1^{\top} \overline{\psi}$, where $ \overline{\psi} :=\frac{1}{T_0 -K}\sum_{t > K}^{T_0} \psi_t$. As a result, $\frac{1}{\sqrt{T_0-K}}\sigma(1) \sim \xi_1 = u_1^{\top} \psi_{T_0+1}$ as long as $ \overline{\psi}$ is close to $\psi_{T_0+1}$.} In this case, using the upper bound \eqref{eq:gen_bound} we get:
\begin{equation*}
     \min_{f \in \spn{\mathbf{1},\m{F}}}\| f -\mu\|_2^2 \le \|\Sigma\|_{op}\sum_{j = 1}^{\min\{d,T_0-K\}} \frac{\xi_j^2}{\sigma^2(j) + \|\Sigma\|_{op} } \sim \left(\frac{\|\Sigma\|_{op} }{T_0}\right)^{1-\frac1{\kappa}}
\end{equation*}
This example demonstrates that $\Tef$ can behave as $T_0^{1-\frac1{\kappa}}$ in models where the dimension of the factors is large. Notably, the logic for a slower rate here is different than in the policy shock example in Section \ref{subsec:exm}. There the factor was irrelevant for explaining the past but was very relevant for predicting the future. In the current model, the factors less critical in explaining the past are also less important in predicting the future. However, one cannot ignore most of the irrelevant factors because, when taken together, they become sufficiently strong. 

We use the derivations above to state another corollary of Theorem \ref{th:const_share}.
\begin{corollary}\label{cor:int_fe}
    Suppose (a) Assumptions \ref{as:pot_outcomes} - \ref{as:sel} hold; (b) $Y_{t}(0)$ satisfies \eqref{eq:int_fe}, $\sigma^2(j) \sim T_0 j^{-\kappa}$ and $\xi_j^2 \sim \frac{\sigma^2(j)}{ T_0}$, where $\kappa > 3$; (c) $\lo$ satisfies \eqref{eq:assign_model_if} with  $\left(\xi_j^{(sel)}\right)^2 \sim \frac{\sigma^2(j)}{ T_0}$; (d) $\eta, \epsilon_{1},\dots, \epsilon_{T_0+1}$ are independent mean-zero random variables with the uniformly bounded subgaussian norms, and $ 0< \sigma^2_{\min} \le \mathbb{V}[\epsilon_{t}] \le \sigma^2_{\max} <\infty$;  (e) $T_0 \ll n^{\frac{\kappa}{\kappa+1}}$, and $\zeta = O(1)$. Then we have
    \begin{equation*}
    \hat \tau - \tau = \bias + \Pnn \frac{D_i - \pi_i}{1-\pi_i}\frac{\epsilon_{i,T_0+1}}{\E[\pi]}+ o_p\left(\frac{1}{\sqrt{n}}\right) + o_p\left(\frac{1}{T_0^{1-\frac1{\kappa}}}\right),
    \end{equation*}
    and if $T_0 \gg n^{\frac{\kappa}{2(\kappa-1)}}$, then  $\sqrt{\overline \pi n}(\hat \tau - \tau)\rightharpoonup_{d} \mathcal{N}(0, \sigma^2_{as})$, where $\sigma^2_{as} =     \E\left[\frac{\mathbb{V}[\epsilon_{i,T_0+1}]}{(1-\pi)} | D= 1\right]$.
\end{corollary}
This result demonstrates that the SC estimator can be asymptotically unbiased in the model with an increasing number of factors. However, the restrictions on the relationship between $T_0$ and $n$ become more stringent. For example, if $\kappa = 4$, which implies a relatively fast convergence of the singular values, then for asymptotic unbiasedness, we require $T_0^{\frac54}\ll n \ll T_0^\frac32$. 

Following the same logic as in \cite{freeman2023linear}, one can interpret the model with a growing dimension of the interactive fixed effects as a semiparametric model with two-way unobserved heterogeneity. In this case, $\kappa$ can be interpreted as the smoothness of the nonparametric part of that model.

\subsection{VAR models}\label{subsec:var}
We now briefly discuss the role that additional covariates can play in estimation. For simplicity, we do this in the context of a single additional variable $Z_t$, assuming that $X_t = (Y_t, Z_t)$. We specify the evolution of the underlying potential outcomes using a VAR model:
\begin{equation}\label{eq:var_model}
\begin{aligned}
    Y_{t}(0) = \lambda_t^{Y} + \eta^\top \psi_{t}^{Y} + \sum_{k=1}^K\left(\beta_{1,t,k}^{Y} Y_{t-k}(0) + \beta_{2,t,k}^{Y} Z_{t-k}(0)\right) + \epsilon_{t}^{Y},\\
    Z_{t}(0) =  \lambda_t^{Z} + \eta^\top \psi_{t}^{Z} + \sum_{k=1}^K\left(\beta_{1,t,k}^{Z} Z_{t-k}(0) + \beta_{2,t,k}^{Z} Y_{t-k}(0)\right) + \epsilon_{t}^{Z},
\end{aligned}
\end{equation}
where $\mathbb{E}[\eta] = 0$, $\mathbb{V}[\eta] = \mathcal{I}_d$, and $\mathbb{E}\qty[\left(\epsilon_{t}^{Y},\epsilon_{t}^{Z}\right)| \eta, Y_{t-1}(0), Z_{t-1}(0), \dots] = 0$. This model allows $Z_t$ to be either a sequentially or a strictly exogenous variable. For the latter case, we need to set the coefficients $\beta^{Z}_{2,t,k}$ to zero and assume that $\epsilon_{t}^{Y}, \epsilon_{t}^{Z}$ are uncorrelated.

As before, we define the transformed variables:
\begin{align*}
    \tilde Y_{ t} = Y_{t} - \sum_{k=1}^K\left(\beta_{1,t,k}^{Y} Y_{t-k} + \beta_{2,t,k}^{Y} Z_{t-k}\right) -\lambda_t^{Y}, \quad
    \tilde Z_{ t} = Z_{t} - \sum_{k=1}^K\left(\beta_{1,t,k}^{Z} Z_{t-k} + \beta_{2,t,k}^{Z} Y_{t-k}\right) - \lambda_t^{Z},
\end{align*}
Using this notation, we get the following bound:
\begin{multline}\label{eq:var_bound}
      \min_{f \in \spn{\mathbf{1},\m{F}}}\| f -\mu\|_2^2 \le \left\|\eta \psi_{T_0+1}^{Y} -   \sum_{t = K+1}^{T_0}c_t^Y  \tilde Y_{ t} -  \sum_{t = K+1}^{T_0}c_t^Z\tilde Z_{t} \right\|_2^2 =\\
      \left\|  \psi_{T_0+1}^{Y}  -  \sum_{t = K+1}^{T_0}c_t^Y  \psi_t^{Y}-  \sum_{t = K+1}^{T_0}c_t^Z\psi_t^{Z}\right\|_2^2 + \sum_{t = K+1}^{T_0} \mathbf{c}_t^{\top}\Sigma_t\mathbf{c}_t,
\end{multline}
where $\Sigma_t = \mathbb{V}\qty[\left(\epsilon_{t}^{Y},\epsilon_{t}^{Z}\right)]$, and $\mathbf{c}_t^\top = \qty(c_t^{Y}, c_t^{Z})$. 

One can view this bound as a minor generalization of \eqref{eq:gr_dim} and optimize it over the coefficients as we did in \eqref{eq:gen_bound}. We present it separately to emphasize two different roles of $Z_{t}$. The first one is apparent from \eqref{eq:var_bound} -- we can use this variable to predict the relevant unobserved heterogeneity. This generalizes the discussion in Section \ref{subsec:exm}, where we showed how the availability of $Z_{t}$ can increase $\Tef$ in cases with unobserved policy shocks. $Z_{t}(0)$ plays an additional role in \eqref{eq:var_model}, allowing us to introduce an additional time-varying shock $\epsilon_{t}^Z$. This variable affects the future outcomes through the VAR system and can be part of the selection equation.

\section{Staggered Adoption}\label{sec:discussion}
In this section, we describe the adaptation of the previously developed results to the staggered adoption applications. We first briefly discuss the general principles behind the methods currently used for such settings and their potential problems. We then propose a particular estimator for contemporaneous treatment effects and outline the critical conceptual assumptions needed for its validity.

\subsection{Status quo}

Applications where units adopt the treatment sequentially, commonly called staggered designs, are ubiquitous in economics. The standard tool used to analyze such designs is the TWFE regression \eqref{eq:event_study} and its recent extensions for models with heterogenous effects (e.g., \citealp{de2020two,callaway2021difference,sun2021estimating,borusyak2021revisiting}). Another option, in particular for applications with few treated units, is the adaptation of the SC method described in \cite{ben2022synthetic} (see also \citealp{arkhangelsky2021synthetic}  and \citealp{cattaneo2022uncertainty}). 

All these solutions rely on the same principle, transforming the staggered design problem into a sequence of more straightforward block design problems. In particular, for a group of units that adopts the treatment in period $t$, researchers construct a suitable control group using some of the units that have not received the treatment. Once this group is constructed, the resulting data is analyzed using methods for block designs. Two choices for the control group are particularly prominent in the literature: it either includes all non-treated units or only ``never treated'' units. See \cite{callaway2021difference} for a discussion of these two control groups.

Our proposal below follows the same logic, with one important caveat: we suggest using SC only to learn the contemporaneous effects of the treatment. This distinguishes our proposal from some of the abovementioned methods, which explicitly focus on dynamic effects (e.g., \citealp{callaway2021difference, sun2021estimating}). This caveat is due to the possibility of dynamic selection -- the adoption of the treatment based on past outcomes. Dynamic selection arises naturally in staggered adoption applications with observational (e.g., \citealp{heckman2007dynamic}) or experimental data (e.g., \citealp{xiong2023optimal}).

To understand why dynamic selection creates a problem, consider period $t$ and a group of units that have not yet adopted the treatment. Under natural generalizations of the assumptions from Section \ref{sec:estimator}, which we discuss below, these units form a suitable control group for the period $t$, allowing us to estimate the contemporaneous effects. To learn the dynamic effects, we need a control group that has not received the treatment in future periods, e.g., period $t+1$. However, precisely the fact that this group has not received the treatment tells us that their outcomes in period $t+1$ should be systematically different. Using this group without additional adjustments will produce biased estimators of dynamic effects. This problem applies to all the methods mentioned above to the extent that they estimate dynamic effects. These estimators are valid only in the absence of dynamic selection. 

Dynamic selection is well understood in the literature on sequential unconfoundedness in biostatistics, e.g., \cite{robins2000marginal}, which also offers a solution. To identify and estimate dynamic effects, one uses sequential one-period comparisons and projects them back to the current period. See \cite{viviano2021dynamic} for a recent balancing algorithm that implements this logic in settings without unobserved heterogeneity. Our proposal below can be used as the first step toward developing similar algorithms for environments with unobserved heterogeneity. 

\subsection{Estimator and Assumptions}

To incorporate the possibility of multiple treatment periods, we enrich the setup discussed in Section \ref{sec:estimator} and assume that for each period $t \in \{-T_0,\dots, T_1\}$, we observe $\{X_{i,t}, W_{i,t}, Y_{i,t}\}_{i=1}^n$. Here $X_{i,t}$ includes all information observed up to period $t$ for unit $i$, $W_{i,t}$ describes the treatment status of unit $i$ in period $t$, and $Y_{i,t}$ describes the outcome of interest. Our first assumption restricts the behavior of $W_{i,t}$ over time.
\begin{assumption}\label{as:st_ad}\textsc{(Staggered adoption)}\\
    For every $t \in \{-T_0,\dots, T_1\}$ we have $W_{i,t+1} \ge W_{i,t}$, $W_{i,-1} \equiv 0$.
\end{assumption}
This assumption implies that no units are treated in the first $T_0$ periods, and overall, there are $T_1+1$ possible treatment dates. We use it to define for each $t \ge 0$ two subsamples: $\m{D}_{t,1} := \{i: W_{i,t} = 1, W_{i,t-1} = 0\}$ and $\m{D}_{t,0} := \{i: W_{i,t} = 0\}$. Let $n_{t,1}$ and $n_{t,0}$ be the size of the corresponding subsample and define $n_t := n_{t,1} + n_{t,0}$, $\overline \pi_t:= \frac{n_{t,1}}{n_t}$. 

In each period $t\ge 0$, we form the following estimator:
\begin{equation}\label{eq:st_estimator}
    \hat \tau^{(t)} := \frac{ \sum_{i \in \m{D}_{t,1}} Y_{i,t}}{\overline \pi_t n_{t}} -\frac{\sum_{i \in \m{D}_{t,0}} \hat \omega_i^{(t)} Y_{i,t}}{n_t} 
\end{equation}
The weights $\{\hat \omega^{(t)}_{i}\}_{i \in \m{D}_{t,1}\cup \m{D}_{t,0}}$ are constructed in the same way as before:
\begin{equation}\label{eq:st_pr_problem}
\begin{aligned}
   \hat \omega^{(t)} := \argmin_{\omega\ge 0}&\left\{\frac{\zeta^2}{n_t^2}\sum_{i \in \m{D}_1 \cup \m{D}_0} \omega_i\log(\omega_i) + \sum_{l=1}^{p^{(t)}} \qty(\frac{\sum_{i\in \m{D}_{t,1}} \phi^{(t)}_l(X_i)}{ \overline \pi_t n_{t}} - \frac{\sum_{i\in \m{D}_{t,0}}\omega_i \phi^{(t)}_l(X_i)}{n_t} )^2 \right\}\\
\text{subject to: } &\frac{1}{n_t}\sum_{i \in \m{D}_{t,0}}\omega_i = 1.
\end{aligned}
\end{equation}
Compared to \eqref{eq:pr_problem} this estimator has two differences: for each period $t$ we have different samples and use different functions for balancing $\phi^{(t)}(X_i) = \left(\phi_{1}^{(t)}(X_i), \dots, \phi_{p^{(t)}}^{(t)}(X_i)\right)$. Both of these changes are natural: the sample size $n_t$ diminishes over time, whereas the amount of the pre-treatment information increases. 

Next, we formalize the causal model behind a relevant part of the observed data.
\begin{assumption}\label{as:dyn_pot_out}\textsc{(Dynamic potential outcomes)}\\
    For every $t$ there exist potential outcomes $X_{i,t}(0), Y_{i,t}(1), Y_{i,t}(0)$ such that $X_{i,t}$ and $Y_{i,t}$ satisfy
\begin{equation*}
\begin{aligned}
     &X_{i,t}\{W_{i,t-1} = 0\}  = X_{i,t}(0)\{W_{i,t-1} = 0\},\\
     &Y_{i,t}\{W_{i,t-1} = 0\} = (Y_{i,t}(1) \{W_{i,t} = 1\} + Y_{i,t}(0) \{W_{i,t} = 0\})\{W_{i,t-1} = 0\}
\end{aligned}
\end{equation*}
\end{assumption}
The first part of this assumption specifies that the available information prior to period $t$ that we observe for units not treated in period $t-1$ corresponds to the baseline potential outcomes. This restriction generalizes the no-anticipation part of Assumption \ref{as:pot_outcomes}. Importantly, it does not restrict the behavior of $X_{i,t}$ in any other situation. The second part of the assumption relates the observed outcomes $Y_{i,t}$ to the underlying potential ones but only for units with $W_{i,t-1} = 0$. For those units, we interpret the observed outcomes in the same way as before. Since units can be treated in later periods, this restriction also incorporates the no-anticipation assumption. Assumption \ref{as:dyn_pot_out} does not specify a relationship between the observed and potential outcomes for units treated in earlier periods. The extension to a full dynamic model is conceptually straightforward but is irrelevant given our focus on contemporaneous effects.

Assumption \ref{as:dyn_pot_out} allows us to expand our estimator into two parts:
\begin{equation*}
      \hat \tau^{(t)}  =  \frac{ \sum_{i \in \m{D}_{t,1}} (Y_{i,t}(1) - Y_{i,t}(0)}{\overline \pi_t n_{t}} +\qty(\frac{\sum_{i \in \m{D}_{t,1}} Y_{i,t}(0)}{\overline \pi_t n_{t}} - \frac{\sum_{i \in \m{D}_{t,0}} \hat \omega_i^{(t)} Y_{i,t}(0)}{n_t})
\end{equation*}
The first part of this expansion is an in-sample contemporaneous effect of the treatment in period $t$ for units that adopted the policy in the same period, which we denote $\tau_t$. Despite its dependence on $t$, this effect is static in nature and does not capture any dynamics. The second part of the expansion is the error term. 

To complete the model, we need to specify sampling and selection mechanisms.  We do this with the following assumption, which is the generalization of Assumption \ref{as:sel}.
\begin{assumption}\label{as:dyn_sel}\textsc{(Dynamic selection)}\\
       (a) unit-level data $\left\{X_{i,t}, Y_{i,t}, W_{i,t}, \eta_i\right\}_{t= -T_0}^{T_1}$ are i.i.d. over $i$; (b) $W_{i,t} \independent Y_{i,t}(0)\Bigl|\, \eta_i, X_{i,t}, W_{i,t-1} = 0$, and  $\pi_{i,t} := \E[W_{i,t}| \eta_i, X_{i,t}, W_{i,t-1} = 0]$ belongs to $(0,1)$ with probability 1.
\end{assumption}
This restriction is a natural generalization of Assumption \ref{as:sel} to dynamic contexts. It is a version of the sequential ignorability assumption common in biostatistics (e.g., \citealp{robins2000marginal}) adapted to staggered designs. It implies that in each period, the treatment decision is based on the available information $X_{i,t}$, unobserved characteristic $\eta_i$, and some unobserved shocks that are unrelated to the potential outcomes.  This structure arises naturally in dynamic economic models; see \cite{heckman2007dynamic} for a comprehensive treatment. 

We do not formally analyze the behavior of the estimation error $\hat \tau_t - \tau_t$. Statistical guarantees analogous to Theorems \ref{th:van_share} - \ref{th:const_share} can be derived under natural extensions of Assumption \ref{as:emp_set} and technical assumptions from Section \ref{sec:results}. In particular, similar results will hold in an asymptotic regime where $T_0$ increases to infinity, even if $T_1$ is constant. Notably, the vanishing share results are particularly natural in staggered designs because we can expect $\overline \pi_{t}$ to be small for larger values of $t$. These results describe the marginal behavior for each $t\ge 0$, but as long as $T_1$ is finite, we expect the same guarantees to hold simultaneously for all $t\ge 0$.

\section{Conclusion}\label{sec:conc}
We analyze the large sample properties of the SC method. We derive the asymptotic representation of the resulting estimator using high-level assumptions on the assignment process and the complexity of unobservables. Our results imply that the SC estimator is asymptotically unbiased and normal in a large class of linear panel data models as long as the number of observed pre-treatment periods is large enough. In particular, this justifies using it as an alternative to the DiD estimators. We also show that the SC estimator can fail in models that feature unobserved heterogeneity in the persistence of time-varying shocks. 

\bibliographystyle{plainnat}
\bibliography{references}

\begin{thebibliography}{80}
\providecommand{\natexlab}[1]{#1}
\providecommand{\url}[1]{\texttt{#1}}
\expandafter\ifx\csname urlstyle\endcsname\relax
  \providecommand{\doi}[1]{doi: #1}\else
  \providecommand{\doi}{doi: \begingroup \urlstyle{rm}\Url}\fi

\bibitem[Abadie(2005)]{abadie2005semiparametric}
Alberto Abadie.
\newblock Semiparametric difference-in-differences estimators.
\newblock \emph{The Review of Economic Studies}, 72\penalty0 (1):\penalty0
  1--19, 2005.

\bibitem[Abadie(2021)]{abadie2021using}
Alberto Abadie.
\newblock Using synthetic controls: Feasibility, data requirements, and
  methodological aspects.
\newblock \emph{Journal of Economic Literature}, 59\penalty0 (2):\penalty0
  391--425, 2021.

\bibitem[Abadie and Gardeazabal(2003)]{abadie2003}
Alberto Abadie and Javier Gardeazabal.
\newblock The economic costs of conflict: A case study of the basque country.
\newblock \emph{American Economic Review}, 93\penalty0 (-):\penalty0 113--132,
  2003.

\bibitem[Abadie and L’hour(2021)]{abadie2021penalized}
Alberto Abadie and J{\'e}r{\'e}my L’hour.
\newblock A penalized synthetic control estimator for disaggregated data.
\newblock \emph{Journal of the American Statistical Association}, 116\penalty0
  (536):\penalty0 1817--1834, 2021.

\bibitem[Abadie et~al.(2010)Abadie, Diamond, and Hainmueller]{Abadie2010}
Alberto Abadie, Alexis Diamond, and Jens Hainmueller.
\newblock Synthetic control methods for comparative case studies: Estimating
  the effect of {C}alifornia's tobacco control program.
\newblock \emph{Journal of the American Statistical Association}, 105\penalty0
  (490):\penalty0 493--505, 2010.

\bibitem[Abadie et~al.(2020)Abadie, Athey, Imbens, and
  Wooldridge]{abadie2020sampling}
Alberto Abadie, Susan Athey, Guido~W Imbens, and Jeffrey~M Wooldridge.
\newblock Sampling-based versus design-based uncertainty in regression
  analysis.
\newblock \emph{Econometrica}, 88\penalty0 (1):\penalty0 265--296, 2020.

\bibitem[Abbring and Van~den Berg(2003)]{abbring2003nonparametric}
Jaap~H Abbring and Gerard~J Van~den Berg.
\newblock The nonparametric identification of treatment effects in duration
  models.
\newblock \emph{Econometrica}, 71\penalty0 (5):\penalty0 1491--1517, 2003.

\bibitem[Acemoglu et~al.(2019)Acemoglu, Naidu, Restrepo, and
  Robinson]{acemoglu2019democracy}
Daron Acemoglu, Suresh Naidu, Pascual Restrepo, and James~A Robinson.
\newblock Democracy does cause growth.
\newblock \emph{Journal of political economy}, 127\penalty0 (1):\penalty0
  47--100, 2019.

\bibitem[Alvarez and Arellano(2003)]{alvarez2003time}
Javier Alvarez and Manuel Arellano.
\newblock The time series and cross-section asymptotics of dynamic panel data
  estimators.
\newblock \emph{Econometrica}, 71\penalty0 (4):\penalty0 1121--1159, 2003.

\bibitem[Andersson(2019)]{andersson2019carbon}
Julius~J Andersson.
\newblock Carbon taxes and co 2 emissions: Sweden as a case study.
\newblock \emph{American Economic Journal: Economic Policy}, 11\penalty0
  (4):\penalty0 1--30, 2019.

\bibitem[Andrews and Kasy(2019)]{andrews2019identification}
Isaiah Andrews and Maximilian Kasy.
\newblock Identification of and correction for publication bias.
\newblock \emph{American Economic Review}, 109\penalty0 (8):\penalty0
  2766--2794, 2019.

\bibitem[Arellano(2003)]{arellano2003panel}
Manuel Arellano.
\newblock \emph{Panel data econometrics}.
\newblock Oxford university press, 2003.

\bibitem[Arkhangelsky and Imbens(2022)]{arkhangelsky2022doubly}
Dmitry Arkhangelsky and Guido~W Imbens.
\newblock Doubly robust identification for causal panel data models.
\newblock \emph{The Econometrics Journal}, 25\penalty0 (3):\penalty0 649--674,
  2022.

\bibitem[Arkhangelsky et~al.(2021)Arkhangelsky, Athey, Hirshberg, Imbens, and
  Wager]{arkhangelsky2021synthetic}
Dmitry Arkhangelsky, Susan Athey, David~A Hirshberg, Guido~W Imbens, and Stefan
  Wager.
\newblock Synthetic difference-in-differences.
\newblock \emph{American Economic Review}, 111\penalty0 (12):\penalty0
  4088--4118, 2021.

\bibitem[Armstrong and Koles{\'a}r(2021)]{armstrong2018finite}
Timothy~B Armstrong and Michal Koles{\'a}r.
\newblock Finite-sample optimal estimation and inference on average treatment
  effects under unconfoundedness.
\newblock \emph{Econometrica}, 89\penalty0 (3):\penalty0 1141--1177, 2021.

\bibitem[Ashenfelter(1978)]{ashenfelter1978estimating}
Orley Ashenfelter.
\newblock Estimating the effect of training programs on earnings.
\newblock \emph{The Review of Economics and Statistics}, pages 47--57, 1978.

\bibitem[Ashenfelter and Card(1985)]{ashenfelter1985using}
Orley Ashenfelter and David Card.
\newblock Using the longitudinal structure of earnings to estimate the effect
  of training programs.
\newblock \emph{The Review of Economics and Statistics}, 67\penalty0
  (4):\penalty0 648--660, 1985.

\bibitem[Athey et~al.(2018)Athey, Imbens, and Wager]{athey2018approximate}
Susan Athey, Guido~W Imbens, and Stefan Wager.
\newblock Approximate residual balancing: debiased inference of average
  treatment effects in high dimensions.
\newblock \emph{Journal of the Royal Statistical Society Series B: Statistical
  Methodology}, 80\penalty0 (4):\penalty0 597--623, 2018.

\bibitem[Bai(2003)]{bai2003inferential}
Jushan Bai.
\newblock Inferential theory for factor models of large dimensions.
\newblock \emph{Econometrica}, 71\penalty0 (1):\penalty0 135--171, 2003.

\bibitem[Bai(2009)]{bai2009panel}
Jushan Bai.
\newblock Panel data models with interactive fixed effects.
\newblock \emph{Econometrica}, 77\penalty0 (4):\penalty0 1229--1279, 2009.

\bibitem[Belloni et~al.(2014)Belloni, Chernozhukov, and Hansen]{belloni2014jep}
Alexandre Belloni, Victor Chernozhukov, and Christian Hansen.
\newblock High-dimensional methods and inference on structural and treatment
  effects.
\newblock \emph{Journal of Economic Perspectives}, 28\penalty0 (2):\penalty0
  29--50, 2014.

\bibitem[Ben-Michael et~al.(2021{\natexlab{a}})Ben-Michael, Feller, Hirshberg,
  and Zubizarreta]{ben2021balancing}
Eli Ben-Michael, Avi Feller, David~A Hirshberg, and Jos{\'e}~R Zubizarreta.
\newblock The balancing act in causal inference.
\newblock \emph{arXiv preprint arXiv:2110.14831}, 2021{\natexlab{a}}.

\bibitem[Ben-Michael et~al.(2021{\natexlab{b}})Ben-Michael, Feller, and
  Rothstein]{ben2018augmented}
Eli Ben-Michael, Avi Feller, and Jesse Rothstein.
\newblock The augmented synthetic control method.
\newblock \emph{Journal of the American Statistical Association}, 116\penalty0
  (536):\penalty0 1789--1803, 2021{\natexlab{b}}.

\bibitem[Ben-Michael et~al.(2022)Ben-Michael, Feller, and
  Rothstein]{ben2022synthetic}
Eli Ben-Michael, Avi Feller, and Jesse Rothstein.
\newblock Synthetic controls with staggered adoption.
\newblock \emph{Journal of the Royal Statistical Society Series B: Statistical
  Methodology}, 84\penalty0 (2):\penalty0 351--381, 2022.

\bibitem[Beyhum and Gautier(2022)]{beyhum2022factor}
Jad Beyhum and Eric Gautier.
\newblock Factor and factor loading augmented estimators for panel regression
  with possibly nonstrong factors.
\newblock \emph{Journal of Business \& Economic Statistics}, 41\penalty0
  (1):\penalty0 270--281, 2022.

\bibitem[Blundell and Bond(1998)]{blundell1998initial}
Richard Blundell and Stephen Bond.
\newblock Initial conditions and moment restrictions in dynamic panel data
  models.
\newblock \emph{Journal of econometrics}, 87\penalty0 (1):\penalty0 115--143,
  1998.

\bibitem[Blundell et~al.(1999)Blundell, Griffith, and
  Van~Reenen]{blundell1999market}
Richard Blundell, Rachel Griffith, and John Van~Reenen.
\newblock Market share, market value and innovation in a panel of british
  manufacturing firms.
\newblock \emph{The review of economic studies}, 66\penalty0 (3):\penalty0
  529--554, 1999.

\bibitem[Blundell et~al.(2002)Blundell, Griffith, and
  Windmeijer]{blundell2002individual}
Richard Blundell, Rachel Griffith, and Frank Windmeijer.
\newblock Individual effects and dynamics in count data models.
\newblock \emph{Journal of econometrics}, 108\penalty0 (1):\penalty0 113--131,
  2002.

\bibitem[Bonhomme et~al.(2022)Bonhomme, Lamadon, and
  Manresa]{bonhomme2022discretizing}
St{\'e}phane Bonhomme, Thibaut Lamadon, and Elena Manresa.
\newblock Discretizing unobserved heterogeneity.
\newblock \emph{Econometrica}, 90\penalty0 (2):\penalty0 625--643, 2022.

\bibitem[Borusyak et~al.(2021)Borusyak, Jaravel, and
  Spiess]{borusyak2021revisiting}
Kirill Borusyak, Xavier Jaravel, and Jann Spiess.
\newblock Revisiting event study designs: Robust and efficient estimation.
\newblock \emph{arXiv preprint arXiv:2108.12419}, 2021.

\bibitem[Callaway and Sant’Anna(2021)]{callaway2021difference}
Brantly Callaway and Pedro~HC Sant’Anna.
\newblock Difference-in-differences with multiple time periods.
\newblock \emph{Journal of Econometrics}, 225\penalty0 (2):\penalty0 200--230,
  2021.

\bibitem[Cattaneo et~al.(2021)Cattaneo, Feng, and
  Titiunik]{cattaneo2021prediction}
Matias~D Cattaneo, Yingjie Feng, and Rocio Titiunik.
\newblock Prediction intervals for synthetic control methods.
\newblock \emph{Journal of the American Statistical Association}, 116\penalty0
  (536):\penalty0 1865--1880, 2021.

\bibitem[Cattaneo et~al.(2022)Cattaneo, Feng, Palomba, and
  Titiunik]{cattaneo2022uncertainty}
Matias~D Cattaneo, Yingjie Feng, Filippo Palomba, and Rocio Titiunik.
\newblock Uncertainty quantification in synthetic controls with staggered
  treatment adoption.
\newblock \emph{arXiv preprint arXiv:2210.05026}, 2022.

\bibitem[Cavallo et~al.(2013)Cavallo, Galiani, Noy, and
  Pantano]{cavallo2013catastrophic}
Eduardo Cavallo, Sebastian Galiani, Ilan Noy, and Juan Pantano.
\newblock Catastrophic natural disasters and economic growth.
\newblock \emph{Review of Economics and Statistics}, 95\penalty0 (5):\penalty0
  1549--1561, 2013.

\bibitem[Chamberlain(1982)]{chamberlain1982multivariate}
Gary Chamberlain.
\newblock Multivariate regression models for panel data.
\newblock \emph{Journal of econometrics}, 18\penalty0 (1):\penalty0 5--46,
  1982.

\bibitem[Chamberlain(1984)]{chamberlain1984panel}
Gary Chamberlain.
\newblock Panel data.
\newblock \emph{Handbook of econometrics}, 2:\penalty0 1247--1318, 1984.

\bibitem[Chernozhukov et~al.(2018)Chernozhukov, Chetverikov, Demirer, Duflo,
  Hansen, Newey, and Robins]{chernozhukov2018double}
Victor Chernozhukov, Denis Chetverikov, Mert Demirer, Esther Duflo, Christian
  Hansen, Whitney Newey, and James Robins.
\newblock Double/debiased machine learning for treatment and structural
  parameters.
\newblock \emph{The Econometrics Journal}, 21\penalty0 (1):\penalty0 C1--C68,
  2018.

\bibitem[Chernozhukov et~al.(2021)Chernozhukov, W{\"u}thrich, and
  Zhu]{chernozhukov2021exact}
Victor Chernozhukov, Kaspar W{\"u}thrich, and Yinchu Zhu.
\newblock An exact and robust conformal inference method for counterfactual and
  synthetic controls.
\newblock \emph{Journal of the American Statistical Association}, 116\penalty0
  (536):\penalty0 1849--1864, 2021.

\bibitem[Currie et~al.(2020)Currie, Kleven, and Zwiers]{currie2020technology}
Janet Currie, Henrik Kleven, and Esm{\'e}e Zwiers.
\newblock Technology and big data are changing economics: Mining text to track
  methods.
\newblock In \emph{AEA Papers and Proceedings}, volume 110, pages 42--48.
  American Economic Association 2014 Broadway, Suite 305, Nashville, TN 37203,
  2020.

\bibitem[De~Chaisemartin and d’Haultfoeuille(2020)]{de2020two}
Cl{\'e}ment De~Chaisemartin and Xavier d’Haultfoeuille.
\newblock Two-way fixed effects estimators with heterogeneous treatment
  effects.
\newblock \emph{American Economic Review}, 110\penalty0 (9):\penalty0
  2964--2996, 2020.

\bibitem[Doudchenko and Imbens(2016)]{doudchenko2016balancing}
Nikolay Doudchenko and Guido~W Imbens.
\newblock Balancing, regression, difference-in-differences and synthetic
  control methods: A synthesis.
\newblock Technical report, National Bureau of Economic Research, 2016.

\bibitem[Ferman and Pinto(2021)]{ferman2021synthetic}
Bruno Ferman and Cristine Pinto.
\newblock Synthetic controls with imperfect pretreatment fit.
\newblock \emph{Quantitative Economics}, 12\penalty0 (4):\penalty0 1197--1221,
  2021.

\bibitem[Freeman and Weidner(2023)]{freeman2023linear}
Hugo Freeman and Martin Weidner.
\newblock Linear panel regressions with two-way unobserved heterogeneity.
\newblock \emph{Journal of Econometrics}, 237\penalty0 (1):\penalty0 105498,
  2023.

\bibitem[Ghanem et~al.(2022)Ghanem, Sant'Anna, and
  W{\"u}thrich]{ghanem2022selection}
Dalia Ghanem, Pedro~HC Sant'Anna, and Kaspar W{\"u}thrich.
\newblock Selection and parallel trends.
\newblock \emph{arXiv preprint arXiv:2203.09001}, 2022.

\bibitem[Gin{\'e} and Nickl(2021)]{gine2021mathematical}
Evarist Gin{\'e} and Richard Nickl.
\newblock \emph{Mathematical foundations of infinite-dimensional statistical
  models}.
\newblock Cambridge university press, 2021.

\bibitem[Goodman-Bacon(2021)]{goodman2021difference}
Andrew Goodman-Bacon.
\newblock Difference-in-differences with variation in treatment timing.
\newblock \emph{Journal of Econometrics}, 225\penalty0 (2):\penalty0 254--277,
  2021.

\bibitem[Graham et~al.(2012)Graham, de~Xavier~Pinto, and
  Egel]{graham2012inverse}
Bryan~S Graham, Cristine~Campos de~Xavier~Pinto, and Daniel Egel.
\newblock Inverse probability tilting for moment condition models with missing
  data.
\newblock \emph{The Review of Economic Studies}, 79\penalty0 (3):\penalty0
  1053--1079, 2012.

\bibitem[Hahn and Kuersteiner(2002)]{hahn2002asymptotically}
Jinyong Hahn and Guido Kuersteiner.
\newblock Asymptotically unbiased inference for a dynamic panel model with
  fixed effects when both n and t are large.
\newblock \emph{Econometrica}, 70\penalty0 (4):\penalty0 1639--1657, 2002.

\bibitem[Hainmueller(2012)]{hainmueller2012entropy}
Jens Hainmueller.
\newblock Entropy balancing for causal effects: A multivariate reweighting
  method to produce balanced samples in observational studies.
\newblock \emph{Political analysis}, 20\penalty0 (1):\penalty0 25--46, 2012.

\bibitem[Heckman and Navarro(2007)]{heckman2007dynamic}
James~J Heckman and Salvador Navarro.
\newblock Dynamic discrete choice and dynamic treatment effects.
\newblock \emph{Journal of Econometrics}, 136\penalty0 (2):\penalty0 341--396,
  2007.

\bibitem[Hirano et~al.(2003)Hirano, Imbens, and Ridder]{hirano2003efficient}
Keisuke Hirano, Guido~W Imbens, and Geert Ridder.
\newblock Efficient estimation of average treatment effects using the estimated
  propensity score.
\newblock \emph{Econometrica}, 71\penalty0 (4):\penalty0 1161--1189, 2003.

\bibitem[Hirshberg and Wager(2021)]{hirshberg2021augmented}
David~A Hirshberg and Stefan Wager.
\newblock Augmented minimax linear estimation.
\newblock \emph{The Annals of Statistics}, 49\penalty0 (6):\penalty0
  3206--3227, 2021.

\bibitem[Holtz-Eakin et~al.(1988)Holtz-Eakin, Newey, and
  Rosen]{holtz1988estimating}
Douglas Holtz-Eakin, Whitney Newey, and Harvey~S Rosen.
\newblock Estimating vector autoregressions with panel data.
\newblock \emph{Econometrica: Journal of the econometric society}, pages
  1371--1395, 1988.

\bibitem[Imai and Ratkovic(2014)]{imai2014covariate}
Kosuke Imai and Marc Ratkovic.
\newblock Covariate balancing propensity score.
\newblock \emph{Journal of the Royal Statistical Society: Series B (Statistical
  Methodology)}, 76\penalty0 (1):\penalty0 243--263, 2014.

\bibitem[Imbens and Rubin(2015)]{imbens2015causal}
Guido~W Imbens and Donald~B Rubin.
\newblock \emph{Causal Inference in Statistics, Social, and Biomedical
  Sciences}.
\newblock Cambridge University Press, 2015.

\bibitem[Imbens et~al.(1998)Imbens, Spady, and Johnson]{imbens1998information}
Guido~W Imbens, Richard~H Spady, and Phillip Johnson.
\newblock Information theoretic approaches to inference in moment condition
  models.
\newblock \emph{Econometrica}, 66\penalty0 (2):\penalty0 333--357, 1998.

\bibitem[Jones and Marinescu(2022)]{jones2022labor}
Damon Jones and Ioana Marinescu.
\newblock The labor market impacts of universal and permanent cash transfers:
  Evidence from the alaska permanent fund.
\newblock \emph{American Economic Journal: Economic Policy}, 14\penalty0
  (2):\penalty0 315--340, 2022.

\bibitem[Lecu{\'e} and Mendelson(2013)]{lecue2013learning}
Guillaume Lecu{\'e} and Shahar Mendelson.
\newblock Learning subgaussian classes: Upper and minimax bounds.
\newblock \emph{arXiv preprint arXiv:1305.4825}, 2013.

\bibitem[Mendelson(2014)]{mendelson2014learning}
Shahar Mendelson.
\newblock Learning without concentration.
\newblock In \emph{Conference on Learning Theory}, pages 25--39, 2014.

\bibitem[Mendelson(2015)]{mendelson2015learning}
Shahar Mendelson.
\newblock Learning without concentration.
\newblock \emph{Journal of the ACM (JACM)}, 62\penalty0 (3):\penalty0 1--25,
  2015.

\bibitem[Mendelson(2016)]{mendelson2016upper}
Shahar Mendelson.
\newblock Upper bounds on product and multiplier empirical processes.
\newblock \emph{Stochastic Processes and their Applications}, 126\penalty0
  (12):\penalty0 3652--3680, 2016.

\bibitem[Mitze et~al.(2020)Mitze, Kosfeld, Rode, and W{\"a}lde]{mitze2020face}
Timo Mitze, Reinhold Kosfeld, Johannes Rode, and Klaus W{\"a}lde.
\newblock Face masks considerably reduce covid-19 cases in germany.
\newblock \emph{Proceedings of the National Academy of Sciences}, 117\penalty0
  (51):\penalty0 32293--32301, 2020.

\bibitem[Moon and Weidner(2015)]{moon2015linear}
Hyungsik~Roger Moon and Martin Weidner.
\newblock Linear regression for panel with unknown number of factors as
  interactive fixed effects.
\newblock \emph{Econometrica}, 83\penalty0 (4):\penalty0 1543--1579, 2015.

\bibitem[Moon and Weidner(2017)]{moon2017dynamic}
Hyungsik~Roger Moon and Martin Weidner.
\newblock Dynamic linear panel regression models with interactive fixed
  effects.
\newblock \emph{Econometric Theory}, 33\penalty0 (1):\penalty0 158--195, 2017.

\bibitem[Neyman(1923/1990)]{neyman1923}
Jerzey Neyman.
\newblock On the application of probability theory to agricultural experiments.
  essay on principles. section 9.
\newblock \emph{Statistical Science}, 5\penalty0 (4):\penalty0 465--472,
  1923/1990.

\bibitem[Nickell(1981)]{nickell1981biases}
Stephen Nickell.
\newblock Biases in dynamic models with fixed effects.
\newblock \emph{Econometrica: Journal of the econometric society}, pages
  1417--1426, 1981.

\bibitem[Rambachan and Roth(2020)]{rambachan2020design}
Ashesh Rambachan and Jonathan Roth.
\newblock Design-based uncertainty for quasi-experiments.
\newblock \emph{arXiv preprint arXiv:2008.00602}, 2020.

\bibitem[Rambachan and Roth(2023)]{rambachan2023more}
Ashesh Rambachan and Jonathan Roth.
\newblock A more credible approach to parallel trends.
\newblock \emph{Review of Economic Studies}, page rdad018, 2023.

\bibitem[Robins et~al.(2000)Robins, Hernan, and Brumback]{robins2000marginal}
James~M Robins, Miguel~Angel Hernan, and Babette Brumback.
\newblock Marginal structural models and causal inference in epidemiology.
\newblock \emph{Epidemiology}, pages 550--560, 2000.

\bibitem[Roth(2022)]{roth2022pretest}
Jonathan Roth.
\newblock Pretest with caution: Event-study estimates after testing for
  parallel trends.
\newblock \emph{American Economic Review: Insights}, 4\penalty0 (3):\penalty0
  305--322, 2022.

\bibitem[Rubin(1974)]{rubin1974estimating}
Donald~B Rubin.
\newblock Estimating causal effects of treatments in randomized and
  nonrandomized studies.
\newblock \emph{Journal of Educational Psychology}, 66\penalty0 (5):\penalty0
  688--701, 1974.

\bibitem[Schennach(2007)]{schennach2007point}
Susanne~M Schennach.
\newblock Point estimation with exponentially tilted empirical likelihood.
\newblock \emph{The Annals of Statistics}, pages 634--672, 2007.

\bibitem[Sun and Abraham(2021)]{sun2021estimating}
Liyang Sun and Sarah Abraham.
\newblock Estimating dynamic treatment effects in event studies with
  heterogeneous treatment effects.
\newblock \emph{Journal of Econometrics}, 225\penalty0 (2):\penalty0 175--199,
  2021.

\bibitem[Tan(2020)]{tan2020model}
Zhiqiang Tan.
\newblock Model-assisted inference for treatment effects using regularized
  calibrated estimation with high-dimensional data.
\newblock \emph{The Annals of Statistics}, 48\penalty0 (2):\penalty0 811--837,
  2020.

\bibitem[Vershynin(2018)]{vershynin2018high}
Roman Vershynin.
\newblock \emph{High-dimensional probability: An introduction with applications
  in data science}, volume~47.
\newblock Cambridge university press, 2018.

\bibitem[Viviano and Bradic(2021)]{viviano2021dynamic}
Davide Viviano and Jelena Bradic.
\newblock Dynamic covariate balancing: estimating treatment effects over time.
\newblock \emph{arXiv preprint arXiv:2103.01280}, 2021.

\bibitem[Wang and Zubizarreta(2020)]{wang2020minimal}
Yixin Wang and Jose~R Zubizarreta.
\newblock Minimal dispersion approximately balancing weights: asymptotic
  properties and practical considerations.
\newblock \emph{Biometrika}, 107\penalty0 (1):\penalty0 93--105, 2020.

\bibitem[Wooldridge(2010)]{wooldridge2010econometric}
Jeffrey~M Wooldridge.
\newblock \emph{Econometric analysis of cross section and panel data}.
\newblock MIT press, 2010.

\bibitem[Xiong et~al.(2023)Xiong, Athey, Bayati, and Imbens]{xiong2023optimal}
Ruoxuan Xiong, Susan Athey, Mohsen Bayati, and Guido Imbens.
\newblock Optimal experimental design for staggered rollouts.
\newblock \emph{Management Science}, 2023.

\bibitem[Zubizarreta(2015)]{zubizarreta2015stable}
Jos{\'e}~R Zubizarreta.
\newblock Stable weights that balance covariates for estimation with incomplete
  outcome data.
\newblock \emph{Journal of the American Statistical Association}, 110\penalty0
  (511):\penalty0 910--922, 2015.

\end{thebibliography}

\newpage
\appendix
\footnotesize
\begin{center}
    \Large For Online Publication
\end{center}
\section{Discussion}\label{ap:disc}
In this section, we informally discuss the intuition behind our results. The formal proofs are collected in Appendix \ref{ap:proofs}. To establish the properties of $\hat \tau$, we consider an alternative representation of $\hat\omega$ that follows from applying Fenchel duality to (\ref{eq:pr_problem}).
\begin{lemma}\label{lem:dual}
Define
\begin{equation*}
       \hlo :=\argmin_{f \in \spn{ \mathbf{1},\m{F}}} \Bigl\{ 
    \Pn\left((1-D_i)\exp(f_i) -
   \Pn D_if_i\right) +  \frac{\overline{\pi}\zeta^2}{2n}\|f\|^2_{\m{F}} \Bigl\} .
\end{equation*}
Then $(1-D_i)\hat \omega_i=\frac{(1-D_i)}{\overline{\pi}} \exp(\hlo_i)$.
\end{lemma}
Lemma \ref{lem:dual} is not new, e.g., \citet{hainmueller2012entropy} uses a similar result (see also \citet{wang2020minimal} for a more general formulation), and we state it to provide an important intuition for the weights that solve (\ref{eq:pr_problem}). In particular, we see that $\hlo$ solves an empirical analog of the problem for $\plo$ and thus we can expect:
\begin{equation*}
    (1-D_i) \hat \omega_i \approx \frac{(1-D_i)\exp(\plo_i)}{\E[\pi]} = \frac{1-D_i}{1-\pi_i}\frac{\pi_i \exp(\plo_i - \lo_i)}{\E[\pi]}.
\end{equation*}
If $\plo$ was equal to $\lo$, then $(1-D_i) \hat \omega_i$ would have converged to $\frac{1-D_i}{1-\pi_i}\frac{\pi_i}{\E[\pi]}$. These weights have an important balancing property: for any bounded function $f(X,\eta)$ we have
\begin{equation*}
    \E\left[\frac{1-D}{1-\pi}\frac{\pi}{\E[\pi]}f(X, \eta)\right] = \E[f(X, \eta)|D= 1].
\end{equation*}
In particular, under Assumption \ref{as:sel}, these weights balance all systematic differences between the treated and control groups.

However, our assumptions do not guarantee that $\plo = \lo$, even approximately. Instead, we use a different property of the weights $\frac{1-D_i}{1-\pi_i}\frac{\pi_i \exp(\plo_i - \lo_i)}{\E[\pi]}$, namely that they balance any function in $\spn{ \mathbf{1},\m{F}}$. Formally, ignoring regularization, for any $f \in \spn{ \mathbf{1},\m{F}}$ we have:
\begin{equation}\label{eq:bal_obs}
    \E\left[\frac{1-D}{1-\pi}\frac{\pi\exp(\plo - \lo)}{\E[\pi]}f\right] = \E[f|D= 1].
\end{equation}
From Assumption \ref{as:emp_set} we can expect this to guarantee
\begin{equation*}
     \E\left[\frac{1-D}{1-\pi}\frac{\pi\exp(\plo - \lo)}{\E[\pi]}\mu\right] \approx \E[\mu|D= 1],
\end{equation*}
with the approximation becoming better as $T_0$ goes to infinity. A naive analysis suggests that the approximation error should behave as $\min_{f\in \spn{1,\m{F}}}\|f - \mu\|_2$:
\begin{multline}\label{eq:pes_bound}
     \left|\E\left[\frac{1-D}{1-\pi}\frac{\pi\exp(\plo - \lo)}{\E[\pi]}\mu\right] - \E[\mu|D= 1]\right| = \left|\E\left[\left(\frac{\pi\exp(\plo - \lo)}{\E[\pi]} - \frac{\pi}{\E[\pi]}\right)\mu\right]\right| =\\
     \left|\E\left[\left(\frac{\pi\exp(\plo - \lo)}{\E[\pi]} - \frac{\pi}{\E[\pi]}\right)(\mu-f)\right]\right| \le \| \mu-f\|_2\left\|\frac{\pi\exp(\plo - \lo)}{\E[\pi]} - \frac{\pi}{\E[\pi]}\right\|_2,
\end{multline}
where $f \in \spn{1,\m{F}}$ and we used (\ref{eq:bal_obs}) to get the second equality. 

By definition $\min_{f\in \spn{1,\m{F}}}\|f - \mu\|_2=\frac{1}{\sqrt{\Tef}}$ and \eqref{eq:pes_bound} implies the same upper bound for the SC estimator. To improve over this pessimistic bound, we again ignore regularization and use the balancing property of $\muplo$:
\begin{equation*}
\E\qty[\left(\frac{\pi \exp(\muplo - \lo)}{\E[\pi]}-\frac{\pi}{\E[\pi]}\right)(\mu - f)] = 0.
\end{equation*}
Subtracting this equation from the last equality in \eqref{eq:pes_bound} and using $f = \pmu \in \spn{1,\m{F}}$ we get
\begin{multline*}
     \left|\E\left[\frac{1-D}{1-\pi}\frac{\pi\exp(\plo - \lo)}{\E[\pi]}\mu\right] - \E[\mu|D= 1]\right| = \\
     \left|\E\qty[\left(\frac{\pi\exp(\plo - \lo)}{\E[\pi]} - \frac{\pi \exp(\muplo - \lo)}{\E[\pi]}\right)(\mu - \pmu)]\right| = \\
      \left|\E\qty[\frac{\pi \exp(\muplo - \lo)}{\E[\pi]}\left(\exp(\plo -\muplo) - 1\right)(\mu - \pmu)]\right| \le\\
      \left|\E\qty[\frac{\pi \exp(\muplo - \lo)}{\E[\pi]}(\plo -\muplo)(\mu - \pmu)]\right| + \left|\E\qty[\frac{\pi \exp(\muplo - \lo)}{\E[\pi]}\ell(\muplo-\plo)(\mu - \pmu)]\right|
\end{multline*}
The first part of the last inequality is the population analog of the bias term defined in \eqref{eq:bias}. The second part arises from the nonlinearity of $\ell$ and will be asymptotically negligible compared to the first one.

%\newpage
%\section{Inference Results}\label{ap:inference}

\newpage
\section{Proofs}\label{ap:proofs}

\subsection{Definitions}\label{ap:defs}
By construction, the estimator \eqref{eq:estimator} is invariant with respect to arbitrary constant shifts of functions $f\in \m{F}$. As a result, without loss of generality, we will assume that $\mathbb{E}[f] = 0$ for any $f\in \m{F}$. It then follows that for $f\in \spn{\m{F}}$ we have:
\begin{equation*}
\| f\|_{\m{F}} = \min_{\alpha,\beta}\| f - \alpha - \beta \mu\|_{\m{F}}.
\end{equation*}
as long as $\mu \not\in \spn{\mathbf{1},\m{F}}$. To see this, observe that by construction for any $\alpha, \beta \ne 0$ we have $ f - \alpha - \beta \mu \not \in \spn{\m{F}}$, and by definition of the gauge function we have $\infty = \|f - \alpha - \beta \mu\|_{\m{F}}> \| f\|_{\m{F}}$. As a result, we can extend this norm to a seminorm on $f \in \spn{\mu, \mathbf{1}, \m{F}}$ by defining:
\begin{equation*}
\| f\|_{\m{F}} = \min_{\alpha,\beta}\| f - \alpha - \beta \mu\|_{\m{F}}.
\end{equation*}
We use $<\cdot,\cdot>_{\m{F}}$ to denote the semi-inner product this norm induces. 
%\dmitry{Clarify why this construction actually gives us a semi-inner product. Only do it if you have time.}

Define $\mt{F} := \spn{\mathbf{1},\m{F}}$; we list the definitions of functions that we will use in the proof:
\begin{equation}
\begin{aligned}
    &\hlo :=\argmin_{f \in \mt{F}} \Bigl\{ 
    \Pn\left((1-D_i)\exp(f_i) -
   \Pn D_if_i\right) +  \frac{\overline{\pi}\zeta^2}{2n}\|f\|^2_{\m{F}} \Bigl\}\\
     &  \plo := \argmin_{f \in \mt{F}}  \Bigl\{
    \E\left[\pi\ell(\lo - f)\right] + \frac{\E[\pi]\zeta^2}{2n}\|f\|^2_{\m{F}}\Bigl\}\\
 &\muplo := \argmin_{f \in \spn{ \mt{F}, \mu}} \Bigl\{
    \E\left[\pi\ell(\lo - f)\right]+\frac{\E[\pi]\zeta^2}{2n}\|f\|^2_{\m{F}}\Bigl\}\\
    &\bestmu:= \argmin_{f\in \mt{F}} \| f - \mu\|_2\\
    &\pmu := \argmin_{f\in \mt{F}} \E\left[\frac{\pi}{\E[\pi]}\exp(\muplo -\lo)(\mu - f)^2\right]\\
    &\hmu:= \argmin_{f \in \mt{F}} \left\{\Pn\frac{\pi_i}{\E[\pi]}\exp(\muplo_i -\lo_i)(\mu_i - f_i)^2\right\}
\end{aligned}
\end{equation}

We also define the error terms:
\begin{equation}
    u:= \exp(\muplo - \lo) - 1, \quad
     \nu_{\muplo} := \plo-\muplo, \quad
     \nu_{\mu} := \mu - \pmu, \quad \epsilon:= Y(0) - \mu.
\end{equation}

\subsection{First-order conditions}\label{sec:ap_foc}
We collect the first-order conditions for different objects. We have for any $f \in \mt{F}$ for the population problem for $\pmu$:
\begin{equation}\label{eq:foc_pmu}
    \E\left[\frac{\pi}{\E[\pi]}\exp(\muplo -\lo)\nu_{\mu}f\right] = 0,
\end{equation}
and analogously for the empirical problem for $\hmu$:
\begin{equation}\label{eq:foc_hmu}
    \Pn\frac{\pi_i}{\E[\pi]}\exp(\muplo_i -\lo_i)(\mu_i-\hmu)f_i = 0.
\end{equation}
For any $f\in \spn{\mt{F}, \mu}$ we have from the FOC for $\muplo$:
\begin{equation}\label{eq:foc_muplo}
    \E\left[\pi uf\right] =  - \frac{\E[\pi]\zeta^2}{n}<\muplo,f>_{\m{F}},
\end{equation}
and for any $f \in \mt{F}$ we have from the FOC for $\plo$:
\begin{equation}\label{eq:foc_plo}
    \E\left[\pi(\exp(\plo - \lo) - 1)f\right] =  - \frac{\E[\pi]\zeta^2}{n}<\plo,f>_{\m{F}}.
\end{equation}
Subtacting \eqref{eq:foc_muplo} from \eqref{eq:foc_plo} we get for any $f \in \mt{F}$:
\begin{multline}\label{eq:alt_foc_plo}
     \E\left[\pi\left(\exp(\plo-\lo) - \exp(\muplo - \lo)\right)f\right] = \E\left[\pi\exp(\muplo - \lo)\left(\exp(\nu_{\muplo}) - 1\right)f\right] = \\
     -  \frac{\E[\pi]\zeta^2}{n}<\plo -\muplo,f>_{\m{F}},
\end{multline}
which is an alternative FOC for $\plo$.

We use \eqref{eq:alt_foc_plo} to derive an equivalent definition for $\plo$:
\begin{equation}\label{eq:def_plo}
\plo := \argmin_{f \in \mt{F}}  \Bigl\{
    \E\left[\pi\exp(\muplo - \lo)\ell(\muplo - f)\right] + \frac{\E[\pi]\zeta^2}{2n}\|f-\muplo\|^2_{\m{F}}\Bigl\}
\end{equation}
To see why this is correct, observe that
\begin{equation}
\begin{aligned}
\label{eq:quick_der}
&\E\left[\pi\ell(\lo - f)\right]+\frac{\E[\pi]\zeta^2}{2n}\|f\|^2_{\m{F}} -  \E\left[\pi\ell(\lo - \muplo)\right]-\frac{\E[\pi]\zeta^2}{2n}\|\muplo\|^2_{\m{F}} \\
&= \E\left[\pi\exp(\muplo - \lo)\ell(\muplo - f)\right] + \frac{\E[\pi]\zeta^2}{2n}\|f-\muplo\|^2_{\m{F}} + \E\left[\pi u(f - \muplo)\right] +  \frac{\E[\pi]\zeta^2}{n}<\muplo,f - \muplo>_{\m{F}} \\
&=  \E\left[\pi\exp(\muplo - \lo)\ell(\muplo - f)\right] + \frac{\E[\pi]\zeta^2}{2n}\|f-\muplo\|^2_{\m{F}},
\end{aligned}
\end{equation}
Here the second equality holds because, taking $f=f-\muplo$ in \eqref{eq:foc_muplo}, we see the last two terms in the penultimate expression are equal and opposite. To show that the first holds, we use this elementary identity to combine the penalty terms
\begin{equation*}
    \| f\|_{\m{F}}^2 - \| \muplo\|_{\m{F}}^2 = \| f - \muplo\|_{\m{F}}^2 + 2< f - \muplo, \muplo>_{\m{F}}
\end{equation*}
and this arithmetic to combine the loss terms
\begin{multline*}
    \ell(\lo - f) - \ell(\lo - \muplo) = \exp(f - \lo) - \exp(\muplo - \lo) - (f - \muplo) =\\
    \exp(f- \lo) - \exp(\muplo - \lo) -  \exp(\muplo - \lo)(f- \muplo) + \exp(\muplo - \lo)(f - \muplo) -  (f - \muplo) = \\
    \exp(\muplo - \lo)(\exp(f - \muplo) - (f - \muplo) - 1) + (\exp(\muplo - \lo)-1)(f - \muplo)  = \\
    \exp(\muplo - \lo) \ell(\muplo - f) + u (f- \muplo).
\end{multline*}

 \newpage
\subsection{Properties of the population objects}\label{ap:pop_objects}
\begin{lemma}\label{lem:init_bound}
    Suppose Assumptions \ref{as:subg} - \ref{as:overlap} hold, then there exists an absolute constant $\tilde L_{\muplo}$ such that $\| \muplo-\lo\|_2 \le  \tilde L_{\muplo}$.
\end{lemma}
\begin{proof}
By Assumption \ref{as:subg} we have $\| \lo-\muplo\|_2 < \infty$, then using Assumption \ref{as:overlap} and $\ell(x) \ge \ell(|x|)$ we get for any $\epsilon \in (0, \epsilon_0]$:
\begin{align*}
    &\frac{\E[\pi \ell(\lo - \muplo)]}{\E[\pi]}  \ge  \frac{\E[\pi \ell(|\lo - \muplo|)]}{\E[\pi]} = \frac{\E\left[\pi \ell\left(\|\lo - \muplo\|_2\frac{|\lo - \muplo|}{\|\lo - \muplo\|_2}\right)\right]}{\E[\pi]} \ge \\
    &\ell\left(\frac{\|\lo - \muplo\|_2}{2}\right)\frac{\E\left[\pi \{2|\lo - \muplo| \ge \|\lo - \muplo\|_2\} \right]}{\E[\pi]} \ge\\
    &\ell\left(\frac{\|\lo - \muplo\|_2}{2}\right)q_{\epsilon}\E\{\pi \ge q_{\epsilon}\E[\pi]\}\{2|\lo - \muplo| \ge \|\lo - \muplo\|_2\}  \ge \\
    &\ell\left(\frac{\|\lo - \muplo\|_2}{2}\right)q_{\epsilon}\left(\E\{\pi \ge q_{\epsilon}\E[\pi]\} + \E\{2|\lo - \muplo| \ge \|\lo - \muplo\|_2\} -1\right)  \ge \\
     &\ell\left(\frac{\|\lo - \muplo\|_2}{2}\right)q_{\epsilon}\left(\E\{2|\lo - \muplo| \ge \|\lo - \muplo\|_2\} -\epsilon\right). 
\end{align*}
    By Assumption \ref{as:subg} we have $\E\{2|\lo - \muplo| \ge \|\lo - \muplo\|_2\} > c^{\star}>0$, where $c^{\star}$ is an absolute constant. Choosing $\epsilon^{\star} = \min\left\{\epsilon_0, \frac{c^{\star}}{2}\right\}$ we get from Assumption \ref{as:miss_lo}:
    \begin{equation*}
       L_{\muplo} \ge \frac{\E[\pi \ell(\lo - \muplo)]}{\E[\pi]} \ge \ell\left(\frac{\|\lo - \muplo\|_2}{2}\right)\frac{q_{\epsilon^{\star}}c^{\star}}{2} \Rightarrow  2\ell_{+}^{-1}\left(\frac{2L_{\muplo}}{q_{\epsilon^{\star}}c^{\star}}\right) \ge \|\lo - \muplo\|_2,
    \end{equation*}
    where $\ell_{+}^{-1}$ is the inverse of $\ell(|x|)$. The result follows by defining  $\tilde L_{\muplo} :=2\ell_{+}^{-1}\left(\frac{2L_{\muplo}}{q_{\epsilon^{\star}}c^{\star}}\right)$
\end{proof}

The next two lemmas connect the weighted loss functions we used to construct $\plo$ and $\pmu$ with the usual $L^2$ norm.
\begin{lemma}\label{lem:l_2_eq}
Suppose Assumptions \ref{as:subg} - \ref{as:overlap} holds, then we have for any $x$ such that $\|x\|_{\psi_2} \le C$
\begin{equation*}
    C_1 \|x\|_2^2 \le \E\left[\frac{\pi\exp(\muplo -\lo) }{\E[\pi]}x^2\right] \le C_2 \|x\|_2^2.
\end{equation*}
\end{lemma}
\begin{proof}
We start with the upper bound:
    \begin{align*}
        \E\left[\frac{\pi\exp(\muplo -\lo) }{\E[\pi]}x^2\right]  \le 
        \left\|\frac{\pi\exp(\muplo -\lo) }{\E[\pi]}\right\|_2\|x\|_4^2  \le
        C \left\|\frac{\pi\exp(\muplo -\lo) }{\E[\pi]}\right\|_2\|x\|_2^2,
    \end{align*}
where the last inequality follows by subgaussianity of $x$, and $\left\|\frac{\pi\exp(\muplo -\lo) }{\E[\pi]}\right\|_2$ is bounded by Corollary \ref{cor:mom_prod}.\footnote{The proof of Corollary \ref{cor:mom_prod} for $\lambda_4 = 0$ does not depend on the current result.}

To prove the lower bound we consider the following set of inequalities for $\epsilon \in (0, \epsilon_0]$ and $x > 0$:
\begin{align*}
&\E\left[\frac{\pi\exp(\muplo -\lo) }{\E[\pi]}x^2\right] = \|x\|_2^2 \E\left[\frac{\pi\exp(\muplo -\lo) }{\E[\pi]}\left(\frac{|x|}{\|x\|_2}\right)^2\right] \ge\\
&\frac{\|x\|_2^2q_{\epsilon}\exp( -x \| \muplo -\lo\|_{\psi_2})}{4} \E\{ \pi \ge q_{\epsilon} \E[\pi]\} \{2 |x| \ge \|x\|_2\}\{\muplo -\lo \ge -x \| \muplo -\lo\|_{\psi_2}\} = \\
&\frac{\|x\|_2^2q_{\epsilon}\exp( -x \| \muplo -\lo\|_{\psi_2})}{4}\Bigl ( \E\{ \pi \ge q_{\epsilon} \E[\pi]\} \{2 |x| \ge \|x\|_2\} - \\
&\E\{ \pi \ge q_{\epsilon} \E[\pi]\} \{2 |x| \ge \|x\|_2\}\{\muplo -\lo \le -x \| \muplo -\lo\|_{\psi_2}\}\Bigr) \ge \\
&\frac{\|x\|_2^2q_{\epsilon}\exp( -x \| \muplo -\lo\|_{\psi_2})}{4}\Bigl ( \E\{ \pi \ge q_{\epsilon} \E[\pi]\} \{2 |x| \ge \|x\|_2\} - 
\E\{\muplo -\lo \le -x \| \muplo -\lo\|_{\psi_2}\}\Bigr).
\end{align*}
We then choose the same $\epsilon^{\star}$ as in Lemma \ref{lem:init_bound}, use the result of Lemma \ref{lem:init_bound} and Assumption \ref{as:subg} to guarantee $\| \muplo -\lo\|_{\psi_2}$ is bounded, and then choose $x$ large enough so that $\E\{\muplo -\lo \le -x \| \muplo -\lo\|_{\psi_2}\} = \frac{c^{\star}}{4}$, which proves the result. 
\end{proof}

\begin{lemma}\label{lem:ell_eq}
    Suppose Assumptions \ref{as:subg} - \ref{as:overlap} hold, then for any $x$ such that $\|x\|_{\psi_2} \le C\|x\|_2$ and $\|x\|_{2} = o(1)$ we have
    \begin{equation*}
       C_1 \|x\|_2^2 \le \E\left[\frac{\pi \exp(\muplo -\lo)}{\E[\pi]}\ell(x)\right] \le C_2 \|x\|_2^2
    \end{equation*}
\end{lemma}
\begin{proof}
For the lower bound, we use the same argument as in the previous lemma and the fact that $\ell(x)\ge \ell(|x|)$ is quadratic in the neighborhood of zero. For the upper bound, we have the following:
\begin{equation*}
    \E\left[\frac{\pi \exp(\muplo -\lo)}{\E[\pi]}\ell(x)\right]\le \sqrt{\E\left[\frac{\pi\exp(\muplo -\lo) }{\E[\pi]}\right]^2}\sqrt{\E\left[\ell^2(x)\right]}
\end{equation*}
The first multiplier is bounded by Corollary \ref{cor:mom_prod}. For the second multiplier, we have by Taylor's theorem with MVT remainder $\ell(x) = \exp(-y)\frac{x^2}{2}$ for some $y$ between $0$ and $x$. 
It then follows:
\begin{multline*}
    \E[\ell^2(x)] = \E[\ell^2(x)\{x <0\}] + \E[\ell^2(x)\{x >0\}] \le \E\left[\frac{x^4}{4}\{x >0\}\right] + \E\left[\exp(-2x)\frac{x^{4}}{4}\{x <0\}\right] \le \\
   \|x\|_4^4 + \E\left[\exp(-2x)\frac{x^{4}}{4}\{x <0\}\right] \le O(\|x\|_2^4) + \E\left[\exp(-2x)\frac{x^{4}}{4}\{x <0\}\right].
\end{multline*}
For $x\ge 0$ define $f(x):=\exp(2x)\frac{x^{4}}{4}$, and observe that it is strictly monotone.  We then have:
\begin{align*}
    &\E\left[\exp(2x)\frac{x^{4}}{4}\{x >0\}\right] = \int_{0}^{\infty} \E\left\{\exp(2x)\frac{x^{4}}{4} > u\right\}du = 
    \int_{0}^{\infty} \E\left\{x > f^{-1}(u)\right\}du = \\
    &\int_{0}^{\infty} \E\left\{x > t\right\}f^{\prime}(t)dt = 
    \int_{0}^{\infty} \E\left\{x > t\right\}\left(2\exp(2t)\frac{t^4}{4} + \exp(2t)t^3\right)dt \le\\
     &2\int_{0}^{\infty} \exp\left(-C\frac{t^2}{\|x\|_{\psi_2}^2}\right)\left(2\exp(2t)\frac{t^4}{4} + \exp(2t)t^3\right)dt = \\
     &2\|x\|_{\psi_2} \int_{0}^{\infty} \exp\left(-Ct^2\right)\left(2\exp(2t\|x\|_{\psi_2})\frac{(\|x\|_{\psi_2}t)^4}{4} + \exp(2t\|x\|_{\psi_2})(\|x\|_{\psi_2}t)^3\right)dt \le \\
     &O \qty(\|x\|^5_{\psi_2} + \|x\|^4_{\psi_2}) \le O\qty(\|x\|_2^4 (1 + \|x\|_2)) = O(\|x\|_2^4)
\end{align*}
Where the first inequality follows from the subgaussianity of $x$, and the inequality follows because $\| x\|_{\psi_2}$ is bounded by assumption. 
\end{proof}

We use these three lemmas to establish the rate for $\nu_{\mu}$ and $\nu_{\muplo}$ from the approximation rate for $\mu$.
\begin{corollary}\label{cor:pop_proj}
Suppose Assumption \ref{as:emp_set} and Assumptions \ref{as:subg} - \ref{as:overlap} hold, then we have:
\begin{equation*}
    \begin{aligned}
        &\E\left[\frac{\pi \exp(\muplo -\lo)}{\E[\pi]} \nu_{\mu}^2\right] = O(\|\mu - \bestmu\|_2^2), \quad \|\nu_{\mu}\|_2^2 = O(\|\mu - \bestmu\|_2^2),\\
          &\E\left[\frac{\pi \exp(\muplo -\lo)}{\E[\pi]} \ell( -\nu_{\muplo}) \right] = O\qty(\|\mu - \bestmu\|_2^2 +  \frac{\zeta^2}{n}), \quad \|\nu_{\muplo}\|_2^2 = O\qty(\|\mu - \bestmu\|_2^2+\frac{\zeta^2}{n}),\\
              &\frac{\zeta^2}{2n}\| \muplo - \plo\|^2_{\m{F}} = O\qty(\|\mu - \bestmu\|_2^2 +  \frac{\zeta^2}{n}).
    \end{aligned}
    \end{equation*}
\end{corollary}

\begin{proof}
To prove the first line we use optimality of $\nu_{\mu}$:
\begin{equation*}
    \E\left[\frac{\pi \exp(\muplo -\lo)}{\E[\pi]} \nu^2_{\mu}\right] \le \E\left[\frac{\pi \exp(\muplo -\lo)}{\E[\pi]} (\mu - \bestmu)^2\right]  = O(\|\mu - \bestmu\|_2^2),
\end{equation*}
where the last equality follows from Assumptions \ref{as:miss_lo}, \ref{as:emp_set}  and the Lemma \ref{lem:l_2_eq}. By the same lemma it follows that $\|\nu_{\mu}\|_2^2 = O(\|\mu - \bestmu\|_2^2)$. 

For the second line we use alternative definition of $\plo$ we derived in \eqref{eq:def_plo} together with Assumptions \ref{as:miss_lo}, using $f^{\star} := f_{\muplo} + \beta_u \bestmu\in\mt{F}$:
\begin{equation}\label{ap:new_bound}
\begin{aligned}
&\E\qty[\frac{\pi\exp(\muplo -\lo)}{\E[\pi]}\ell(\muplo - \plo)] +  \frac{\zeta^2}{n}\| \muplo-\plo\|^2_{\m{F}}  
\le \E\qty[\frac{\pi\exp(\muplo -\lo)}{\E[\pi]}\ell(\muplo - f^{\star})] + \frac{\zeta^2}{2n}\| \muplo - f^{\star}\|^2_{\m{F}}\\
&= \E\qty[\frac{\pi\exp(\muplo -\lo)}{\E[\pi]}\ell(\beta_\mu (\bestmu - \mu))] + \frac{\zeta^2 \beta_\mu^2}{2n}\|\bestmu - \mu\|^2_{\m{F}} \\ 
& =O(\beta_{\mu}^2\|\mu - \bestmu\|_2^2) +  \frac{\zeta^2 \beta_\mu^2}{2n}\|\bestmu\|^2_{\m{F}} = O\qty(\|\mu - \bestmu\|_2^2 +  \frac{\zeta^2}{n}).
\end{aligned}
\end{equation}
where the next to las equality follows from Lemma \ref{lem:ell_eq} and the fact that by definition $\|\bestmu - \mu\|^2_{\m{F}} =\|\bestmu\|^2_{\m{F}}$. The last equality follows from Assumption \ref{as:out_mom} which implies the equivalence betweeen the two norms and $\|\bestmu\|^2_{\m{F}} = O\qty(\|\bestmu-\E[\bestmu]\|^2_{2}) = O\qty(\mathbb{V}[Y(0)]) = O\qty(1)$. The same inequality implies
\begin{equation*}
    \frac{\zeta^2}{2n}\| \muplo - \plo\|^2_{\m{F}} = O\qty(\|\mu - \bestmu\|_2^2 +  \frac{\zeta^2}{n}).
\end{equation*}

Applying Lemma \ref{lem:ell_eq} in other direction we get $\|\nu_{\muplo}\|_2^2 = O\qty(\|\mu - \bestmu\|_2^2+\frac{\zeta^2}{n})$. 
\end{proof}
We use this result to establish the following corollary, which we will use multiple times throughout the proof.
\begin{corollary}\label{cor:mom_prod}
    Suppose Assumption \ref{as:emp_set} and Assumptions \ref{as:subg} - \ref{as:overlap} hold. Then for any $\lambda_1, \lambda_2,\lambda_3,\lambda_4 \ge 0$ and  $\lambda_1+ \lambda_2 \le 9$ we have:
    \begin{equation*}
        \E\left[\exp(\lambda_1 \lo) \pi^{\lambda_2}\exp(\lambda_3(\muplo - \lo))\exp(\lambda_4 \nu_{\muplo}) \right] = O(\E[\pi]^{\lambda_1+\lambda_2})
    \end{equation*}
\end{corollary}
\begin{proof}
By definition, we have:
\begin{equation*}
    \left(\frac{\pi}{\E[\pi]}\right)^{\lambda_2} = \left(\frac{\exp(\lo)}{\E[\pi](1+\exp(\lo)))}\right)^{\lambda_2} \le \exp(\lambda_2(\lo - \log(\E[\pi]))).
\end{equation*}
Using this inequalities together with H\"older's inequality, Assumptions \ref{as:subg} - \ref{as:overlap}, Lemma \ref{lem:init_bound}, and Corollary \ref{cor:pop_proj} we get:
    \begin{align*}
       &\frac{\E\left[\exp(\lambda_1 \lo) \pi^{\lambda_2}\exp(\lambda_3(\muplo - \lo))\exp(\lambda_4 \nu_{\muplo}) \right] }{\E[\pi]^{\lambda_1+\lambda_2}} = \\
       &\E\left[\exp(\lambda_1 (\lo-\log(\E\pi)))\left(\frac{\pi}{\E[\pi]}\right)^{\lambda_2}\exp(\lambda_3(\muplo - \lo))\exp(\lambda_4 \nu_{\muplo}) \right] \le \\
        &\E\left[\exp((\lambda_1 + \lambda_2) (\lo - \log(\E[\pi])))\exp(\lambda_3(\muplo - \lo))\exp(\lambda_4 \nu_{\muplo}) \right] \le\\
        &\left\|\exp((\lambda_1 + \lambda_2) (\lo - \log(\E[\pi]))) \right\|_{\frac{10}{\lambda_1 + \lambda_2}} \|\exp(\lambda_3(\muplo - \lo))\|_{\frac{20}{10 - \lambda_1 + \lambda_2}}\times 
        &\|\exp(\lambda_4\nu_{\muplo})\|_{\frac{20}{10 - \lambda_1 + \lambda_2}} =  O(1).
    \end{align*}
Note that we use an iterated version of H\"older's inequality,
\[ \norm{abc}_1 \le \norm{a}_p\norm{bc}_q \le \norm{a}_p \norm{b}_{2q}\norm{c}_{2q} \qfor 1/p+1/q=1, \]
taking $p=10/(\lambda_1+\lambda_2)$ and $q=10/(10-\lambda_1+\lambda_2)$.
\end{proof}

\newpage
\subsection{Properties of the oracle objects}\label{ap:mu_prop}
In what follows, we will routinely use bounds for $\pmu- \hmu$, which we establish in the next lemma.

\begin{lemma}
\label{lemma:mu-l2-convergence}
If the assumptions of Lemma~\ref{lem:or_proj} hold, $\norm{\hmu-\pmu}_2 = O_p(r)$ if $r$ satisfies the fixed point conditions
\begin{equation*}
\begin{aligned}
    \eta r &\ge \E \sup_{\substack{\delta \in \mu - \mt{F} \\ \norm{\delta} \le r}} \Pn \varepsilon_i \delta_i \qand
    \frac{\eta r^2}{\norm{\mu-\pmu}_2} &\ge \E \sup_{\substack{\delta \in \mu - \mt{F} \\ \norm{\delta} \le r}} \Pn g_i \delta_i.
    \end{aligned}
    \end{equation*}
    Here $\eta>0$ is a sufficiently small constant, $\varepsilon_1 \ldots \varepsilon_n$ are independent Rademacher RVs independent of the observed data, and $g_1 \ldots g_n$ an analogous sequence of standard normal RVs.
\end{lemma}

\begin{proof}
\textbf{Step 1. Finding an Inequality Characterizing $\pmu-\hmu$.}
Let $T_i = (\pi_i/E\pi_i) \exp(\muplo_i - \plo_i)$.
In this notation, $\hmu$ and $\pmu$ minimize $P_n T (\mu - m)^2$ and $P T (\mu-m)^2$
respectively over $m \in \mt{F}$. Because $\hmu$
minimizes $P_n T (\mu - m)^2$ over a set containing $\pmu$,
\begin{equation}
   0 \ge  P_n T (\mu - \hmu)^2 - P_n T (\mu - \pmu)^2 =    P_n T (\pmu - \hmu)^2 + 2 P_n T (\mu - \pmu) (\pmu - \hmu) 
\end{equation}
Here we've used the identity $a^2-b^2 = (a-b)^2 + 2b(a-b)$ with $a=\mu-\hmu$ and $b=\mu-\pmu$. Furthermore, because $\pmu$ is an orthogonal projection of $\mu$ onto the convex set $\mt{F}$, $P T (\mu - \pmu) (\pmu - m) \ge 0$ for all $m \in \mt{F}$ and therefore for $m=\hmu$. Subtracting $P T (\mu - \pmu) (\pmu - \hmu) \ge 0$ from the previous inequality, we get
\begin{equation}
0 \ge P_n T (\pmu - \hmu)^2 + 2 (P_n-P) T (\mu - \pmu) (\pmu - \hmu) 
\end{equation}

\textbf{Step 2. Reduction to the Sphere}
We've shown that $\hat\delta=\pmu-\hmu$ satisfies 
$P_n T \delta^2 + 2(P_n - P)T (\mu - \pmu) \delta \le 0$, and it follows that
$P_n T \hat\delta^2 < r^2$ if $P_n T \delta^2 + 2(P_n-P)T(\mu - \pmu) \delta > 0$ for all $\delta \in \pmu - \mt{F}$ with $\norm{\delta}_2 \ge r$. By a scaling argument, it suffices to show that this holds for all such $\delta$ with $\norm{\delta}_2 = r$. To see this, suppose $\norm{\delta}_2 \ge r$ and consider the point $\delta_r = r\delta/\sqrt{P T \delta^2}$ with $\norm{\delta_r}_2=r$.
As a result of the convexity of $\pmu - \mt{F}$, it's in $\mu - \mt{F}$, as it's on the segment between two points $\delta$ and $0$ in the set $\pmu - \mt{F}$. Furthermore,
$P_n T \delta_r^2 + 2(P_n-P)T(\mu - \pmu) \delta_r > 0$ implies $P_n T \delta_r^2 + 2(P_n-P)T (\mu - \pmu) \delta_r > 0$, as in terms of $s = r/\sqrt{P T \delta^2} \le 1$, we have
\begin{equation} 
P_n T \delta_r^2 + 2(P_n-P) T (\mu - \pmu)\delta_r =  s^2 P_n T \delta^2 
+ s \ 2(P_n-P) T (\mu - \pmu)\delta) \le s \qty{ P_n T \delta^2 + 2(P_n-P)T(\mu - \pmu)\delta} 
\end{equation}
with the inequality here holding because $s^2 \le s$ and $P_n T \delta^2$ is nonnegative. 

\textbf{Step 3. Characterizing the Sphere's Radius}
We conclude by finding a radius $r$ for which, with high probability, all $\delta \in \pmu - \mt{F}$ with $\norm{\delta}_2 = r$ satisfy $P_n T \delta^2 + 2(P_n-P) T (\mu - \pmu) \delta > 0$. In particular, we'll find one for which all such $\delta$ satisfy $P_n T \delta^2 \ge \eta r^2$ and $2(P_n-P)T (\mu-\pmu)\delta > -\eta r^2$ for some constant $\eta$. 

\textbf{The lower bound.}
The claimed lower bound holds with probability tending to one for $r$ satisfying our linear fixed point condition. This follows by substituting $\hat\mu$ for $\hat\mu^{(i)}$ into the argument used to establish a rate of convergence of $\hat\mu^{(i)}$ to $\pmu$ in Lemma~\ref{lem:or_proj}. In particular, the argument following \eqref{eq:loo-zeroth-order-mu}.

\textbf{The upper bound.}
By a multiplier inequality due to Mendelson \citep[Corollary 1.10]{mendelson2016upper}, 
\begin{equation}
\sup_\delta (P_n-\P) T (\mu-\pmu)\delta
\le c \norm{T (\mu-\pmu)}_{L_q} \E \sup_\delta \Pn g_i \delta_i
\end{equation}
where $q>2$, the supremum is implicitly taken over the set of $\delta$ considered, and $g_i$ are standard normals independent of the observed data. Finally, a variant of Lemma~\ref{lem:l_2_eq} appropriate to $q>2$
implies that $\norm{T(\mu-\pmu)}_{L_2} \le c \norm{\mu-\pmu}_2$. Consequently,  we have the claimed upper bound if $c\norm{\mu-\pmu}_{2} \E \sup_\delta \Pn g_i \delta_i \le \eta r^2$. 
\end{proof}

\begin{lemma}\label{lem:or_proj}
    Suppose $\mt{F}$ is convex and Assumption \ref{as:emp_set} and Assumptions \ref{as:subg}-\ref{as:overlap} and \ref{as:out_mom} hold. Then we have:
    \begin{equation*}
    \begin{aligned}
    \max_{i}|\mu_i-\hmu_i| = O_p(\max\qty{r, \|\mu - \pmu\|_2}\sqrt{\log(n)}) 
    \qqtext{ for $r$ satisfying } \eta r \ge \E \sup_{\substack{\delta \in \mu - \mt{F} \\ \norm{\delta} \le r}} \Pn \varepsilon_i \delta_i 
    \end{aligned}
    \end{equation*}
    where $\eta>0$ is sufficiently small and $\varepsilon_1 \ldots \varepsilon_n$ are independent Rademacher RVs independent of the observed data.
\end{lemma}
\begin{proof}
Let $\hmu_{i}^{(i)}$ be the LOO estimator with unit $i$ removed. From optimality for $\hmu$ and $\hmu^{(i)}$ we get:
\begin{equation}\label{eq:full_opt_mu}
    \Pn \frac{\pi_j\exp(\muplo_j -\lo_j)}{\E[\pi]}(\mu_j - \hmu_j)^2 \le  \Pn \frac{\pi_j\exp(\muplo_j -\lo_j)}{\E[\pi]}(\mu_j - \hmu_j^{(i)})^2
\end{equation}
Similarly, from optimality for $\hlo^{(i)}$ we get the opposite inequality:
\begin{equation}\label{eq:loo_opt_mu}
        \frac{1}{n}\sum_{j\ne i} \frac{\pi_j\exp(\muplo_j -\lo_j)}{\E[\pi]}(\mu_j - \hmu_j^{(i)})^2 \le \frac{1}{n}\sum_{j\ne i} \frac{\pi_j\exp(\muplo_j -\lo_j)}{\E[\pi]}(\mu_j - \hmu_j)^2
\end{equation}
Adding \eqref{eq:full_opt_mu} and \eqref{eq:loo_opt_mu} and eliminating the terms we get the following:
\begin{equation*}
    |\mu_i - \hmu_i| \le |\mu_i - \hmu_i^{(i)}| 
\end{equation*}
For each $i$ conditional on $\hmu^{(i)}$ each $\mu_i- \hmu_i^{(i)}$ is subgaussian with $\|\mu_i- \hmu_i^{(i)}\|_{\psi_2} \le L_{\psi_2} \|\mu_i- \hmu_i^{(i)}\|_{2}$ by Assumption \ref{as:subg}. Using Lemma~\ref{lem:max_in}, a maximal inequality for conditionally subgaussian variables, this implies the following bound.
\begin{equation*}
    \max_i|\mu_i -\hmu_i^{(i)}| = O_p\left( \max_i\|\mu_i- \hmu_i^{(i)}\|_2 \sqrt{\log(n)}\right).
\end{equation*}
To bound the norm $\|\mu_i- \hmu_i^{(i)}\|_2$, we compare the values of the leave-one-out loss at its minimizer $\mu^{(i)}$ and $\tilde\mu$. 
\begin{equation}
\begin{aligned}\label{eq:key_in_mu}
\frac{1}{n}\sum_{j\ne i} \frac{\pi_j\exp(\muplo_j -\lo_j)}{\E[\pi]}(\mu_j - \hmu_j^{(i)})^2 
&\le \frac{1}{n}\sum_{j\ne i} \frac{\pi_j\exp(\muplo_j -\lo_j)}{\E[\pi]}(\mu_j - \pmu_j)^2  \\
   &\le \Pn \frac{\pi_j\exp(\muplo_j -\lo_j)}{\E[\pi]}(\mu_j - \pmu_j)^2 = O_p(\norm{\mu-\pmu}_2^2)
\end{aligned}
\end{equation}
where the last equivalence follows from \ref{cor:pop_proj} and Markov inequality. 
To conclude, we'll show that this weighted empirical norm of $\mu-\hmu^{(i)}$
is --- up to constant factors --- an upper bound on $\norm{\mu-\hmu^{(i)}}_2$.
To do this, we'll work with binary lower bounds of the form $X \ge \theta \{ X \ge \theta \}$ on the factors in this norm. For any $\epsilon < \epsilon_0$, and $x > 0$
\begin{equation}
\label{eq:loo-zeroth-order-mu}
\begin{aligned}
     \frac{1}{n}\sum_{j\ne i} \frac{\pi_j\exp(\muplo_j -\lo_j)}{\E[\pi]}(\mu_j - \hmu_j^{(i)})^2 \ge\frac{q_{\epsilon}\exp( -x\| \muplo -\lo\|_{\psi_2}) \| \mu - \hmu^{(i)}\|_2}{4}  \times\\
      \frac{1}{n}\sum_{j\ne i} \left\{ \pi _j \ge q_{\epsilon} \E[\pi]\right\} \left\{ \muplo_j -\lo_j \ge  -x\| \muplo -\lo\|_{\psi_2}\right\}\left\{ 4|\mu_j - \hmu_j^{(i)}| \ge  \| \mu - \hmu^{(i)}\|_2\right\} 
\end{aligned}
\end{equation}
Because our product of indicators is in $\{0,1\}$, we know that the missing $j=i$ term could contribute at most $1/n$ to this average; thus, writing $\Pn$ for the average over all $j \in 1 \ldots n$,
\begin{equation}
\label{eq:zeroth-order-mu}
\begin{aligned}
    \frac{1}{n}\sum_{j\ne i} \left\{ \pi _j \ge q_{\epsilon} \E[\pi]\right\} \left\{ \muplo_j -\lo_j \ge  -x\| \muplo -\lo\|_{\psi_2}\right\}\left\{ 4|\mu_j - \hmu_j^{(i)}| \ge  \| \mu - \hmu^{(i)}\|_2\right\} \ge \\
        \Pn \left\{ \pi _j \ge q_{\epsilon} \E[\pi]\right\} \left\{ \muplo_j -\lo_j \ge  -x\| \muplo -\lo\|_{\psi_2}\right\}\left\{ 4|\mu_j - \hmu_j^{(i)}| \ge  \| \mu - \hmu^{(i)}\|_2\right\} - \frac{1}{n}.
\end{aligned}
\end{equation}
Via some additional calculations, we arrive at a simplified lower bound.
\begin{align*}
&\Pn \left\{ \pi _j \ge q_{\epsilon} \E[\pi]\right\} \left\{ \muplo_j -\lo_j \ge  -x\| \muplo -\lo\|_{\psi_2}\right\}\left\{ 4|\mu_j - \hmu_j^{(i)}| \ge  \| \mu - \hmu^{(i)}\|_2\right\}  \\
&\overset{(a)}{\ge}\Pn \left\{ \pi _j \ge q_{\epsilon} \E[\pi]\right\}\left\{ 4|\mu_j - \hmu_j^{(i)}|  \ge  \| \mu - \hmu^{(i)}\|_2\right\} -  \Pn\left\{ \muplo_j -\lo_j \le  -x\| \muplo -\lo\|_{\psi_2}\right\}  \\
&\overset{(b)}{\ge} \Pn \left\{ \pi _j \ge q_{\epsilon} \E[\pi]\right\} - 1 + \Pn\left\{ 4|\mu_j - \hmu_j^{(i)}| \ge  \| \mu - \hmu^{(i)}\|_2\right\} -  \Pn\left\{ \muplo_j -\lo_j \le  -x\| \muplo -\lo\|_{\psi_2}\right\} \\
&\overset{(c)}{=}  \qty( \E \left\{ \pi \ge q_{\epsilon} \E[\pi]\right\} - 1 ) +
   \Pn\left\{ 4|\mu_j - \hmu_j^{(i)}| \ge  \| \mu - \hmu^{(i)}\|_2\right\} - \E\left\{ \muplo -\lo \le  -x\| \muplo -\lo\|_{\psi_2}\right\} + o_p(1) \\
&\overset{(d)}{\ge} -\epsilon +
   \Pn\left\{ 4|\mu_j - \hmu_j^{(i)}| \ge  \| \mu - \hmu^{(i)}\|_2\right\}- \E\left\{ \muplo -\lo \le  -x\| \muplo -\lo\|_{\psi_2}\right\} + o_p(1).
\end{align*}
Here (a) follows from the inclusion/exclusion identity $1_A 1_B = 1_A - 1_A 1_{\neg B} \ge 1_A - 1_{\neg B}$, (b) from 
the elementary inequality $1_A 1_B \ge 1_A + 1_B - 1$ (consider the cases that $1_A1_B =1$ and $1_A1_B=0$) (c) from the law of large numbers, and (d) from Assumption \ref{as:overlap}. 
We will argue that---with appropriate choice of $x$ and $\epsilon$---the essential behavior of this sum is the same as the term $\Pn\left\{ 4|\mu_j - \hmu_j^{(i)}| \ge  \| \mu - \hmu^{(i)}\|_2\right\}$, which we will now lower-bound.

It suffices to find a uniform lower bound on $\Pn \{ 4 \abs{\delta}_i \ge \norm{\delta}_2 \}$ for $\delta \in \mu - \mt{F}$, as $\mu - \hat\mu^{(i)}$ is --- for all $i$ --- in this set. We use one from \citet{mendelson2014learning}. With probability $1-2\exp(\eta^2n/2)$,\footnote{This is Theorem 5.3 of \citet{mendelson2014learning} for $\tau=1/4$, where we've taken $\eta$ to be the lower bound discussed below Theorem 4.1 for $Q_H(1/2)$ based on the $L_2$-$L_4$ norm equivalence $\norm{\delta}_4 \le c \norm{\delta}_2$ implied by our subgaussianity assumption (Assumption~\ref{as:subg}).}
\begin{equation}
\begin{aligned}
& \Pn \{ 4 \abs{\delta}_i \ge \norm{\delta}_2 \} \ge \eta/4 \qqtext{ for all } \delta \in \mu - \mt{F} \qqtext{ with } \norm{\delta}_2 \ge r \\
&\text{ if } \quad \frac{\eta r}{64} \ge \E \sup_{\substack{\delta \in \mu - \mt{F} \\ \norm{\delta} \le r}} \Pn \varepsilon_i \delta_i \qqtext{ where }
\eta = \text{the bound discussed in the footnote}
\end{aligned}
\end{equation}
Taking constants $x$ and $\epsilon$ appropriately,
this implies that $\norm{\mu - \hmu^{(i)}}_2 = O_p(\norm{\mu - \pmu}_2)$
when $\norm{\mu - \pmu}_2 \ge r$ and therefore, generally, that
$\norm{\mu - \hmu^{(i)}}_2 = O_p\qty(\max\qty{r, \norm{\mu - \pmu}_2})$.
\end{proof}

In the case that $\mt{F}$ is --- or is contained in --- a $p$-dimensional subspace, 
$cr\sqrt{p/n}$ bounds the gaussian and rademacher complexities involved in the fixed point conditions of the lemmas above. Thus, the linear fixed point condition is satisfied if 
$\eta r \ge c r \sqrt{p/n}$, i.e., irrespective of $r$ it holds if $p/n$ is bounded by a sufficiently small constant. And the quadratic one holds for $r \propto \sqrt{p/n} \norm{\mu-\pmu}_2$. This implies the following result.

\begin{corollary}
\label{coro:mu-rates}
Suppose $\mt{F}$ is a set of dimension no larger than $p$. Then 
if $p/n$ is bounded by a sufficiently small constant, we have
\begin{equation*}
\norm{\hmu-\pmu}_2 = O_p\qty(\norm{\mu-\pmu}_2\sqrt{\frac{p}{n}}) = O_p\qty(\norm{\mu-\pmu}) \qand
\max_{i}|\hmu_i-\mu_i| = O_p\qty(\norm{\mu-\pmu}_2\sqrt{\log(n)}).
\end{equation*}
\end{corollary}

\newpage
\subsection{Empirical convergence}\label{ap:weights_prop}

In this section, we establish convergence results for $\hlo$, which we later use to analyze parts of the error of the estimator. 
\begin{lemma}\label{lem:conv_rate_1}
Suppose Assumptions \ref{as:subg}-\ref{as:out_mom} hold, $\norm{\mu-\bestmu}_2 \ll 1$, $\m{F}$ is contained in a $p \ll n$ dimensional subspace,
and $\zeta^2 = O(\sqrt{pn}).$
Then,  
    \begin{equation*}
                \| \hlo - \plo\|_2 = O_p\qty(\sqrt{\frac{p}{\E[\pi]n}}), \quad  \Pn (\hlo_i - \plo_i)^2 = O_p\qty(\frac{p}{\E[\pi]n}), \quad \Pn \frac{(1-D_i) \exp(\plo_i)}{\E[\pi]} \ell(\hlo_i -  \plo_i) = O_p\qty(\frac{p}{\E[\pi]n}).
               % \qqtext{ for $r$ satisfying } \E \sup_{\norm{\delta} \le r} \Pn g_i \delta_i \le \sqrt{\E[\pi]} r^2 \qqtext{ i.e. } r = c\sqrt{p/ \E[\pi] n} \qqtext{ if } r \to 0.
    \end{equation*}
%And there should be another fixed point condition---a quadratic one---that we can derive from the upper bounded \eqref{eq:theta-uniform-ub-sufficient} using multiplier inequalities.
\end{lemma}

\begin{proof}
    From the definition of $\hlo$ we get
\begin{equation}
\label{eq:key_emp_inq_without_foc}
\begin{aligned}
0 &\ge \Pn\left((1-D_i)\exp(\hlo_i)   -
 D_i \hlo_i\right) + \frac{\overline\pi\zeta^2}{2n}\|\hlo\|^2_{\m{F}} - 
    \Pn\left((1-D_i)\exp( \plo_i)   -  D_i  \plo_i \right) - \frac{\overline\pi\zeta^2}{2n}\| \plo\|^2_{\m{F}} \\
&=\Pn(1-D_i)\exp\left( \plo_i \right)\ell(\plo_i-\hlo_i)
    + \Pn\left((1-D_i)\exp( \plo_i) - D_i \right)(\hlo_i -  \plo_i) \\
    &+ \frac{\overline\pi\zeta^2}{2n}\left(\|\hlo -\plo \|^2_{\m{F}} +  2<\hlo -\plo, \plo>_{\m{F}}\right). 
\end{aligned}
\end{equation}

The first order condition for $\plo$ is that, for any $f \in \mt{F}$,
\begin{equation}
    \E[\pi (\exp(\plo-\lo) - 1)(f-\plo)] +  \frac{\E[\pi]\zeta^2}{n} <\plo,f-\plo>= 0.
\end{equation}
It follows, subtracting this first order condition from \eqref{eq:key_emp_inq_without_foc} for $f = \hlo$, that  
\begin{equation}
\label{eq:key_emp_inq}
\begin{aligned}
    0 &\ge \Pn \frac{(1-D_i) \exp(\plo_i)}{\E[\pi]} \ell(\plo_i -\hlo_i) +
           (\Pn-\P)\frac{\left((1-D_i)\exp( \plo_i) - D_i \right)}{\E[\pi]}(  \plo_i-\hlo_i) \\
           + &\frac{\overline\pi\zeta^2}{2\E[\pi]n}\|\hlo -\plo \|^2_{\m{F}} + \frac{(\overline\pi- \E[\pi])\zeta^2}{\E[\pi]n}<\hlo -\plo, \plo>_{\m{F}} 
\end{aligned}
\end{equation}

Now we've found that $\hat\delta = \plo-\hlo$ \emph{does not} satisfy the inequality
\begin{equation}
\begin{aligned}
    \label{eq:emp_inq_delta}
0 &< \Pn \frac{(1-D_i) \exp(\plo_i)}{\E[\pi]} \ell(\delta_i)+ (\Pn-\P)\frac{\left((1-D_i)\exp( \plo_i) - D_i \right)}{\E[\pi]}\delta_i \\
&+ \frac{\overline \pi\zeta^2}{2\E[\pi]n}\| \delta \|^2_{\m{F}} + \frac{(\overline\pi- \E[\pi])\zeta^2}{\E[\pi]n}<\delta, \plo>_{\m{F}} 
\end{aligned}
\end{equation}

We will find a radius $r$ for which, on a high probability event, every $\delta$ with $\norm{\delta}_2 \ge r$  \emph{does} satisfy it, implying that on that event $\norm{\hat\delta}_r \le r$. To do this, we'll 
work with the following uniform bounds, which we'll prove hold for some constant $\eta$ on 
an event of probability tending to one.

\begin{equation}
\label{eq:theta-uniform-bounds}
\begin{aligned}
&\Pn \frac{(1-D_i)\exp\left( \plo_i \right)}{\E[\pi]}\ell(\delta_i) \ge \eta
 \ell\left(\frac{\norm{\delta}_2}{4}\right) 
\quad && \text{ for all } \delta \in \mt{F} \text{ with } \norm{\delta}_2 \ge r \\
&\abs{(\Pn-\P)\frac{\left((1-D_i)\exp( \plo_i) - D_i \right)}{\E[\pi]}\delta_i} + 
 \frac{(\overline\pi- \E[\pi])\zeta^2}{\E[\pi]n}<\delta, \plo>_{\m{F}} 
\le \frac{\eta}{33} r\norm{\delta}_2
\quad && \text{ for all } \delta \in \mt{F}
\end{aligned}
\end{equation}
On this event, to show that a curve with $\norm{\delta}_2 \ge r$ satisfies \eqref{eq:emp_inq_delta}, 
it suffices to show that $\ell(\norm{\delta}_2/4) \ge r\norm{\delta}_2/33$. That is, letting $x=\norm{\delta}_2/4$,
that $\ell(x)/x \ge 4r / 33$. We can see that this ratio is increasing in $x$ by calculating its derivative.
\begin{equation}
\frac{d}{dx}\qty{\ell(x)/x}=\frac{\qty{-(1+x)e^{-x}+1}}{x^2} \ge 0 \qqtext{ when } e^x \ge 1+x,
\end{equation}
i.e., considering the Taylor series for $e^x$, for all $x \ge 0$. Thus, this condition holds when $\norm{\delta}_2 \ge r$ if it does when $\norm{\delta}_2=r$. To show that it does when $\norm{\delta}_2=r$, we'll use Taylor's theorem with MVT remainder.
Because $\ell(x)=f(x)-f(0)-f'(0)x$ for $f(x)=e^{-x}$, Taylor's theorem implies $\ell(x)=e^{-y}x^2/2$ and therefore 
$\ell(x)/x=e^{-y}x/2$ for some $y \in [0,x]$. Thus, 
\begin{equation}
\ell(r/4)/(r/4) \ge 4r/33 \qqtext{ if } e^{-y}r/8 \ge 4r/33 \qqtext{ i.e. if } e^{-y} \ge 32/33.   
\end{equation}
This holds if $y \le \log(33/32)$ and therefore if $r/4 \le \log(33/32)$; because $r \to 0$
by assumption this holds eventually. We conclude by proving the bounds \eqref{eq:theta-uniform-bounds}.

\textit{The lower bound.}
Using the property $\ell(x) \ge \ell(|x|)$ and bounds of the form
$Z \ge \theta \{ Z \ge \theta \}$ on factors in each term, 
 we get:
\begin{align*}
   \frac{\Pn(1-D_i)\exp\left( \plo_i \right)}{\E[\pi]}\ell(\delta_i)
&= \Pn \frac{1-D_i}{1-\pi_i}\frac{\pi_i\exp(\muplo_i -\lo_i)}{\E[\pi]}\exp(\plo_i-\muplo_i)\ell(\delta_i)  \\
&\ge \min_i \exp(\plo_i - \muplo_i) q_{\epsilon}\exp( -x\| \muplo -\lo\|_{\psi_2})  \ell\left(\frac{\norm{\delta}_2}{4}\right) \\
&\times \Pn\frac{1-D_i}{1-\pi_i}\left\{ \pi _i \ge q_{\epsilon} \E[\pi]\right\}\left\{ \muplo_i -\lo_i\ge  -x\| \muplo -\lo\|_{\psi_2}\right\} \{ 4| \delta_i| \ge \|\delta\|_2\}.
\end{align*}
From here we proceed essentially as we did from 
the analogous equation \eqref{eq:loo-zeroth-order-mu}
in the proof of Lemma \ref{lem:or_proj}.
\begin{align*}
&\Pn\frac{1-D_i}{1-\pi_i}\left\{ \pi _i \ge q_{\epsilon} \E[\pi]\right\}\left\{ \muplo_i -\lo_i\ge  -x\| \muplo -\lo\|_{\psi_2}\right\} \{ 4| \delta_i | \ge \|\hlo - \plo\|_2\} \\
&\overset{(a)}{\ge} \Pn \frac{1-D_i}{1-\pi_i}\left\{ \pi _i \ge q_{\epsilon} \E[\pi]\right\} + \Pn\frac{1-D_i}{1-\pi_i}\left\{ 4| \delta_i| \ge  \|\delta\|_2\right\} - \Pn\frac{1-D_i}{1-\pi_i}  \\
&- \Pn\frac{1-D_i}{1-\pi_i}\left\{ \muplo_i -\lo_i \le  -x\| \muplo -\lo\|_{\psi_2}\right\} \\
&\overset{(b)}{=}\E\left\{ \pi _i \ge q_{\epsilon} \E[\pi]\right\} + \Pn\frac{1-D_i}{1-\pi_i}\left\{ 4| \delta_i| \ge  \|\delta\|_2\right\} - 1 -  \\
& - \E\left\{ \muplo_i -\lo_i \le  -x\| \muplo -\lo\|_{\psi_2}\right\}  + o_p(1)  \\
&\overset{(c)}{\ge}  -\epsilon - \E\left\{ \muplo_i -\lo_i \le  -x\| \muplo -\lo\|_{\psi_2}\right\} + \Pn\frac{1-D_i}{1-\pi_i}\left\{ 4| \delta_i| \ge  \|\delta_i \|_2\right\} + o_p(1)
\end{align*}
Here (a) follows from inclusion/exclusion as in Lemma~\ref{lem:or_proj}; (b) from the law of large numbers, observing that
$\E[1-D_i \mid \pi_i]/(1-\pi_i) = 1$; and (c) from Assumption \ref{as:overlap}.
 and (c) because Assumption \ref{as:overlap} is that the first term in line (b) is at least $1-\epsilon$. Choosing $x$ and $\epsilon$ appropriately as in Lemma~\ref{lem:or_proj}, we have the lower bound
\begin{equation}
\label{eq:theta-est-lb-intermediate}
   \Pn\frac{(1-D_i)\exp\left( \plo_i \right)}{\E[\pi]}\ell(\delta_i) \ge 
c \min_i \exp(\plo_i - \muplo_i)  \ell\left(\frac{\norm{\delta}_2}{4}\right) \times c \Pn\frac{1-D_i}{1-\pi_i}\left\{ 4| \delta_i| \ge  \|\delta_i \|_2\right\}.
\end{equation}

The last factor on the right is, with probability tending to one, no smaller than some constant $c > 0$ for all $\delta$ in the set $\mt{F}$ with $\norm{\delta}_2 \ge r$ if $r$ satisfies the following linear fixed point condition for a sufficiently small constant $\eta > 0$ and $q > 2$.
\begin{equation}
\label{eq:theta-linear-fp-cond}
    \eta r \ge \norm*{\frac{1-D_i}{1-\pi_i}}_q \E \sup_{\substack{\delta \in \mt{F} \\ \norm{\delta} \le r}} \Pn g_i \delta_i 
    \qqtext{ which, for $q=3$, is }
     \frac{\eta r}{\sqrt[3]{\E\qty{\frac{1}{(1-\pi_i)^2}}}} \ge \E \sup_{\substack{\delta \in \mt{F} \\ \norm{\delta} \le r}} \Pn g_i \delta_i 
\end{equation}
Here $g_1 \ldots g_n$ is a sequence of independent standard normals independent of our data. The first version follows from a variant of the proof of Theorem 5.3 of \citet{mendelson2014learning} that uses the multiplier inequality \citep[Corollary 1.10]{mendelson2016upper} to replace the multipliers $(1-D_i)/(1-\pi_i)$ with a constant multiple of their $L_q$ norm.\footnote{In particular, we use \citet[Corollary 1.10]{mendelson2016upper} to bound $(\Pn-\P)\{(1-D_i)/(1-\pi_i)\}\phi_u(\delta_i)$ by $\norm{(1-D_i)/(1-\pi_i)}_q \Pn g_i \phi_u(\delta_i)$ where \citet{mendelson2015learning} uses the bounded difference inequality, then proceed as in \citet{mendelson2015learning}.}
Our sufficient version follows by taking $q=3$ and observing that 
\begin{equation}
\E \qty[\qty(\frac{1-D_i}{1-\pi_i})^3] 
= \E\qty[ \frac{ \E\qty{ (1-D_i)^3 \mid \pi_i} }{ (1-\pi_i)^3 } ] 
= \E\qty[ \frac{ \E\qty{ (1-D_i) \mid \pi_i} }{ (1-\pi_i)^3 } ] 
= \E\qty[ \frac{ 1-\pi }{ (1-\pi_i)^3 } ] = \E\qty[ \frac{1}{(1-\pi)^2} ]
\end{equation}

We derive a simplified version of \eqref{eq:theta-est-lb-intermediate} by making this substitution and
substituting a bound for the first factor.
Because $\min_i \exp(\plo_i - \muplo_i) \ge \exp(-\max_i \abs{\plo_i - \muplo_i})$ and
the subgaussianity of $\plo_i-\muplo_i$ implies $\abs{\plo_i - \muplo_i} \le c\norm{\plo-\muplo}_2 \sqrt{\log(n)}$ with probability tending to one, we get the following bound. For $r$ satisfying \eqref{eq:theta-linear-fp-cond},
with probability tending to one,
\begin{equation}
   \Pn \frac{(1-D_i)\exp\left( \plo_i \right)}{\E[\pi]}\ell(\delta_i) \ge c
\exp\qty(\norm{\plo - \muplo}_2 \sqrt{\log(n)}) \ell\left(\frac{\norm{\delta}_2}{4}\right) 
\quad \text{ for all } \delta \in \mt{F} \text{ with } \norm{\delta}_2 \ge r.
\end{equation}
Furthermore, this exponential factor is bounded, as
as Corollary~\ref{cor:pop_proj} gives a bound $\norm{\plo - \muplo}_2 = O(\norm{\mu-\bestmu}_2 + 1/\sqrt{n})$
from which it follows that---given our assumed bound on $\norm{\mu-\bestmu}_2$---the exponentiated quantity is bounded.  Thus, for some constant $\eta$, the lower bound in \eqref{eq:theta-uniform-bounds} holds.

\textit{The upper bound.}

We will show that each term is bounded by half our upper bound on the sum, i.e., that for  all $\delta \in \mt{F}$,
\begin{equation}
\label{eq:ub-both}
\begin{aligned}
\abs{(\Pn-\P)\frac{\left((1-D_i)\exp( \plo_i) - D_i \right)}{\E[\pi]}\delta_i} \le \frac{\eta/2}{33} r\norm{\delta}_2 \qqtext{and}  
 \frac{(\overline\pi- \E[\pi])\zeta^2}{\E[\pi]n}<\delta, \plo>_{\m{F}}  \le  \frac{\eta/2}{33} r\norm{\delta}_2 
 \end{aligned}
\end{equation}
For the second, we use Chebyshev's inequality. Because $\bar\pi - \E[\pi] = (\Pn - \P) D_i$ has variance $\E[\pi](1-\E[\pi])/n$, it follows by Chebyshev's inequality that with probability $1-t$, $\abs{\overline\pi- \E[\pi]} \le t^{-1/2}\sqrt{\E[\pi](1-\E[\pi])/n}$
and conseqently our bound is satisfied when
$$
\frac{\sqrt{\E[\pi](1-\E[\pi])}\zeta^2}{\E[\pi]n}\frac{<\delta, \plo>_{\m{F}}}{\sqrt{n}}  \le  t^{1/2}\frac{\eta/2}{33} r\norm{\delta}_2.
$$

Furthermore, via the Cauchy-Schwarz bound $<\delta, \plo>_{\m{F}} \le \norm{\delta}_{\m{F}}\norm{\theta}_{\m{F}}$ and Assumption~\ref{as:out_mom} 
that $\norm{\cdot}_{\m{F}} \le c \norm{\cdot}_2$ for some fixed constant $c$, this holds when
$$
c^2 \frac{\sqrt{\E[\pi](1-\E[\pi])}\zeta^2}{\E[\pi] n} \norm{\delta}_2\frac{\norm{\plo}_2}{\sqrt{n}}  \le  t^{1/2}\frac{\eta/2}{33} r\norm{\delta}_2 \quad \text{ and therefore when } \quad r \ge 66c^2t^{-1/2}\eta^{-1} \frac{\norm{\plo}_2}{\sqrt{n}}  \frac{\sqrt{\E[\pi](1-\E[\pi])}\zeta^2}{\E[\pi] n}.
$$
Assumption \ref{as:overlap} and previous Lemmas guarantee that $ \|\plo\|_{2} = o(\sqrt{n})$. As a result, this holds for $t$ tending to zero as long as $\zeta^2 = O(\sqrt{\E[\pi]} r n).$

To establish the first bound in \eqref{eq:ub-both}, we will focus on proving the bound below. From this, it follows that the upper bound from \eqref{eq:ub-both} holds for all $\delta \in \mt{F}$ via a scaling argument. Every $\delta \in \mt{F}$ is $(\norm{\delta}/r)(r\delta/\norm{\delta})$; this bound holds for the latter factor; and multiplying both sides
by the former gives the upper bound we claimed.
\begin{equation}
\label{eq:theta-uniform-ub-sufficient}
\abs{(\Pn-\P)\frac{\left((1-D_i)\exp( \plo_i) - D_i \right)}{\E[\pi]}\delta_i} < \frac{\eta/2}{33} r^2
 \qqtext{ for all } \delta \in \mt{F} \qqtext{ with } \norm{\delta}_2 = r.
\end{equation}

This is a bound on a centered multiplier process. We decompose the multiplier.
\begin{align*}
    (1-D) \exp(\plo) - D = \frac{1-D}{1-\pi}(1-\pi)\exp(\plo) - \pi\frac{D}{\pi} = 
    \pi \left(\frac{1-D}{1-\pi}\exp(\plo -\lo) - \frac{D}{\pi}\right) = \\
     \pi \left(\exp(\plo -\lo) - 1\right) + \frac{(\pi - D)}{1-\pi}\pi(\exp(\plo-\lo)-1)  + \frac{(\pi - D)}{1-\pi}.
\end{align*}

By using results from \citet{mendelson2016upper}; Corollary~\ref{cor:mom_prod} with $\lambda_1=1$, $\lambda_2=0$ and $\lambda_3=\lambda_4=1$; and the triangle inequality,
\begin{align*}
        \sup_{\substack{\delta \in \mt{F} \\ \norm{\delta} \le r }}\abs*{ (\Pn-\P)\frac{\pi_i \left(\exp(\plo_i -\lo_i) - 1\right)}{\E[\pi]} \delta_i } 
&= O_p\qty(\norm*{\frac{\pi \left(\exp(\plo -\lo) - 1\right)}{\E[\pi]}}_{3})\E \sup_{\substack{\delta \in \mt{F} \\ \norm{\delta} \le r }} \Pn g_i \delta_i \\
&= O_p\qty( \E \sup_{\substack{\delta \in \mt{F} \\ \norm{\delta} \le r }} \Pn g_i \delta_i ) = O_p\qty(r\sqrt{\frac{p}{n}}).
\end{align*}
where the penultimate equality follows from Corollaries \ref{cor:mom_prod} and \ref{cor:pop_proj}.

Similarly, we have:
\begin{align*}
  \sup_{\substack{\delta \in \mt{F} \\ \norm{\delta} \le r }}\left| \Pn\frac{\pi_i - D_i}{1-\pi_i}\frac{\pi_i \left(\exp(\plo_i -\lo_i) - 1\right)}{\E[\pi]}\delta_i\right| = \\
    \left\|\frac{\pi - D}{1-\pi}\frac{\pi \left(\exp(\plo -\lo) - 1\right)}{\E[\pi]}\right\|_{3}\E \sup_{\substack{\delta \in \mt{F} \\ \norm{\delta} \le r }} \Pn g_i \delta_i = O_p\left(\E \sup_{\substack{\delta \in \mt{F} \\ \norm{\delta} \le r }} \Pn g_i \delta_i\right) = O_p\qty(r\sqrt{\frac{p}{n}}),
\end{align*}
To get the penultimate equality, we use the bound:
\begin{align*}
    \left\|\frac{\pi - D}{1-\pi}\frac{\pi \left(\exp(\plo -\lo) - 1\right)}{\E[\pi]}\right\|_{3}^3 = \E\left[\frac{\pi(\pi^2 +(1-\pi)^2)}{(1-\pi)^2}\left(\frac{\pi \left(\exp(\plo -\lo) - 1\right)}{\E[\pi]}\right)^3\right] \le \\
    \E\left[\left(\frac{\pi \left(\exp(\plo -\lo) - 1\right)}{\E[\pi]}\right)^3\right] +     \E\left[\exp(2\lo)\left(\frac{\pi \left(\exp(\plo -\lo) - 1\right)}{\E[\pi]}\right)^3\right] = O(1), 
\end{align*}

Finally, we analyze the last term. By Markov's inequality, we have:
\begin{equation*}
    \sup_{\delta_i \in \mt{F}, \|\delta\|_2\le r }\left| \Pn\frac{\pi_i -D_i}{(1-\pi_i)}\delta_i\right| = O_p\left(\mathbb{E}\qty[\sup_{\delta_i \in \mt{F}, \|\delta\|_2\le r }\left| \Pn\frac{\pi_i -D_i}{(1-\pi_i)}\delta_i\right|]\right)
\end{equation*}
%\dmitry{Fix notation}
Using the $L^2$ norm instead of $L^1$ we get:
\begin{align*}
\left\|\sup_{\delta_i \in \mt{F}, \|\delta\|_2\le r } \Pn\frac{\pi_i -D_i}{(1-\pi_i)}\delta_i\right\|_2 
&= \left\|\sup_{\|X_i^\top \beta\|_2\le r } \Pn\frac{\pi_i -D_i}{(1-\pi_i)}X_i^\top \beta\right\|_2 \\
&\le  \left\|\sup_{\|X_i^\top \beta\|_2\le r } \norm{\Pn\frac{\pi_i -D_i}{(1-\pi_i)}X_i \Sigma^{-1/2}}_{\ell_2} \norm{\Sigma^{1/2}\beta}_{\ell_2} \right\|_2 \\
& = r \|\norm{\Pn\frac{\pi_i -D_i}{(1-\pi_i)}X_i \Sigma^{-1/2}}_{\ell_2}\|_2
\end{align*}
We then compute the last norm squared:
\begin{align*}
    &\left\|\norm{\Pn\frac{\pi_i -D_i}{(1-\pi_i)}X_i \Sigma^{-1/2}}_{\ell_2}^2\right\|_2^2 = \frac{1}{n} \mathbb{E} \left[\frac{\pi}{1-\pi} \| X \Sigma^{-1/2}\|_{\ell_2}^2\right] =  \frac{\E[\pi]}{n} \mathbb{E} \left[ \exp(\lo -  \log(\E[\pi])) \| X \Sigma^{-1/2}\|_{\ell_2}^2\right]\\
    & =  \frac{\E[\pi]}{n} \sum_{j=1}^p  \mathbb{E} \left[ \exp(\lo -  \log(\E[\pi])) \tilde X_j^2\right] \le \frac{\E[\pi]}{n} \|\exp(\lo -  \log(\E[\pi]))\|_2  \sum_{j=1}^p \| \tilde X_j^2\|_2 = O\qty(\frac{\E[\pi]p}{n})
\end{align*}
where $\tilde X_j$ is the $j$-th feature of the vector $X \Sigma^{-1/2}$. The last equality follows from Corollary \ref{cor:mom_prod} and because by construction and Assumption \ref{as:subg} $\tilde X_j$ are subgaussian with unit variance, and thus their fourth moment is bounded by a constant. As a result, we can conclude 
\begin{equation*}
     \sup_{\delta_i \in \mt{F}, \|\delta\|_2\le r }\left| \Pn\frac{\pi_i -D_i}{\E[\pi](1-\pi_i)}\delta_i\right| = O_p\left(r\sqrt{\frac{p}{\E[\pi]n}}\right).
\end{equation*}
Combining the lower bound and the upper bounds, we can conclude that $r = \sqrt{\frac{p}{\E[\pi]n}}$, proving the first result. The bound on $\Pn \frac{(1-D_i) \exp(\plo_i)}{\E[\pi]} \ell(\hlo_i -  \plo_i)$ follows from \eqref{eq:emp_inq_delta}. Finally, to get the bound on the empirical squared norm, we observe:
%\dmitry{Add details}
\begin{align*}
     \left\|\| \hlo - \plo\|^2_2 -  \Pn (\hlo_i - \plo_i)^2\right| \le  \| \hlo - \plo\|^2_2\left\|\hat \Sigma_{norm} - \mathcal{I}_p\right\|_{op} = o_p\left( \| \hlo - \plo\|^2_2\right)
\end{align*}
where $\hat \Sigma_{norm}$ is the empirical covariance matrix normalized by the population covariance matrix. The last implication follows from Theorem 4.7.1 in \cite{vershynin2018high}. The bound for $  \Pn (\hlo_i - \plo_i)^2$ then follows.

\end{proof}
Our final lemma established a version of a uniform bound for $(\hlo_i - \plo_i)$.
\begin{lemma}\label{lem:conv_rate_3}
    Suppose Assumptions \ref{as:subg}-\ref{as:out_mom} hold and $\|\mu - \bestmu\|_2 = o\qty(\frac{1}{\sqrt{\log(n)}})$, then 
    \begin{equation*}
        \max\{ (1-D_i)(\hlo_i - \plo_i),0\}) = O_p(1).
    \end{equation*}
\end{lemma}

\begin{proof}
    We use LOO estimators $\hlo^{(i)}$ and get from optimality for $\hlo$:
\begin{equation}\label{eq:full_opt}
    \Pn \left[(1-D_i) \exp(\hlo_i) - D_i \hlo_i\right]  + \frac{\overline{\pi}\zeta^2}{2n}\| \hlo \|_{\m{F}} \le  \Pn \left[(1-D_j) \exp(\hlo_j^{(i)}) - D_i \hlo_j^{(i)}\right]  + \frac{\overline{\pi}\zeta^2}{2n}\| \hlo^{(i)} \|_{\m{F}}.
\end{equation}
Similarly, from optimality for $\hlo^{(i)}$ we get the opposite inequality:
\begin{equation}\label{eq:loo_opt}
        \frac{1}{n}\sum_{j\ne i} \left[(1-D_j) \exp(\hlo_j^{(i)}) - D_i \hlo_j^{(i)}\right]  + \frac{\overline{\pi}\zeta^2}{2n}\| \hlo \|_{\m{F}} \le \frac{1}{n}\sum_{j\ne i}  \left[(1-D_j) \exp(\hlo_j) - D_i \hlo_j\right]  + \frac{\overline{\pi}\zeta^2}{2n}\| \hlo^{(i)} \|_{\m{F}}.
\end{equation}
Adding \eqref{eq:full_opt} and \eqref{eq:loo_opt} and eliminating the terms we get the following:
\begin{align*}
    (1-D_i) \exp(\hlo_i) - D_i \hlo_i \le  (1-D_i) \exp(\hlo_i^{(i)}) - D_i \hlo_i^{(i)}\Rightarrow \\
    (1-D_i) \exp(\hlo_i) \le (1-D_i) \exp(\hlo_i^{(i)}) \Rightarrow(1-D_i) \hlo_i \le (1-D_i)\hlo_i^{(i)} \Rightarrow\\
    (1- D_i) (\hlo_i - \plo_i) \le (1- D_i) (\hlo_i^{(i)} - \plo_i) \le |\hlo_i^{(i)} - \plo_i|  \Rightarrow\\
    \max\{\max_i (1- D_i) (\hlo_i - \plo_i),0\} \le \max_i|\hlo_i^{(i)} - \plo_i|.
\end{align*}
The first implication follows from multiplying both sides by $(1-D_i)$, the second implication follows from monotonicity of $\exp(\cdot)$, the third implication follows by subtracting $\plo_i$, and the final one follows because $\max_{i}|\hlo_i^{(i)} - \plo_i|\ge 0$.%\footnote{I wonder if at some point we can refine this argument and build a stronger connection to the LOO estimator, similar to one we have for standard OLS. For example, Lemma 3.1 in \href{https://tongzhang-ml.org/papers/nc03_loo.pdf}{this paper} seems to be useful, but I haven't explored the connection.}

Because all $\hlo_i^{(i)} - \plo_i$ are subgaussian (conditionally on $\hlo^{(i)}$), it follows:
\begin{equation*}
    \max_i|\hlo_i^{(i)} - \plo_i| = O_p\left(\sqrt{\log(n)}\max_i\|\hlo^{(i)} - \plo\|_2\right),
\end{equation*} 
and our next step is to show $\max_i\|\hlo^{(i)} - \plo\|_2 \ll\sqrt{\log(n)}$. The argument for this follows the same steps as in the previous proof and in the proof of Lemma \ref{lem:or_proj} and is omitted. 

\end{proof}

\newpage
\subsection{Error analysis}\label{ap:error}

\subsubsection{Decomposition}
We decompose the error into three parts:
\begin{equation}\label{eq:decomp}
\begin{aligned}
    &\xi= \Pn \left[(1-D_i)\hat\omega_i  - \frac{D_i}{\overline \pi}\right] \epsilon_i + 
\Pn \left[(1-D_i) \hat\omega_i - \frac{D_i}{\overline \pi}\right] (\mu_i- \hmu_i) + \\
&\Pn \left[(1-D_i) \hat\omega_i - \frac{D_i}{\overline \pi}\right]\hmu_i =: \xi_1 + \xi_2 + \xi_3.
\end{aligned}
\end{equation}
The last error, $\xi_3$, is the in-sample imbalance term, which we will control using the empirical FOCs. We leave its analysis for later and focus on the first two terms. 

To understand the behavior of the first two error terms, we'll work with a decomposition
of the multiplier term, i.e. the one in square brackets. This is our first step.
\begin{equation}
\label{eq:multiplier-decomp-step-1}
\begin{aligned}
&\left[(1-D_i) \overline \pi\hat\omega_i -D_i\right] \\
&= (1-D_i) \exp(\hlo_i) - D_i \\
&=    (1-D_i) \exp( \muplo_i) - D_i + (1-D_i)\qty{ \exp(\hlo_i)- \exp( \muplo_i)}\\
&=    (\pi_i - D_i)\qty{\exp(\muplo_i) + 1}   + (1-\pi_i) \exp(\muplo_i) - \pi_i +  (1-D_i)\qty{ \exp(\hlo_i)- \exp( \muplo_i)}  \\
&=     \frac{\pi_i - D_i}{1-\pi_i}(\pi_iu_i + 1)   + \pi_i u_i +  (1-D_i)\qty{ \exp(\hlo_i)- \exp( \muplo_i)}
\end{aligned} 
\end{equation}
In the last step, we use the identities $\frac{\pi_i}{1-\pi_i}\exp(-\lo_i) =1$ and $u_i=\exp(\muplo_i - \lo_i) - 1$, which imply 
\begin{equation*}
\begin{aligned}
\exp(\muplo_i) 
&=\frac{\pi_i}{1-\pi_i}\exp(\muplo_i-\lo_i) 
=\frac{\pi_i}{1-\pi_i}(u_i + 1) \quad \text{ and therefore } \\
\exp(\muplo_i)+1 
&= \frac{\pi_i}{1-\pi_i}(u_i + 1) + 1 
= \qty(\frac{\pi_i}{1-\pi_i}u_i + \frac{\pi_i}{1-\pi_i}) + \frac{1-\pi_i}{1-\pi_i} 
= \frac{\pi_i}{1-\pi_i}u_i + \frac{1}{1-\pi_i} 
\\
& =\frac{1}{1-\pi_i}\qty(\pi_i u_i + 1) \quad \text{ and } 
(1-\pi_i)\exp(\muplo_i) - \pi_i
= \pi_i\qty{\exp(\muplo_i - \lo_i) - 1} = \pi_i u_i.
\end{aligned}
\end{equation*}  

From here, we expand the last term of \eqref{eq:multiplier-decomp-step-1} by expanding $\exp(\muplo_i)$
around $\exp(\plo_i)$. 
\begin{align*}
(1-D_i)\qty{ \exp(\hlo_i)- \exp( \muplo_i) } \\
= (1-D_i)\qty{ \exp(\hlo_i)- \exp( \plo_i) }  
&+ (1-D_i)\qty{ \exp(\plo_i)- \exp( \muplo_i) } \\
= (1-D_i)\qty{ \exp(\hlo_i)- \exp( \plo_i) }  
&+ (1-D_i)\frac{\pi}{1-\pi}\exp(-\theta_i) \cdot \exp(\muplo_i)\qty{ \exp(\plo_i-\muplo_i)- 1} \\
= (1-D_i)\qty{ \exp(\hlo_i)- \exp( \plo_i) } 
&+ \frac{(1-D_i)\pi_i}{1-\pi_i}\exp(\muplo_i - \lo_i)\qty{ \exp(\plo_i -\muplo_i) - 1 } \\
= (1-D_i)\qty{ \exp(\hlo_i)- \exp( \plo_i) } 
&+ \frac{\qty(\qty{1-\pi_i} + \qty{\pi_i-D_i})\pi_i}{1-\pi_i}(u_i + 1)\qty{ \exp(\plo_i -\muplo_i) - 1 } \\
= (1-D_i)\qty{ \exp(\hlo_i)- \exp( \plo_i) } 
&+ \pi_i (u_i+1)\qty{ \exp(\plo_i -\muplo_i) - 1 } \\
&+ \frac{(\pi_i-D_i)\pi_i}{1-\pi_i}(u_i + 1)\qty{ \exp(\plo_i -\muplo_i) - 1 } \\
\end{align*}
Here, in the penultimate step, we've used the definitional identity $\exp(\muplo_i-\lo_i)=u_i+1$ along with the arithmetically obvious identity $1-D_i=1-\pi_i + \pi_i-D_i$. Substituting the result into 
\eqref{eq:multiplier-decomp-step-1} and grouping multiples of $\pi_i-D_i$ together yields the following.
\begin{equation}
\label{eq:multiplier-decomp}
\begin{aligned}
\left[(1-D_i) \overline \pi\hat\omega_i -D_i\right] 
&= \frac{\pi_i - D_i}{1-\pi_i}\qty[ \pi_iu_i + 1 +  \pi_i (u_i + 1)\qty{ \exp(\plo_i -\muplo_i) - 1 }] \\
&+ \pi_i u_i + (1-D_i)\qty{ \exp(\hlo_i)- \exp( \plo_i) } \\
&+ \pi_i (u_i + 1)\qty{ \exp(\plo_i -\muplo_i) - 1 } 
\end{aligned}
\end{equation}

This yields the following decompositions of $\bar\pi\xi_1$ and $\bar\pi\xi_2$.

\begin{equation}\label{eq:bias_decomp}
\begin{aligned}
\bar \pi\xi_1 
&= \Pn \frac{\pi_i - D_i}{1-\pi_i}\qty[ \pi_iu_i + 1 +  \pi_i (u_i + 1) \qty{ \exp(\plo_i -\muplo_i) - 1 }]\varepsilon_i \\
&+ \Pn \pi_i u_i\varepsilon_i + \Pn (1-D_i)\qty{ \exp(\hlo_i)- \exp( \plo_i) }\varepsilon_i \\ 
&+ \Pn \pi_i (u_i + 1)\qty{ \exp(\plo_i -\muplo_i) - 1 }\varepsilon_i \\
\bar\pi\xi_2
&= \Pn \frac{\pi_i - D_i}{1-\pi_i}\qty[ \pi_iu_i + 1 +  \pi_i (u_i + 1) \qty{ \exp(\plo_i -\muplo_i) - 1 }](\mu_i - \hmu_i) \\
&+ \Pn \pi_i u_i \nu_{\mu,i} + \Pn \pi_i u_i (\pmu_i-\hmu_i) + \Pn (1-D_i)\qty{ \exp(\hlo_i)- \exp( \plo_i) }(\mu_i - \hmu_i) \\
&+ \Pn \pi_i (u_i + 1)\qty{ \exp(\plo_i -\muplo_i) - 1 }(\mu_i - \hmu_i) 
\end{aligned}
\end{equation}
In the latter, we've used the identity $\mu_i = \pmu_i + \nu_{\mu,i}$ to break down the term  $\pi_i u_i (\mu_i - \hmu_i)$ into two.

\subsubsection{Population oracle comparison}
\begin{lemma}\label{lem:or_noise}
    Suppose  Assumptions \ref{as:emp_set}, \ref{as:subg} - \ref{as:overlap}, \ref{as:out_mom} hold, and $\|\mu - \bestmu\|_2 = o\left(\frac{1}{\sqrt{\log(n)}}\right)$, then:
\begin{equation*}
   \overline \pi\xi_1 =   \Pn \frac{\pi_i - D_i}{1-\pi_i}(\pi_iu_i + 1)\epsilon_i+\Pn\pi_iu_i \epsilon_i + o_p\left(\frac{\E[\pi]}{\sqrt{n}}\right) + 
   \Pn(1-D_i)( \exp(\hlo_i)- \exp(\plo_i))\epsilon_i.
\end{equation*}
\end{lemma}
\begin{proof}

Comparing our claimed asymptotic approximation to our decomposition \eqref{eq:bias_decomp}, we see that it's equivalent to the following bound.    
\begin{align*}
&\Pn\qty{\frac{\pi_i - D_i}{1-\pi_i} + 1 }\qty[ \pi_i (u_i+1)\qty{ \exp(\plo_i -\muplo_i) - 1 }]\varepsilon_i = o_p\qty(\frac{\E[\pi]}{\sqrt{n}}) \\
\end{align*}

To establish this, we will show that this quantity's variance is $o\left(\frac{\E[\pi]^2}{n}\right)$.
By the law of total variance, because the conditional mean of this quantity is zero, this variance
is the \emph{expectation of the conditional variance}. And because $\pi_i-D_i$ are independent conditionally mean-zero random variables with conditional variance $\pi_i(1-\pi_i)$, 
$$
\begin{aligned}
\E\qty[ \qty( \frac{\pi_i - D_i}{1-\pi_i} + 1)^2 \mid X, \eta] =  
&\E\qty[ \qty(\frac{\pi_i - D_i}{1-\pi_i})^2 \mid X, \eta] + 2\E\qty[ \qty(\frac{\pi_i - D_i}{1-\pi_i}) \mid X, \eta] + 1 \\
&= \frac{\pi_i(1-\pi_i)}{(1-\pi_i)^2} + 0 + 1 = \frac{\pi_i}{1-\pi_i} + \frac{1-\pi_i}{1-\pi_i} = \frac{1}{1-\pi_i}.
\end{aligned}
$$
Thus, using Assumption \ref{as:sel} we get:
\begin{align*}
&\frac{1}{n}\E \Pn \sigma_i^2 \E\qty[ \qty( \frac{\pi_i - D_i}{1-\pi_i} + 1)^2 \mid X, \eta]  \qty[ \pi_i (u_i + 1)\qty{ \exp(\plo_i -\muplo_i) - 1 }]^2 \\ 
&=\frac{1}{n}\E \Pn \frac{\sigma_i^2}{1-\pi_i} \qty[ \pi_i (u_i+1)\qty{ \exp(\plo_i -\muplo_i) - 1 }]^2 \\ 
&\le \E \frac{\max_{i}\sigma^2_{i}\max_{i}\qty{\exp(\plo_i - \muplo_i)-1}^2}{n} \Pn\left[\pi_i (u_i+1)\right]^2  = o_p\left(\frac{\E[\pi]^2}{n}\right) \\
\end{align*}
where $\sigma_i^2 := \sigma^2(X_i, \eta_i)$. The last equality follows from two arguments. First,  by Markov inequality and Corollary \ref{cor:mom_prod} we have:
\begin{equation*}
    \Pn\left[\pi_i (u_i+1)\right]^2 = \Pn\left[\pi_i \exp(\muplo_i - \lo_i)\right]^2 = O_p\left(\E\left[\pi^2 \exp(2(\muplo - \lo))\right]\right) = O_p\left(\E[\pi]^2]\right).
\end{equation*}
Second, by Corollary \ref{cor:pop_proj} and the assumed bound on $\|\mu - \bestmu\|_2$ we have that $\|\nu_{\muplo}\|_2 = O(\|\mu - \bestmu\|_2) = o\left(\frac{1}{\sqrt{\log(n)}}\right)$. By Assumption \ref{as:subg} this implies $\|\nu_{\muplo}\|_{\psi_2} = o\left(\frac{1}{\sqrt{\log(n)}}\right)$. Then, we have by the maximal inequality for subgaussian random variables:
\begin{multline}\label{eq:exp_max_bound}
    (\exp(\nu_{\muplo,i})-1)^2 \le (\exp(|\nu_{\muplo,i}|)-1)^2 \le (\exp(\max_i|\nu_{\muplo,i}|)-1)^2 = \\ O_p\left((\exp(\sqrt{\log(n)}\|\nu_{\muplo}\|_{\psi_2})-1)^2\right) = o_p(1)
\end{multline}
\end{proof}

Our next lemma establishes results for $\bar \pi\xi_2$.
\begin{lemma}\label{lem:or_bias}

    Suppose  Assumptions \ref{as:emp_set}, \ref{as:subg} -\ref{as:out_mom} hold, and $\|\mu - \bestmu\|_2 = o\left(\frac{1}{\sqrt{\log(n)}}\right)$.
    And, letting $r$ is any solution to a version of the fixed point conditions from Lemma~\ref{lemma:mu-l2-convergence} in which the Rademacher multipliers are replaced by gaussian ones, suppose $r \not\to \infty$. Then,
\begin{align*}
    \overline{\pi}\xi_2 &=   \Pn\pi_i \exp(\muplo-\lo)\nu_{\muplo,i}\nu_{\mu,i} + 
    o_p\left(\sqrt{\frac{\E[\pi]}{n}}\right) + o_p(\E[\pi]\|\mu - \bestmu\|_2^2) +  O_p\qty(\E[\pi]  \min\qty{r,\  \frac{r^2}{\norm{\mu-\bestmu}}_2}) \\
    &+ \Pn(1-D_i)( \exp(\hlo_i)- \exp( g_i))(\mu_i  - \hmu_i).
\end{align*}
\end{lemma}
Note that, as discussed in the lead-up to Corollary~\ref{coro:mu-rates},
in the case that $\mt{F}$ is --- or is contained in --- a $p$-dimensional subspace, the linear fixed point condition is satisfied if 
$\eta r \ge c r \sqrt{p/n}$, i.e., irrespective of $r$ it holds if $p/n$ is bounded by a sufficiently small constant. And the quadratic one holds for $r \propto \sqrt{p/n} \norm{\mu-\pmu}_2$. Thus, as long as $p/n \to 0$, 
we can subtitute in the following bound on $r^2/\norm{\mu-\bestmu}$
\begin{equation}
\frac{\qty(\sqrt{p/n}\norm{\mu-\pmu}_2)^2}{\norm{\mu-\bestmu}_2} 
= (p/n) \norm{\mu-\bestmu}_2 \frac{\norm{\mu-\pmu}_2}{\norm{\mu-\bestmu}_2} \lesssim (p/n) \norm{\mu-\bestmu}_2,
\end{equation}
where the boundedness of the ratio $\norm{\mu-\pmu}_2/\norm{\mu-\bestmu}_2$
is a result of Corollary~\ref{cor:pop_proj}. This implies the following simplified result.

\begin{corollary}
\label{coro:or_bias}
Suppose $\mt{F}$ is a set of dimension no larger than $p$. Then 
if $p/n$ is bounded by a sufficiently small constant, we have
\begin{equation*}
\begin{aligned}
\overline{\pi}\xi_2 &=   \Pn\pi_i \exp(\muplo-\lo)\nu_{\muplo,i}\nu_{\mu,i} + 
    o_p\left(\sqrt{\frac{\E[\pi]}{n}}\right) + o_p(\E[\pi]\|\mu - \bestmu\|_2^2) +  O_p\qty(\E[\pi]\norm{\mu-\bestmu}_2 \frac{p}{n}) \\
    &+ \Pn(1-D_i)( \exp(\hlo_i)- \exp( g_i))(\mu_i  - \hmu_i).
\end{aligned}
\end{equation*}
\end{corollary}

\begin{proof}

Comparing our claimed asymptotic approximation to our decomposition \eqref{eq:bias_decomp}, we see that it's equivalent to the following bound.     

\begin{equation}
\label{eq:bias_remainder}
\begin{aligned}
&\Pn \frac{\pi_i - D_i}{1-\pi_i}\qty[ \pi_iu_i + 1 +  \pi_i\exp(\muplo-\lo)\qty{ \exp(\nu_{\muplo,i}) - 1 }](\mu_i - \hmu_i) \\
& +\qty{ \Pn \pi_i \exp(\muplo-\lo)\qty{ \exp(\nu_{\muplo,i}) - 1 }(\mu_i - \hmu_i) - \Pn \pi_i\exp(\muplo-\lo)\nu_{\muplo,i}\nu_{\mu,i}}\\
&+ \Pn \pi_i u_i\nu_{\mu,i}  +\\
&\Pn \pi_i u_i (\pmu_i-\hmu_i)  \\
& = o_p\qty(\sqrt{\frac{\E[\pi]}{n}} + \E[\pi]\|\mu - \bestmu\|_2^2) + O_p\qty( E[\pi] \min\qty{r, \frac{r^2}{\norm{\mu-\bestmu}_2} }).
\end{aligned}
\end{equation}

\paragraph{Term 1 of \eqref{eq:bias_remainder}:}
The first term is an average of terms that are conditional on $X,\nu$, independent with mean zero. We will show that its variance is $o(\frac{\E[\pi]}{n})$.
By the law of total variance, because its conditional mean is constant, its variance is the expectation of the conditional variance. And because $(\pi_i-D_i)/(1-\pi_i)$ has conditional variance $\pi_i(1-\pi_i)/(1-\pi_i)^2=\pi_i/(1-\pi_i)$, this is 
\begin{equation}
\begin{aligned}
&\frac{1}{n} \E \Pn \frac{\pi_i}{1-\pi_i}\qty[ \pi_iu_i + 1 +  \pi_i (u_i+1) \qty{ \exp(\plo_i -\muplo_i) - 1 }]^2(\mu_i - \hmu_i)^2 \\
&\le 
\frac{1}{n} \E \max_i (\mu_i - \hmu_i)^2  \Pn \frac{\pi_i}{1-\pi_i}\qty[ \pi_iu_i + 1 +  \pi_i (u_i+1)\qty{ \exp(\plo_i -\muplo_i) - 1 }]^2\\
\end{aligned}
\end{equation}
Because $\max_i (\mu_i - \hmu_i)^2 = o_p(1)$, from \ref{lem:or_proj} and assumption that $\|\mu - \bestmu\|_2 = o_p(\frac{1}{\sqrt{\log(n)}})$ it suffices to show that the $\Pn$ factor is $O_p(\E[\pi])$. The quantity in square brackets is 
$a+b+c$ for $a=\pi_i (u_i + 1)$, $b=1-\pi_i$, and $c=\pi_i (u_i+1)\qty{ \exp(\plo_i -\muplo_i) - 1 }$,
and via the elementary inequality $(a+b+c)^2 \le 3(a^2+b^2+c^2)$, it suffices to show that 
$\Pn\qty{ \pi_i/(1-\pi_i) } x^2 = O_p(\E[\pi])$ for $x = a,b,c$. 

We address each term in turn. First, we have:
\begin{multline}\label{eq:bound_1}
    \Pn \frac{\pi_i}{(1-\pi_i)}(\pi_i(u_i + 1))^2 = \Pn \pi_i^2\exp(\lo_i)\exp(2(\muplo_i - \lo_i)) \\
    =O_p\left(\E\left[\pi^2\exp(\lo)\exp(2(\muplo - \lo))\right]\right) = O_p(\E[\pi]^3),
\end{multline}
where the second equality follows from Markov inequality, and the last equality follows from  Corollary \ref{cor:mom_prod} for $\lambda_1=1$, $\lambda_2=\lambda_3=2$, and $\lambda_4=0$. Second, we have
\begin{align*}
    \Pn \frac{\pi_i}{(1-\pi_i)}(1-\pi_i)^2 = \Pn \pi_i(1-\pi_i) \le \Pn \pi_i = O_p\left(\E [\pi]\right).
\end{align*}
Finally, we have
\begin{align*}
   &\Pn\frac{\pi_i}{1-\pi_i}\left[\pi_i \exp(\muplo_i-\lo_i)[\exp(\nu_{\muplo,i})-1]\right]^2 =\\
   &\Pn\exp(\lo_i)\left[\pi_i \exp(\muplo_i-\lo_i)[\exp(\nu_{\muplo,i})-1]\right]^2 \le\\
   &\max_{i}[\exp(\nu_{\muplo,i})-1]^2\Pn\exp(\lo_i)\left[\pi_i \exp(\muplo_i-\lo_i)\right]^2=o_p\left(\E[\pi]^3\right)
\end{align*}
where the last equality follows, because $\max_{i}[\exp(\nu_{\muplo,i})-1]^2$ is bounded by \eqref{eq:exp_max_bound} and $\Pn\exp(\lo)\left[\pi_i \exp(\muplo_i-\lo_i)\right]^2 = O_p\qty(\E[\pi]^3)$ from \eqref{eq:bound_1}.

\paragraph{Term 2 of \eqref{eq:bias_remainder}:}
We will show that, after dividing this term by $\E[\pi]$, the result is 
$o_p(1/\sqrt{n}) + \min\qty{r\norm{\mu-\bestmu}_2, r^2})$.
The latter term here does not appear in \eqref{eq:bias_remainder} directly, as it is $\norm{\mu-\bestmu}_2 \ll 1$ times a term that does---a term related to Term 4 of \eqref{eq:bias_remainder}.

To establish this bound, we work with the following decomposition of the aformentioned quotient with $\E[\pi]$.
\begin{equation}\label{eq:bias_nasty}
\begin{aligned}
    &\Pn\frac{\pi_i \exp(\muplo_i-\lo_i)}{\E[\pi]}[\exp(\nu_{\muplo,i})-1](\mu_i - \hmu_i) -\Pn \frac{\pi_i}{\E[\pi]}\exp(\muplo-\lo)\nu_{\muplo,i}\nu_{\mu,i}\\
    &=\Pn\frac{\pi_i \exp(\muplo_i-\lo_i)}{\E[\pi]}[\exp(\nu_{\muplo,i})-1](\pmu_i - \hmu_i)  + \\
    &\qty{\Pn\frac{\pi_i \exp(\muplo_i-\lo_i)}{\E[\pi]}[\exp(\nu_{\muplo,i})-1]\nu_{\mu,i}
    -\Pn \frac{\pi_i}{\E[\pi]}\exp(\muplo-\lo)\nu_{\muplo,i}\nu_{\mu,i}}.
\end{aligned}
\end{equation}
Our goal is to show that both terms in this decomposition are of the order $o_p(\| \mu - \bestmu\|_2^2)$.

First, we analyze the second term:
\begin{align*}
    &\Pn\frac{\pi_i \exp(\muplo_i-\lo_i)}{\E[\pi]}[\exp(\nu_{\muplo,i})-1]\nu_{\mu,i} -\Pn \frac{\pi_i}{\E[\pi]}\exp(\muplo_i-\lo_i)\nu_{\muplo,i}\nu_{\mu,i}\\
    &=\Pn\frac{\pi_i \exp(\muplo_i-\lo_i)}{\E[\pi]}[\exp(\nu_{\muplo,i})-\nu_{\muplo,i}  -1]\nu_{\mu,i} \\
    &=\Pn\frac{\pi_i \exp(\muplo_i-\lo_i)}{\E[\pi]}\ell(-\nu_{\muplo,i})\nu_{\mu,i}.
\end{align*}
We use H\"older inequality for the last expression:
\begin{equation*}
     \Pn\frac{\pi_i \exp(\muplo_i-\lo_i)}{\E[\pi]}\ell(-\nu_{\muplo,i})\nu_{\mu,i} \le \max_{i}|\nu_{\mu,i}|\Pn\frac{\pi_i \exp(\muplo_i-\lo_i)}{\E[\pi]}\ell(-\nu_{\muplo,i}) = O_p\qty(\sqrt{\log(n)}\|\mu - \bestmu\|_2^3),
\end{equation*}
where the last equality follows from bounds from Corollary \ref{cor:pop_proj}:
\begin{align*}
    &\max_{i}|\nu_{\mu,i}| = O_p\qty(\sqrt{\log(n)}\| \nu_\mu\|_{\psi_2}) = O_p\qty(\sqrt{\log(n)}\| \nu_\mu\|_{2}) =  O_p\qty(\sqrt{\log(n)}\| \mu - \bestmu\|_2),\\
    &\Pn\frac{\pi_i \exp(\muplo_i-\lo_i)}{\E[\pi]}\ell(-\nu_{\muplo,i}) = O_p\qty(\E\left[\frac{\pi \exp(\muplo-\lo)}{\E[\pi]}\ell(-\nu_{\muplo})\right]) = O_p\qty(\| \mu - \bestmu\|_2^2).
\end{align*}
We now return to the first term in \eqref{eq:bias_nasty} and recognize that $(\pmu_i - \hmu_i) \in \mt{F}$ and thus from the FOC \ref{eq:alt_foc_plo} it follows:
\begin{align*}
    &\Pn\frac{\pi_i \exp(\muplo_i-\lo_i)}{\E[\pi]}[\exp(\nu_{\muplo,i})-1](\pmu_i - \hmu_i) \\
    &= (\Pn - \P)\frac{\pi_i \exp(\muplo_i-\lo_i)}{\E[\pi]}[\exp(\nu_{\muplo,i})-1](\pmu_i - \hmu_i) - \frac{\zeta^2}{n}<\plo -\muplo,\pmu  - \hmu)>_{\m{F}}.
\end{align*}
We use empirical process arguments below to bound the first term, while for the second one we have:
\begin{equation}
\label{eq:term2-offset-bound}
\begin{aligned}
    &\qty| \frac{\zeta^2}{n}<\plo -\muplo,\pmu  - \hmu)>_{\m{F}}| \le \frac{\zeta^2}{n} \|\plo -\muplo\|_{\m{F}} \|\pmu  - \hmu\|_{\m{F}} \\
    &\overset{(a)}{\le} O_p\qty(\sqrt{\|\mu - \bestmu\|_2^2 +  \frac{1}{n}})\sqrt{\frac{\zeta^2}{n} } O_p(r) = 
    r O_p\qty(\norm{\mu-\bestmu}/\sqrt{n} + 1/n) = o_p(1/\sqrt{n})
\end{aligned}
\end{equation}
where to get (a) we use Corollary \ref{cor:pop_proj} to bound $\norm{\plo - \muplo}_{\m{F}}$ together with the norm equivalence $\norm{\cdot}_{\m{F}} \lesssim \norm{\cdot}_2$ which implies that 
$\|\pmu  - \hmu\|_{\m{F}} = O\qty(\|\pmu  - \hmu\|_{2}) = O_p(r)$.  The last bound follows because $r \not \to \infty$ and $\norm{\mu-\bestmu}_2\to 0$.
 
Next, we analyze the empirical process term. For $r$ satisfying the assumptions of Lemma~\ref{lemma:mu-l2-convergence}, except with gaussian multipliers used throughout\footnote{This modification allows the simplification we make in step (c). And because gaussian complexity is, up to constants, at least as large as Rademacher complexity, it doesn't inference with the application of Lemma~\ref{lemma:mu-l2-convergence}.},
there is a constant $c$ for which $\norm{\hmu-\pmu}_2 \le cr$ with probability $1-\delta$.
Thus, for any $\epsilon > 0$, with probability tending to $1-\delta$,
\begin{equation}
\label{eq:term2-centeredpart-bound}
\begin{aligned}
    &(\Pn - \P)\frac{\pi_i \exp(\muplo_i-\lo_i)}{\E[\pi]}[\exp(\nu_{\muplo,i})-1](\pmu_i - \hmu_i) \\
    &\le c \sup_{\substack{\delta \in \m{F} \\ \norm{\delta}_2 \le r}} (\Pn-\P) \xi_i \delta_i  \qqtext{for} \frac{\pi_i \exp(\muplo_i-\lo_i)}{\E[\pi]}[\exp(\nu_{\muplo,i})-1] \\
    &\overset{(a)}{\le} c \norm{\xi_i}_{2+\epsilon} \E \sup_{\substack{\delta \in \m{F} \\ \norm{\delta}_2 \le r}} \Pn g_i \delta_i 
    \overset{(b)}{\le} c\norm{\nu_{\muplo}}_{2} \E \sup_{\substack{\delta \in \m{F} \\ \norm{\delta}_2 \le r}} \Pn g_i \delta_i 
    \overset{(c)}{\le} c \norm{\nu_{\muplo}}_{2} \min\qty{r, \frac{r^2}{\norm{\pmu - \mu}_2}} \\
    &\overset{(d)}{=}  O(\norm{\bestmu - \mu}_2)
    \min\qty{r, \frac{r^2}{\norm{\pmu - \mu}_2}} =
    O\qty(\min\qty{\norm{\bestmu - \mu}_2 r, r^2}).
    \end{aligned}
\end{equation}
Here, in the step (a), we use a multiplier inequality from \citep[Corollary 1.10]{mendelson2016upper}, in the step (b),
we use the bound $\norm{\xi_i}_{2+\epsilon} = O(\norm{\nu_{\muplo}}_2)$, which we will establish via H\"older's inequality below, in step (c) we use fixed point conditions $r$ satisfies by definition,
and in step (d) a bound from Corollary~\ref{cor:pop_proj}.
Let's bound this multiplier.
\begin{equation}
\begin{aligned}
    &\left\|\frac{\pi\exp(\muplo -\lo)}{\E[\pi]}(\exp(\nu_{\muplo}) - 1)\right\|_{2+\epsilon}^{2+\epsilon} = 
    \left\|\left(\frac{\pi\exp(\muplo -\lo)}{\E[\pi]}\right)^{2+\epsilon}(\exp(\nu_{\muplo}) - 1)^{2+\epsilon}\right\|_{1} \le \\ 
    &\le  \left\|\left(\frac{\pi\exp(\muplo -\lo)}{\E[\pi]}\right)^{2+\epsilon}\right\|_{\frac{4}{2+\epsilon}}\left\|\left(\exp(\nu_{\muplo}) - 1\right)^{2+\epsilon}\right\|_{\frac{4}{2-\epsilon}}.
\end{aligned}
\end{equation}
The first factor is $O(1)$ by Corollary \ref{cor:mom_prod}. We need to perform a computation similar to the one in Lemma \ref{lem:ell_eq} for the second multiplier, which we omit since it follows the same logic. Conceptually $\nu_{\muplo}$ is small, which means that $(\exp(\nu_{\muplo}) - 1) \sim \nu_{\muplo}$ which implies, using the subgaussianity of $\nu_{\muplo}$, that
\begin{equation*}
    \left(\left\|\left(\exp(\nu_{\muplo}) - 1\right)^{2+\epsilon}\right\|_{\frac{4}{2-\epsilon}}\right)^{\frac{1}{2+\epsilon}} = \|\exp(\nu_{\muplo}) - 1\|_{\frac{4(2+\epsilon)}{2-\epsilon}}  \sim \| \nu_{\muplo}\|_{\frac{4(2+\epsilon)}{2-\epsilon}} = O\qty(\| \nu_{\muplo}\|_2) 
\end{equation*}

Summing the bounds from 
\eqref{eq:term2-offset-bound} and
\eqref{eq:term2-centeredpart-bound}, 
we get what we claimed.

\paragraph{Term 3 of \eqref{eq:bias_remainder}:}
Our goal is to show that this term is of the order $o_p\qty(\sqrt{\frac{\E[\pi]}{n}})$, but we will show a stronger bound $o_p\qty(\frac{\E[\pi]}{\sqrt{n}})$.  As before, we divide the term by $\E[\pi]$. By the FOC \eqref{eq:foc_muplo} for $u$ it follows that the terms in the average $\Pn \frac{\pi_i}{\E[\pi]} u_i \nu_{\mu,i}$ have mean zero. We compute the variance and apply CS:
\begin{equation*}
    \frac{1}{n} \E\left[\left(\frac{\pi}{\E[\pi]} u\nu_{\mu}\right)^2\right] \le \frac{1}{n} \sqrt{\E\left[\left(\frac{\pi}{\E[\pi]} u\right)^4\right]} \sqrt{\E\left[\nu^4_{\mu}\right]} =  o\left(\frac{1}{n}\right).
\end{equation*}
Here, the last equality follows from two bounds. First, using the triangle inequality's implication $\abs{(u+1)x} \le \abs{ux} + \abs{x}$ for $x=\pi/\E\pi$ and Corollary~\ref{cor:mom_prod} to bound the resulting terms, we get: 
\begin{equation}
\label{u:moment-bound}
    \left\|\frac{\pi}{\E[\pi]} u\right\|_4 \le \left\|\frac{\pi}{\E[\pi]} \exp(\muplo  -\lo)\right\|_4 + \left\|\frac{\pi}{\E[\pi]} \right\|_4 = O(1).
\end{equation}
Second, from Assumption \ref{as:subg} and Corollary \ref{cor:pop_proj} we have:
\begin{equation*}
    \|\nu_{\mu}\|_4 \le C \|\nu_{\mu}\|_2 = O(\|\mu - \bestmu\|_2) = o(1).
\end{equation*}

\paragraph{Term 4 of \eqref{eq:bias_remainder}:} 
Our goal is to show that this term is, after dividing by $\E[\pi]$, $o_p(1/\sqrt{n})+O_p\qty(\min\qty{r, \frac{r^2}{\mu-\bestmu}_2})$. 
To do this, this we start by centering it and using the FOC \eqref{eq:foc_muplo}:
\begin{equation*}
    \Pn \frac{\pi_i}{\E[\pi]} u_i (\pmu_i-\hmu_i) =  (\Pn - \P) \frac{\pi_i}{\E[\pi]} u_i (\pmu_i-\hmu_i) - \frac{\zeta^2}{n}<\muplo,\pmu-\hmu>_{\m{F}}.
\end{equation*}
We bound the second term:
\begin{align*}
    & \qty| \frac{\zeta^2}{n}<\muplo,\pmu-\hmu>_{\m{F}}| \le \frac{\zeta^2}{n} \|\muplo\|_{\m{F}} \|\pmu-\hmu\|_{\m{F}}. 
\end{align*}
By Assumption \ref{as:out_mom} and derived properties of $\|\pmu-\hmu\|_2$ we have $\|\pmu-\hmu\|_{\m{F}} =  O(\|\pmu-\hmu\|_2) = o_p(1)$. We also have:
\begin{equation*}
     \|\muplo\|_{\m{F}} \le  \|\muplo-\plo\|_{\m{F}} +  \|\plo\|_{\m{F}}
\end{equation*}
By Corollary \ref{cor:pop_proj} we have $\sqrt{\frac{\zeta^2}{n}}\|\muplo-\plo\|_{\m{F}} = o_p(1)$. Using Assumptions \ref{as:overlap}, \ref{as:out_mom} and Lemma \ref{lem:init_bound} we get $\|\plo\|_{\m{F}} = O\qty(\|\plo\|_{2}) = O\qty(\|\lo-\plo\|_{2} + \|\lo\|_2) = o(\sqrt{n})$. It thus follows:
\begin{align*}
    &\frac{\zeta^2}{n} \|\muplo\|_{\m{F}} \|\pmu-\hmu\|_{\m{F}} = \sqrt{\frac{\zeta^2}{n}} o_p(1)\times o_p(1) = o_p\qty(\frac{1}{\sqrt{n}}).
\end{align*}

For the second term, the argument used in 
\eqref{eq:term2-centeredpart-bound}, in combination with 
the $O_p(1)$ moment bound \eqref{u:moment-bound} on $\xi=(\pi/\E[\pi])u$  
below, implies that 
\begin{equation}
\begin{aligned}
(\Pn - \P) \frac{\pi_i}{\E[\pi]} u_i (\pmu_i-\hmu_i) 
&= O_p(\norm{\xi_i}_{2+\epsilon}) \min\qty{r, \frac{r^2}{\norm{\mu-\bestmu}}_2} 
= O_p\qty(\min\qty{r, \frac{r^2}{\norm{\mu-\bestmu}_2}}).
\end{aligned}
\end{equation}

\end{proof}

\subsubsection{Empirical errors}
\begin{lemma}\label{lem:emp_error}
    Suppose  Assumptions \ref{as:emp_set}, \ref{as:subg} -\ref{as:out_mom} hold, and $\|\mu - \bestmu\|_2 = o\left(\frac{1}{\sqrt{\log(n)}}\right)$, then we have 
\begin{equation*}
\begin{aligned}
        &\Pn\frac{(1-D_i)}{\E[\pi]}( \exp(\hlo_i)- \exp(\plo_i))\epsilon_i = o_p\left(\frac{1}{\sqrt{n}}\right),\\
        &\Pn\frac{(1-D_i)}{\E[\pi]}( \exp(\hlo_i)- \exp(\plo_i))(\mu_i- \hmu_i) = O_p\left(\frac{\sqrt{\log(n)}\|\mu - \bestmu\|_2p}{\E[\pi] n}\right) + o_p(\|\mu - \bestmu\|_2^2).
\end{aligned}
\end{equation*}
\end{lemma}
\begin{proof} We compute the conditional variance of the first empirical error:
\begin{equation*}
    \frac{1}{n}\Pn\frac{(1-D_i)}{\E[\pi]^2}( \exp(\hlo_i)- \exp(\plo_i))^2\sigma_i^2 \le \frac{\max_{i}\sigma_i^2}{n} \Pn\frac{(1-D_i)}{\E[\pi]^2}( \exp(\hlo_i)- \exp(\plo_i))^2.
\end{equation*}
We work with the second term:
\begin{equation*}
    \Pn\frac{(1-D_i)}{\E[\pi]^2}( \exp(\hlo_i)- \exp(\plo_i))^2 = \Pn\frac{(1-D_i)}{\E[\pi]^2}\exp(2\plo_i)( \exp((1-D_i)(\hlo_i-\plo_i))- 1)^2.
\end{equation*}
Define $y_i:=\frac{1-D_i}{\E[\pi]}\exp(\plo_i)$ and $x_i = (1-D_i)(\hlo_i - \plo_i)$, then we can decompose the term:
\begin{align*}
    \Pn y_i^2( \exp(x_i)- 1)^2 =\Pn y_i^2\{x_i> 0\} (\exp(x_i) - 1)^2 + \Pn y_i^2\{x_i<0\} (\exp(x_i) - 1)^2 \le \\
    \Pn \{x_i>0\} y_i^2\exp(2x_i)x_i^2 +
    \Pn  \{x_i<0\}y_i^2 (\exp(2x_i) - 2 \exp(x_i) +1) \le \\
    \max_i\exp(2x_i\{x_i \ge 0\})\Pn \{x_i>0\} y_i^2x_i^2 + \Pn  \{x_i<0\}y_i^2 (\ell(-2x_i) - 2\ell(-x_i)) =\\
        \max_i\exp(2x_i\{x_i \ge 0\})\Pn \{x_i>0\} y_i^2x_i^2 + \Pn  \{x_i<0\}y_i^2 (\ell(2|x_i|) - 2\ell(|x_i|)) \le \\
         \max_i\exp(2x_i\{x_i \ge 0\})\Pn \{x_i>0\} y_i^2x_i^2 + \max_{i}\frac{(\ell(2|x_i|) - 2\ell(|x_i|)) }{ \ell(2|x_i|)}\Pn  \{x_i<0\}y_i^2 \ell(2|x_i|).
\end{align*}
By construction, we have:
\begin{equation*}
 \max_i\frac{(\ell(2|x_i|) - 2\ell(|x_i|)) }{ \ell(2|x_i|)} \le 1, \quad \{x_i<0\}\ell(2|x_i|) \le  C|x_i|,
\end{equation*}
and
\begin{equation*}
 \{x_i>0\} x_i^4 \le 24 \ell(-x_i\{x_i >0\}) \le 24 \ell(-x_i).
\end{equation*}
We then have:
\begin{align*}
    \Pn \{x_i>0\} y_i^2x_i^2 =  \Pn y_i \left(y_ix_i^2\{x_i>0\}\right) \le \sqrt{\Pn y_i^3}\sqrt{ \Pn y_ix_i^4\{x_i>0\}} \le \\
     \sqrt{24}\sqrt{\Pn y_i^3}\sqrt{ \Pn y_i\ell(-x_i)}
\end{align*}
Substituting back the values of $x_i$ and $y_i$ we get from Lemma \ref{lem:conv_rate_1}:
\begin{align*}
    \Pn y_i\ell(-x_i) = \Pn \frac{1-D_i}{\E[\pi]}\exp(\plo_i)\ell((1-D_i)(\plo_i - \hlo_i)) \le \\
    \Pn \frac{1-D_i}{\E[\pi]}\exp(\plo_i)\ell(\plo_i - \hlo_i) =\frac{p}{\E[\pi]n} = o_p(1)
\end{align*}

We also have:
\begin{equation*}
    \Pn  \{x_i<0\}y_i^2 \ell(2|x_i|) \le C\Pn  y_i^2 |x_i| \le 
    C\sqrt{\Pn  y_i^4}\sqrt{\Pn x_i^2}
\end{equation*}
Substituting $x_i$ and using Lemma \ref{lem:conv_rate_1} we get:
\begin{equation*}
    \Pn x_i^2 = \Pn(1-D_i)(\plo_i - \hlo_i)^2 \le \Pn(\plo_i - \hlo_i)^2 = O_p\left(\frac{p}{\E[\pi]n}\right) = o_p(1).
\end{equation*}
We have:
\begin{align*}
    \left(\Pn y_i^3\right)^{\frac43} \le \Pn y_i^4 =\Pn\left(\frac{1-D_i}{\E[\pi]}\exp(\plo_i)\right)^4 = \\\max_{i}\exp(4(\plo_i - \muplo_i))\Pn\left(\frac{1-D_i}{\E[\pi]}\exp(\muplo_i)\right)^4 = 
    O_p\left(\E\left(\frac{1-D}{\E[\pi]}\exp(\muplo)\right)^4\right) 
\end{align*}
where third equlaity follows by Markov inequality, Corollary \ref{cor:pop_proj}, Assumption \ref{as:subg} and maximal inequality. We next compute the expectation:
\begin{multline*}
    \E\left[\left(\frac{1-D}{\E[\pi]}\exp(\muplo)\right)^4\right] =  \E\left[\left(\frac{1}{\E[\pi]}\exp(\muplo)\right)^4(1-\pi)\right] = \E\left[\left(\frac{\pi}{\E[\pi]}\exp(\muplo-\lo)\right)^4\frac{1}{(1-\pi)^3}\right] \le \\
    4 \E\left[\left(\frac{\pi}{\E[\pi]}\exp(\muplo-\lo)\right)^4\right] + 4 \E\left[\left(\frac{\pi}{\E[\pi]}\exp(\muplo-\lo)\right)^4 \exp(3\lo)\right],
\end{multline*}
where the last inequality follows from the following implications:
\begin{equation*}
    \frac{1}{1-\pi} = \frac{1}{1 - \frac{\exp(\lo)}{1+\exp(\lo)}} = \frac{1+\exp(\lo)}{1 + \exp(\lo) - \exp(\lo)} = 1+\exp(\lo) \Rightarrow \left(\frac{1}{1-\pi}\right)^{3} \le 4 (1 + \exp(3\lo))
\end{equation*}
Terms $\E\left[\left(\frac{\pi}{\E[\pi]}\exp(\muplo-\lo)\right)^4\right]$ and $\E\left[\left(\frac{\pi}{\E[\pi]}\exp(\muplo-\lo)\right)^4 \exp(3\lo)\right]$ are $O(1)$ by Corollary \ref{cor:mom_prod}.

Finally, from Lemma \ref{lem:conv_rate_3} we have:
\begin{equation*}
    \max_i\exp(2x_i\{x_i \ge 0\}) = \exp(\max\{\max_i(1-D_i)(\hlo_i - \plo_i), 0\}) = O_p(1).
\end{equation*}
Using conditional Chebyshev's inequality, we can thus conclude:
\begin{equation*}
    \Pn\frac{(1-D_i)}{\E[\pi]}( \exp(\hlo_i)- \exp(\plo_i))\epsilon_i = o_p\left(\frac{1}{\sqrt{n}}\right),
\end{equation*}
thus finishing the analysis of the first part of the empirical error.

Next, we analyze the second part of the empirical error:
\begin{align*}
    \Pn\frac{(1-D_i)}{\E[\pi]}( \exp(\hlo_i)- \exp(\plo_i))(\mu_i- \hmu_i) = \\
    \Pn\frac{(1-D_i)}{\E[\pi]}\exp( \plo_i)\ell(\plo_i -\hlo_i)(\mu_i- \hmu_i) +  
    \Pn\frac{(1-D_i)}{\E[\pi]}\exp(\plo_i)(\hlo_i -\plo_i)(\mu_i- \hmu_i)
\end{align*}
Taking the first term, we have the following from Lemma \ref{lem:conv_rate_1}:
\begin{align*}
   \left| \Pn\frac{(1-D_i)}{\E[\pi]}\exp( \tilde g_i)\ell( \plo_i -\hlo_i)(\mu_i- \hmu_i)\right|\le \\
   \max_i|\mu_i- \hmu_i| \Pn\frac{(1-D_i)}{\E[\pi]}\exp( \plo_i)\ell(\plo_i -\hlo_i) = \\
   O_p \left(\frac{\sqrt{\log(n)}\|\mu - \bestmu\|_2p}{\E[\pi]n}\right) 
\end{align*}
For the second one, we have the following: 
\begin{equation}
\label{eq:cross-term-empirical}
\begin{aligned}
    \Pn\frac{(1-\pi_i)}{\E[\pi]}\exp( \plo_i)(\hlo_i -\plo_i)(\mu_i- \hmu_i) = 
     \Pn\frac{(1-D_i)}{\E[\pi]}\exp(\plo_i)(\hlo_i -\plo_i)(\mu_i- \hmu_i) + \\
     \Pn\frac{\pi_i - D_i}{\E[\pi]}\exp(\plo_i)(\hlo_i -\plo_i)(\mu_i- \hmu_i).
\end{aligned}
\end{equation}
We analyze the first part:
\begin{align*}
     \Pn\frac{(1-\pi_i)}{\E[\pi]}\exp( \plo_i)(\hlo_i -\plo_i)(\mu_i- \hmu_i) = 
     \Pn\frac{\pi_i}{\E[\pi]}\exp( \muplo_i-\lo_i)(\hlo_i -\plo_i)(\mu_i- \hmu_i) + \\
     \Pn\frac{\pi_i}{\E[\pi]}\exp(\muplo_i -\lo_i)(\exp( \nu_{\muplo,i})-1)(\hlo_i -\plo_i)(\mu_i- \hmu_i)
\end{align*}
The first term is equal to zero by the FOC \eqref{eq:foc_hmu} because $\hlo - \plo \in \mt{F}$. 
%\dmitry{Changes if we go to RKHS}

The second term has the following form:
\begin{align*}
\Pn\frac{\pi_i}{\E[\pi]}\exp(\muplo_i -\lo_i)(\exp( \nu_{\muplo,i})-1)(\hlo_i -\plo_i)(\mu_i- \hmu_i)\le \\
\max_i| \mu_i- \hmu_i| \sqrt{\Pn\left[\frac{\pi_i}{\E[\pi]}\exp(\muplo_i -\lo_i)(\exp( \nu_{\muplo,i})-1)\right]^2}\sqrt{\Pn(\hlo_i -\plo_i)^2}
\end{align*}
We have from Corollarys \ref{cor:pop_proj} and \ref{cor:mom_prod}, Assumption \ref{as:subg} and maximal inequality:
\begin{align*}
    \Pn\left[\frac{\pi_i}{\E[\pi]}\exp(\muplo_i -\lo_i)(\exp( \nu_{\muplo,i})-1)\right]^2 \le \\
    \max\exp(2|\nu_{\muplo,i}|) \sqrt{\Pn\left[\frac{\pi_i}{\E[\pi]}\exp(\muplo_i -\lo_i)\right]^4}\sqrt{\Pn\nu^4_{\muplo,i}} = O_p\left(\| \nu_{\muplo}\|_4^2\right) =\\
    O_p\left(\| \nu_{\muplo}\|_2^2\right)  = O_p(\|\mu - \bestmu\|_2^2).
\end{align*}
It follows that we have from Lemma \ref{lem:conv_rate_1}:
\begin{align*}
    \Pn\frac{\pi_i}{\E[\pi]}\exp(\muplo_i -\lo_i)(\exp( \nu_{\muplo,i})-1)(\hlo_i -\plo_i)(\mu_i- \hmu_i) = \\
    O_p\left(\sqrt{\log(n)}\|\mu - \bestmu\|_2^2\sqrt{\Pn(\hlo_i -\plo_i)^2}\right) = O_p\left(\|\mu - \bestmu\|_2^2\frac{\sqrt{\log(n)}p}{\sqrt{\E[\pi]n}}\right) = o_p\left(\|\mu - \bestmu\|_2^2\right)
\end{align*}

Because $\hmu_i$ is conditionally independent of $D_i$, we can think of the second term in
\eqref{eq:cross-term-empirical} using a multiplier process, i.e., using the bound
\begin{equation}
\label{eq:cross-term-empirical-second}
\begin{aligned}
\abs*{\Pn\frac{(\pi_i-D_i)\exp(\plo_i)}{\E[\pi]}(\hlo_i -\plo_i)(\mu_i- \hmu_i)}
&= \abs*{\Pn \xi_i (\hlo_i -\plo_i)(\mu_i - \hmu_i)} \qfor \xi_i = \frac{(\pi_i-D_i)\exp(\plo_i)}{\E[\pi]} \\
&\le c\sup_{\substack{\delta \in \mt{F} \\ \|\delta\|_2 \le r_{\lo}}}\abs*{\Pn \xi_i \delta_i (\mu_i - \hmu_i)} \qqtext{ when } \norm{\hlo-\plo}_2 \le cr_{\lo}.
\end{aligned}
\end{equation}
Here $\E[\xi_i \mid X_i, \nu_i]=0$. We will contract out the factors $\mu_i-\hmu_i$, leaving us with the process that we now bound using a multiplier inequality of \citet[Corollary 1.10]{mendelson2016upper}: with probability tending to one for any $q > 2$,
\begin{equation}
\label{eq:cross-term-empirical-second-simplified}
\sup_{\substack{\delta \in \mt{F} \\ \|\delta\|_2 \le r_{\lo}}}\abs*{\Pn \xi_i \delta_i} 
\le \norm{\xi_i}_q \E \sup_{\substack{\delta \in \mt{F} \\ \|\delta\|_2 \le r_{\lo}}} \abs*{ \Pn g_i \delta_i }
\le c \sqrt{E[\pi]}r_{\lo}^2.
\end{equation}
The last bound here follows from the fixed-point condition defining $r_\lo$ and the $O(1)$
bound $\norm{\xi_i}_q$ for $q=3$ established below.

To do this, we will use the following result which follows from \citet[Proposition 3.1.23][]{gine2021mathematical} and the Montgomery-Smith Reflection Principle. If $Y_1, Y_2, \ldots$ are iid sample bounded processes indexed by $F$,
\begin{equation}
\begin{aligned}
P\qty(\sup_{f \in F} \abs{\sum_i a_i Y_i(f)} \ge t\max_i \abs{a_i}) \le c P\qty(\sup_{f \in F} \abs{\sum_i Y_i(f)} \ge ct )
\end{aligned}
\end{equation}
Taking $Y_i=\xi_i \delta_i$ and $F=\qty{ \delta \in \mt{F}  \ : \  \norm{\delta}_2 \le r_{\lo} }$
and $a_i=\mu_i - \hmu_i$, \eqref{eq:cross-term-empirical-second-simplified} implies that the probability on the right tends to zero for $t=n \times c\sqrt{E[\pi]}r_{\lo}^2$, so this implies that the probability on the left does as well.
This is the probability that the bound from \eqref{eq:cross-term-empirical-second} exceeds $c\sqrt{E[\pi]}r_{\lo}^2 \max_{i}\abs{\mu_i-\hmu_i}$. Now observing that on an event of probability $1-\delta$ the `when' clause from \eqref{eq:cross-term-empirical-second} is satisfied 
for $r=r_{\lo}$ and, via Lemma~\ref{lem:or_proj}, $\max_{i}\abs{\mu_i-\hmu_i} \le c\min\qty{r_{\mu}, \norm{\mu-\pmu}_2})$, we see that
the second term in \eqref{eq:cross-term-empirical} is $O_p(\min\qty{r_{\mu}, \norm{\mu-\pmu}_2} \sqrt{E[\pi]\log(n)}r_{\lo}^2)$.

For the moment bound on our multiplier, we used Corollary \ref{cor:mom_prod}:
\begin{align*}
    \left\|\frac{(\pi-D)\exp(\plo)}{\E[\pi]}\right\|_{3}^3 = \E\left[\left(\frac{|D-\pi|}{1-\pi}\right)^3\left(\frac{\pi\exp(\plo-\lo)}{\E[\pi]}\right)^3\right] = \\
    \E\left[\pi\left(\frac{\pi\exp(\plo-\lo)}{\E[\pi]}\right)^3\right] + \E\left[\exp(3\lo)\left(\frac{\pi\exp(\plo-\lo)}{\E[\pi]}\right)^3\right]  = O(1).
\end{align*}

Combining all the results together, we can conclude:
\begin{equation*}
    \Pn\frac{(1-D_i)}{\E[\pi]}( \exp(\hlo_i)- \exp(\plo_i))(\mu_i- \hmu_i) = O_p \left(\frac{\sqrt{\log(n)}\|\mu - \bestmu\|_2p}{\E[\pi]n}\right)  + o_p(\|\mu - \bestmu\|_2^2).
\end{equation*}
\end{proof}

\newpage
\subsection{Balancing}
Finally, we establish the results for the balancing term. 

\begin{lemma}\label{lem:balancing}
    Suppose Assumption \ref{as:subg} - \ref{as:out_mom} hold,  $\zeta = O(1)$. Then we have
   \begin{equation*}
       \Pn \left[(1-D_i) \hat\omega_i - \frac{D_i}{\overline\pi}\right]\hmu_i = o_p\left(\frac{1}{\sqrt{n}}\right).
   \end{equation*} 
\end{lemma}
\begin{proof}
By definition, we have:
    \begin{equation*}
    \Pn \left[(1-D_i) \hat\omega_i - \frac{D_i}{\overline\pi}\right]\hmu_i = \frac{1}{\overline\pi}\Pn[(1-D_i) \exp(\hlo_i) - D_i] \hmu_i.
\end{equation*}
From the first order condition for the empirical problem we have for any $f \in \mt{F}$:
\begin{equation*}
    \frac{1}{\overline\pi}\Pn[(1-D_i) \exp(\hlo_i) - D_i]f_i = -\frac{\zeta^2}{n} <\hlo, f>_{\m{F}}
\end{equation*}
Using this equality for $f = \hmu$ we get:
\begin{equation*}
    \left|\frac{1}{\overline\pi}\Pn[(1-D_i) \exp(\hlo_i) - D_i]\hmu_i\right| = \frac{\zeta^2}{\sqrt{n}} \frac{|<\hlo, \hmu>_{\m{F}}|}{\sqrt{n}} .
\end{equation*}

\begin{equation*}
    |<\hlo, \hmu>_{\m{F}}| \le \|\hlo\|_{\m{F}} \|\hmu\|_{\m{F}}  \le \left( \|\plo - \hlo\|_{\m{F}} + \|\plo \|_{\m{F}}\right)\qty(\|\bestmu\|_{\m{F}} + \|\bestmu - \hmu\|_{\m{F}})
\end{equation*}
By Assumption \ref{as:out_mom} we have  $ \|\plo - \hlo\|_{\m{F}} = O\qty(\|\plo - \hlo\|_2)$, $ \|\bestmu - \hmu\|_{\m{F}} = O\qty(\|\bestmu - \hmu\|_{2})$, $\|\plo \|_{\m{F}} = O\qty(\|\plo \|_{2})$, $\|\bestmu\|_{\m{F}} 
 = O\qty(\|\bestmu\|_2)$. Our results guarantee $\|\plo - \hlo\|_2 = o_p(1)$, $\|\bestmu - \hmu\|_{2} = o_p(1)$, and thus the dominant term is $\|\bestmu - \mathbb{E}[\bestmu]\|_2\times \|\plo \|_{2}$. Assumption \ref{as:out_mom} guarantees that $\|\bestmu- \mathbb{E}[\bestmu]\|_2 = O(1)$. We also have:
 \begin{equation*}
     \|\plo \|_{2} \le  \|\muplo - \plo \|_{2} + \| \lo -\muplo\|_2 + \| \lo\|_2
 \end{equation*}
 The first part behaves as $\| \mu - \bestmu\|_2$ by Corollary \ref{cor:pop_proj}, the second one is bounded by Lemma \ref{lem:init_bound}, and finally, for the last one, we use Assumption \ref{as:overlap} that guarantees $\| \lo\|_2 = o(\sqrt{n})$.

It thus follows:
\begin{equation*}
     \frac{1}{\overline\pi}\Pn[(1-D_i) \exp(\hlo_i) - D_i]\hmu_i = o_p\left(\frac{1}{\sqrt{n}}\right).
\end{equation*}
\end{proof}

\newpage
\subsection{Main result}
We connect all our previous results in the following general theorem which proves Theorem \ref{th:const_share} in the main text.
\begin{theorem}\label{th:main_theorem}
    Suppose Assumptions \ref{as:emp_set}, \ref{as:subg} - \ref{as:out_mom} hold, $\| \mu - \bestmu\|_2 \ll 1$ and $\zeta = O(1)$. Then we have:
    \begin{align*}
        \hat \tau - \tau= \bias + 
        \Pn \frac{\pi_i - D_i}{1-\pi_i}\frac{(\pi_iu_i + 1)}{\E[\pi]}\epsilon_i+\Pn\frac{\pi_iu_i}{\E[\pi]} \epsilon_i +
        o_p\qty(\frac{1}{\Tef}) + o_p\left(\frac{1}{\sqrt{\E[\pi] n}}\right) . 
    \end{align*}
\end{theorem}
\begin{proof}
By the expansion in \eqref{eq:decomp} we have:
\begin{equation*}
    \hat \tau - \tau = \xi_1 + \xi_2 + \xi_3
\end{equation*}
Lemma \ref{lem:balancing} implies that $\xi_3 = o_p\left(\frac{1}{\sqrt{n}}\right)$. Lemma \ref{lem:or_noise} implies:
\begin{equation*}
    \xi_1 = \frac{\E[\pi]}{\overline{\pi}}\left(\Pn \frac{\pi_i - D_i}{1-\pi_i}\frac{(\pi_iu_i + 1)}{\E[\pi]}\epsilon_i+\Pn\frac{\pi_iu_i}{\E[\pi]} \epsilon_i + o_p\left(\frac{1}{\sqrt{n}}\right)  
   \right) + \Pn(1-D_i)( \exp(\hlo_i)- \exp(\plo_i))\frac{\epsilon_i}{\overline{\pi}}
\end{equation*}
$ \frac{\E[\pi]}{\overline{\pi}} = 1 + O_p\left(\frac{1}{\sqrt{\E[\pi]n}}\right)$ as long as $\E[\pi] n\rightarrow \infty$ which is guaranteed by Assumption \ref{as:comp}. As a result, we get:
\begin{equation*}
    \xi_1 = \Pn \frac{\pi_i - D_i}{1-\pi_i}\frac{(\pi_iu_i + 1)}{\E[\pi]}\epsilon_i+\Pn\frac{\pi_iu_i}{\E[\pi]} \epsilon_i + o_p\left(\frac{1}{\sqrt{\E[\pi]n}} 
   \right) + \Pn(1-D_i)( \exp(\hlo_i)- \exp(\plo_i))\frac{\epsilon_i}{\overline{\pi}}
\end{equation*}
Using Lemma \ref{lem:emp_error} and $ \frac{\E[\pi]}{\overline{\pi}} = 1 + O_p\left(\frac{1}{\sqrt{\E[\pi]n}}\right)$ we get
\begin{align*}
    \Pn(1-D_i)( \exp(\hlo_i)- \exp(\plo_i))\frac{\epsilon_i}{\overline{\pi}} = o_p\left(\frac{1}{\sqrt{n}}\right),
\end{align*}
and thus:
\begin{align*}
    \Pn \frac{\pi_i - D_i}{1-\pi_i}\frac{(\pi_iu_i + 1)}{\E[\pi]}\epsilon_i+\Pn\frac{\pi_iu_i}{\E[\pi]} \epsilon_i + o_p\left(\frac{1}{\sqrt{\E[\pi]n}} \right).
\end{align*}
Similarly, Corollary \ref{coro:or_bias} guarantees: 
\begin{align*}
    &\xi_2 = \frac{\E[\pi]}{\overline{\pi}}\left(\bias +  o_p\left(\frac{1}{\sqrt{\E[\pi]n}}\right) + o_p(\|\mu - \bestmu\|_2^2) +  O_p\qty(\norm{\mu-\bestmu}_2 \frac{p}{n}) \right)  \\
&+  \Pn(1-D_i)( \exp(\hlo_i)- \exp( g_i))\frac{(\mu_i  - \hmu_i)}{\overline{\pi}}
\end{align*}
By the same logic as above, this implies:
\begin{align*}
    &\xi_2 = \bias +  o_p\left(\frac{1}{\sqrt{\E[\pi]n}}\right) + o_p\qty(\|\mu - \bestmu\|_2^2) +  O_p\qty(\norm{\mu-\bestmu}_2 \frac{p}{n}) + \\
    &+  \Pn(1-D_i)( \exp(\hlo_i)- \exp( g_i))\frac{(\mu_i  - \hmu_i)}{\overline{\pi}}.
\end{align*}
Using Lemma \ref{lem:emp_error} and $ \frac{\E[\pi]}{\overline{\pi}} = 1 + O_p\left(\frac{1}{\sqrt{\E[\pi]n}}\right)$ we get
\begin{align*}
     \Pn(1-D_i)( \exp(\hlo_i)- \exp( g_i))\frac{(\mu_i  - \hmu_i)}{\overline{\pi}} = O_p\left(\frac{\sqrt{\log(n)}\|\mu - \bestmu\|_2p}{\E[\pi] n}\right) + o_p(\|\mu - \bestmu\|_2^2)
\end{align*}
Since $\frac{\|\mu - \bestmu\|_2p}{\E[\pi] n} = \frac{p}{\E[\pi] n} \frac{1}{\sqrt{\Tef}}$, Assumption \ref{as:comp} guarantees:
\begin{equation*}
      \frac{p}{\E[\pi] n} \ll \max\left\{\frac{1}{\sqrt{\Tef}}, \sqrt{\frac{\Tef}{\mathbb{E}[\pi]n}}\right\} \Leftrightarrow  \frac{p}{\E[\pi] n} \frac{1}{\sqrt{\Tef}} \ll \max\left\{\frac{1}{\Tef}, \frac{1}{\sqrt{\mathbb{E}[\pi]n}}\right\},  
\end{equation*}
which implies $ O_p\left(\frac{\sqrt{\log(n)}\|\mu - \bestmu\|_2p}{\E[\pi] n}\right) = o_p\qty(\max\left\{\frac{1}{\Tef}, \frac{1}{\sqrt{\mathbb{E}[\pi]n}}\right\})$. Combining all the terms we get:
\begin{equation*}
    \hat \tau - \tau  = \bias + \Pn \frac{\pi_i - D_i}{1-\pi_i}\frac{(\pi_iu_i + 1)}{\E[\pi]}\epsilon_i+\Pn\frac{\pi_iu_i}{\E[\pi]} \epsilon_i + o_p\qty(\frac{1}{\Tef}) + o_p\left(\frac{1}{\sqrt{\E[\pi] n}}\right). 
\end{equation*}
\end{proof}
Theorem \ref{th:van_share} follows by observing that $\mathbb{V}\left[\Pn\frac{\pi_iu_i}{\E[\pi]} \epsilon_i\right] = \frac{1}{n}\mathbb{E}\qty[\left(\frac{\pi u}{\E[\pi]}\right)^2 \epsilon)^2] \le \frac{\sigma^2_{\max}}{n} \mathbb{E}\qty[\left(\frac{\pi u}{\E[\pi]}\right)^2] =O\left(\frac{1}{n}\right)$, where the last equality follows from \eqref{u:moment-bound}. By the same type of argument, we have 
\begin{align*}
    \mathbb{V}\left[ \Pn \frac{\pi_i - D_i}{1-\pi_i}\frac{\pi_iu_i}{\E[\pi]}\epsilon_i\right] \le \frac{\sigma^2_{\max}}{n}\mathbb{E}\qty[\exp(\lo)\frac{\pi u}{\E[\pi]}] = O\qty(\frac{1}{n}).
\end{align*}
As a result, both of these terms are negligible compared to $\Pn \frac{\pi_i - D_i}{1-\pi_i}\frac{\epsilon_i}{\E[\pi]} = O_p\qty(\frac{1}{\sqrt{E[\pi]n}})$ in the regime where $\E[\pi] \ll 1$.
\newpage
\subsection{Miscellaneous proofs}\label{ap:misc}

\subsubsection{Quadratic minimization}\label{sec:quadr_min}
Consider an arbitrary $d\times T_0$ matrix $A$, and a $d\times 1$ dimensional vector $b$. Define the minimal value of the following optimization problem:
\begin{equation*}
    V := \min_{x} \| Ax - b\|_2^2 +\sigma^2\|x\|_2^2. 
\end{equation*}
Let $A = UDV^\top$ be the SVD decomposition of $A$, then $x^{\star}$ that solve the optimization problem is equal to $ x^{\star} = V (D^2 + \sigma^2 \mathcal{I}_d)^{-1}D U^\top b)$. Substituting this value in the optimization problem, we get 
\begin{equation*}
    V = \sigma^2 (\sigma^2 b^\top U(D^2 + \sigma^2 \mathcal{I}_d)^{-2} U^\top b + \sigma^2 b^\top UD(D^2 + \sigma^2 \mathcal{I}_d)^{-2}D U^\top b 
    =\sigma^2 b^\top U(D^2 +\sigma^2\mathcal{I}_d)^{-1}U^\top b
\end{equation*}
Defining $\xi :=U^\top b$, we get a simpler expression for the same value:
\begin{equation*}
    V = \sigma^2\sum_{j = 1}^d \frac{\xi_j^2}{d^2_j + \sigma^2}.
\end{equation*}
Suppose $\xi_j^2 \sim \frac{d_j^2}{T_0} \sim j^{-p}$, then we can split the sum into two parts: $T_0 j^{-p} > \sigma^2 \Rightarrow j \le \qty(\frac{T_0}{\sigma^2})^{\frac1p}$, and $j > \qty(\frac{T_0}{\sigma^2})^{\frac1p}$. The first part of the sum behaves as $\qty(\frac{T_0}{\sigma^2})^{\frac1p}\frac{\sigma^2}{T_0}$, and the second part of the sum behaves as $\int_{x \ge  \qty(\frac{T_0}{\sigma^2})^{\frac1p}} x^{-p} dx \sim \qty(\frac{T_0}{\sigma^2})^{\frac{-p+1}{p}}$. Combining the two parts together we get $V\sim\left(\frac{\sigma^2}{T_0}\right)^{1-1/p}$.

\subsubsection{Proof of Corollaries \ref{cor:two-way} - \ref{cor:int_fe}}
We prove the second corollary because the first one follows from it. Our goal is to verify that conditions of Theorem \ref{th:const_share} hold.  Assumption \ref{as:subg} holds because, by definition, any random variable in $\spn{\m{F},\mu,\lo}$ is a linear combination of independent subgaussian random variables with subgaussian norms controlled by the $L^2$ norms. By definition, $\lo \ge \alpha_c$ with positive probability bounded away from zero. This implies that the first part of Assumption \ref{as:overlap} holds. It also guarantees that $\mathbb{E}[\pi]$ is bounded away from zero, and thus the second part of this assumption also holds. The first part is guaranteed because $\lo > c$ with positive probability. By definition, $\epsilon_{t}$ are uncorrelated and have variance bounded from zero and infinity, guaranteeing that Assumption \ref{as:out_mom} hold. We also have $\|\Sigma\|_{op} \le \max_{K+1\le t\le T_0}\mathbb{V}[\epsilon_t] \le \sigma^2_{\max}$ and thus using the derivations from Section \ref{sec:quadr_min} we get:
\begin{equation*}
    \min_{\mathbf{c}}\left\{\left\|\sum_{t = K+1}^{T_0} c_t\psi  -\psi\right\|_{2}^2 + \|\Sigma\|_{op}\sum_{t=K+1}^{T_0}c_{t}^2\right\} =\left(\frac{\sigma^2_{\max}}{T_0^{1-\frac{1}{\kappa}}}\right).
\end{equation*}
This implies that $\Tef \sim T_0^{1-\frac{1}{\kappa}}$, and we also have $T_0\sim p$ which simplifies Assumption \ref{as:comp}:
    \begin{equation*}
   n\min\left\{\max\left\{\frac{T_0^{\frac{\kappa-1}{2\kappa}}}{\sqrt{n}}, \frac{1}{T_0^{\frac{\kappa-1}{\kappa}}}\right\}, 1\right\}\gg T_0 \Leftrightarrow  n \gg T_0^{1+ \frac{1}{\kappa}}.
\end{equation*}

It remains to verify Assumption \ref{as:miss_lo}. Since the variance of $\lo$ is fixed, the first part of Assumption \ref{as:miss_lo} follows by considering $f^{\star} = \E[\lo]$ and observing that by definition:
\begin{equation*}
    \E\left[\frac{\pi}{\E[\pi]}\ell(\lo - \muplo)\right]+\frac{\zeta^2}{2n}\|f\|^2_{\m{F}} \le \E\left[\frac{\pi}{\E[\pi]}\ell(\lo - f^{\star})\right]+\frac{\zeta^2}{2n}\|f^{\star}\|^2_{\m{F}} = \E\left[\frac{\pi}{\E[\pi]}\ell(\lo - \E[\lo])\right] = O(1),
\end{equation*}
where the penultimate inequality follows from the fact that $\|f^{\star}\|^2_{\m{F}} = 0$ by definition of $\|\cdot \|_{\m{F}}$. To verify the second part of Assumption \ref{as:miss_lo} suppose that we have $\| \plo - \muplo\|_2^2 = O\qty(\frac{1}{\Tef})$. We then have:
\begin{align*}
    &\| \plo - \muplo\|_2^2 = \| \plo - f_{\muplo}- \beta_{\mu}\mu\|_2^2 = \| \plo - f_{\muplo}- \beta_{\mu}\bestmu  - \beta_{\mu}(\mu- \bestmu)\|_2^2 \\
    &=\| \plo - f_{\muplo}- \beta_{\mu}\bestmu\|_2^2 + \beta_{\mu}^2 \| \mu- \bestmu\|_2^2 = \| \plo - f_{\muplo}- \beta_{\mu}\bestmu\|_2^2 + \frac{\beta_{\mu}^2}{\Tef} = O\qty(\frac{1}{\Tef}) \Rightarrow \beta_{\mu}^2 = O(1),
\end{align*}
where the second equality follows because $\plo - f_{\muplo}- \beta_{\mu}\bestmu\in \mt{F}$ and $\mu - \bestmu$ is orthogonal to that space. To bound $\| \plo - \muplo\|_2^2$ we use 
$f^{opt} = \argmin_{f\in \mt{F}}\| \lo  - f\|_2^2$ and the following fact:
\begin{align*}
     \E\left[\frac{\pi}{\E[\pi]}\ell(\lo - \muplo)\right]+\frac{\zeta^2}{2n}\|\muplo\|^2_{\m{F}} \le \E\left[\frac{\pi}{\E[\pi]}\ell(\lo - \plo)\right]+\frac{\zeta^2}{2n}\|\plo\|^2_{\m{F}} \le 
     \E\left[\frac{\pi}{\E[\pi]}\ell(\lo - f^{opt})\right]+\frac{\zeta^2}{2n}\|f^{opt}\|^2_{\m{F}} \le \\
      \E\left[\frac{\pi}{\E[\pi]}\ell(\lo - f^{\star})\right] = O(1)
\end{align*}
The same logic as previously for $\mu$ guarantees that $\|\lo - f^{opt}\|_2^2 = O\left(\frac{1}{\Tef}\right)$. Moreover, by construction $\| f^{opt}\|^2_{\m{F}} = O\qty(\sum_{j = 0}^k \alpha_j^2 +\frac{1}{\Tef}) = O(1)$. It then follows from a trivial extension of Lemma \ref{lem:ell_eq}:\footnote{The proof of Lemma \ref{lem:ell_eq} and results it relies on does not depend on the second part of Assumption \ref{as:miss_lo}, so there is no circularity in this argument.} 
\begin{equation*}
\E\left[\frac{\pi}{\E[\pi]}\ell(\lo - f^{opt})\right] = O\qty(\frac{1}{\Tef})\Rightarrow   \E\left[\frac{\pi}{\E[\pi]}\ell(\lo - \plo)\right] =  O\qty(\frac{1}{\Tef} +\frac{1}{n}) \Rightarrow\E\left[\frac{\pi}{\E[\pi]}\ell(\lo - \muplo)\right] =  O\qty(\frac{1}{\Tef} +\frac{1}{n}).
\end{equation*}
Applying the same lemma in other direction, we get $\|\lo - \muplo\|_2 = O\qty(\frac{1}{\sqrt{\Tef}} +\frac{1}{\sqrt{n}})$, $\|\lo - \plo\|_2 = O\qty(\frac{1}{\sqrt{\Tef}} +\frac{1}{\sqrt{n}})$, and as a result $\| \muplo - \plo\|_2^2 =  O\qty(\frac{1}{\Tef} +\frac{1}{n})$. Since $n \gg \Tef$ from Assumption \ref{as:comp}, we have $\| \plo - \muplo\|_2^2 = O\qty(\frac{1}{\Tef})$. Finally, because $\|\lo - \plo\|_2 = o(1)$ it follows that all terms involving $u$ are negligible.

\newpage
\section{Technical lemmas}
\begin{lemma}\label{lem:max_in}
    Suppose we have a collection $\{(X_i,Y_i)\}_{i=1}^n$, where $n \ge n_0$ such that each $X_i$ is subgaussian conditional on $Y_i$ with $\|X_i\|_{\psi_2|Y_i} \le \sigma(Y_i)$ and $\max_{i}\sigma(Y_i) = O_p(r_n)$. Then $\max_i X_i = O_p\left(r_n \sqrt{\log(n)}\right)$.
\end{lemma}
\begin{proof}
     Consider arbitrary $\sigma,\delta >0$ and observe:
\begin{align*}
     \E\left[\left\{\max_i X_i \ge \delta\sigma\sqrt{\log(n)}\right\}\right] \le \\
     \E\left[\left\{\max_i X_i \ge \delta\sigma\sqrt{\log(n)}\right\}\{\max_{i} \sigma(Y_i) \le \sigma\}\right] + \E[\{\max_{i} \sigma(Y_i) > \sigma\}].
\end{align*}
For the first term we have:
\begin{align*}
    \E\left[\left\{\max_i X_i \ge \delta\sigma\sqrt{\log(n)}\right\}\{\max_{i} \sigma(Y_i) \le \sigma\}\right] \le \\
    \sum_{i=1}^n \E\left[\E\left[\left\{ X_i \ge \delta\sigma\sqrt{\log(n)}\right\}|Y_i\right]\{\sigma(Y_i) \le \sigma\}\right] \le \\
 \sum_{i=1}^n\E\left[\exp\left(-C\frac{\delta^2\sigma^2\log(n)}{\sigma^2(Y_i)}\right)\{\sigma(Y_i) \le \sigma\}\right]\le 
n \exp\left(-C\delta^2\log(n)\right) 
\end{align*}
Fix $\epsilon>0$ and choose $\delta_1$ such that $\E[\{\max_{i} \sigma(Y_i) > \delta_1 r_n\}] \le \frac{\epsilon}{2} $. Then choose $\delta_2$ such that $1-C(\delta_1\delta_2)^2 <0$ and  $n_0 \exp\left(-C(\delta_1\delta_2)^2\log(n_0)\right) \le  \frac{\epsilon}{2}$. It then follows:
\begin{equation*}
     \max_{n\ge n_0}\E\left[\left\{\frac{\max_{i\le n} X_i}{r_n\sqrt{\log(n)}} \ge \delta_1\delta_2\right\}\right] \le \epsilon,
\end{equation*}
which by definition implies $\max_i X_i = O_p\left(r_n \sqrt{\log(n)}\right)$.
\end{proof}

\section{Simulation details}\label{ap:pars}
In this section, we describe the parameters we used for the simulations in Section \ref{sec:sim} and the results of the additional simualtions. The parameters of the outcome model for the AR design satisfy
\begin{equation*}
    \eta_i \sim \mathcal{N}(0, \sigma^2_{\eta}), \quad  u^{AR}_{it} \sim \mathcal{N}(0, \sigma^2_{AR}), \quad  \epsilon^{AR}_{i1} \sim \mathcal{N}\left(0, \frac{\sigma^2_{AR}}{1-\rho^2}\right),
\end{equation*}
where $\sigma^2_{\eta} = 1$, $\mathcal{N}(0, \sigma^2_{AR}) = 1$, $\rho =0.5$. The assignment model for the AR design has the following form:
\begin{equation*}
    \E[D_i = 1| \eta_i, Y_i^{T_0}, \nu_i] = \frac{\exp\left(\eta_i + \beta^{AR}_{T_0}\epsilon^{AR}_{iT_0} + \beta^{AR}_{T_0-1}\epsilon^{AR}_{iT_0-1} + \nu_{i}\right)}{1+ \exp\left(\eta_i + \beta^{AR}_{T_0}\epsilon^{AR}_{iT_0} + \beta^{AR}_{T_0-1}\epsilon^{AR}_{iT_0-1} + \nu_{i}\right)} , \quad \nu_{i} \sim \mathcal{N}(0, \sigma^2_{\nu}).
\end{equation*}
where $\beta^{AR}_{T_0} = 0.5$, $\beta^{AR}_{T_0-1} = 0.25$, and $\sigma^2_{\nu} = 0.25$. Note that we do not assume a logit selection model and allow for the misspecification error $\nu_i$.

The parameters of the outcome model for the random walk design satisfy:
\begin{equation*}
    \quad u^{RW}_{it} \sim \mathcal{N}(0, \sigma^2_{RW}),
\end{equation*}
where $\sigma^2_{RW} = \frac{1}{8}$. The assignment model has the form:
\begin{equation*}
      \E[D_i = 1| \eta_i, Y_i^{T_0}, \nu_i] = \frac{\exp\left( \beta^{RW}_{T_0}\epsilon_{i,T_0}+ \nu_{i}\right)}{1+ \exp\left(\beta^{RW}_{T_0}\epsilon_{i,T_0} + \nu_{i}\right)} , \quad \nu_{i} \sim \mathcal{N}(0, \sigma^2_{\nu}),
\end{equation*}
with $\beta^{RW}_{T_0} = 0.1$. Note that the assignment model depends on $\eta_i$ and past outcomes because $\epsilon_{iT_0}$ is not observed. 

We report the results for the distribution of $t$-statistics in the RW and mixture designs in Figures \ref{fig:inf_rw_model} - \ref{fig:inf_mixt_model}. These results confirm the intuition described in the main text. For the RW design, quantiles of $t$-statistics based on $\hat \tau_k^{TWFE}$ and $\hat \tau_{k}^{SC}$ are aligned with the theoretical quantiles of the standard normal distribution. The situation is less positive for the mixture design, with both methods delivering biased results. We can view the effect of this in Table \ref{table:coverage_rw_mix} that reports the coverage of nominal $95\%$ CI for both methods in two designs. The coverage in RW design is close to the nominal one for both estimators and less so in the mixture design. Interestingly, the coverage is uniformly better for the SC estimator. 

\begin{figure}[t!]
    \centering
    \includegraphics[scale = 0.40]{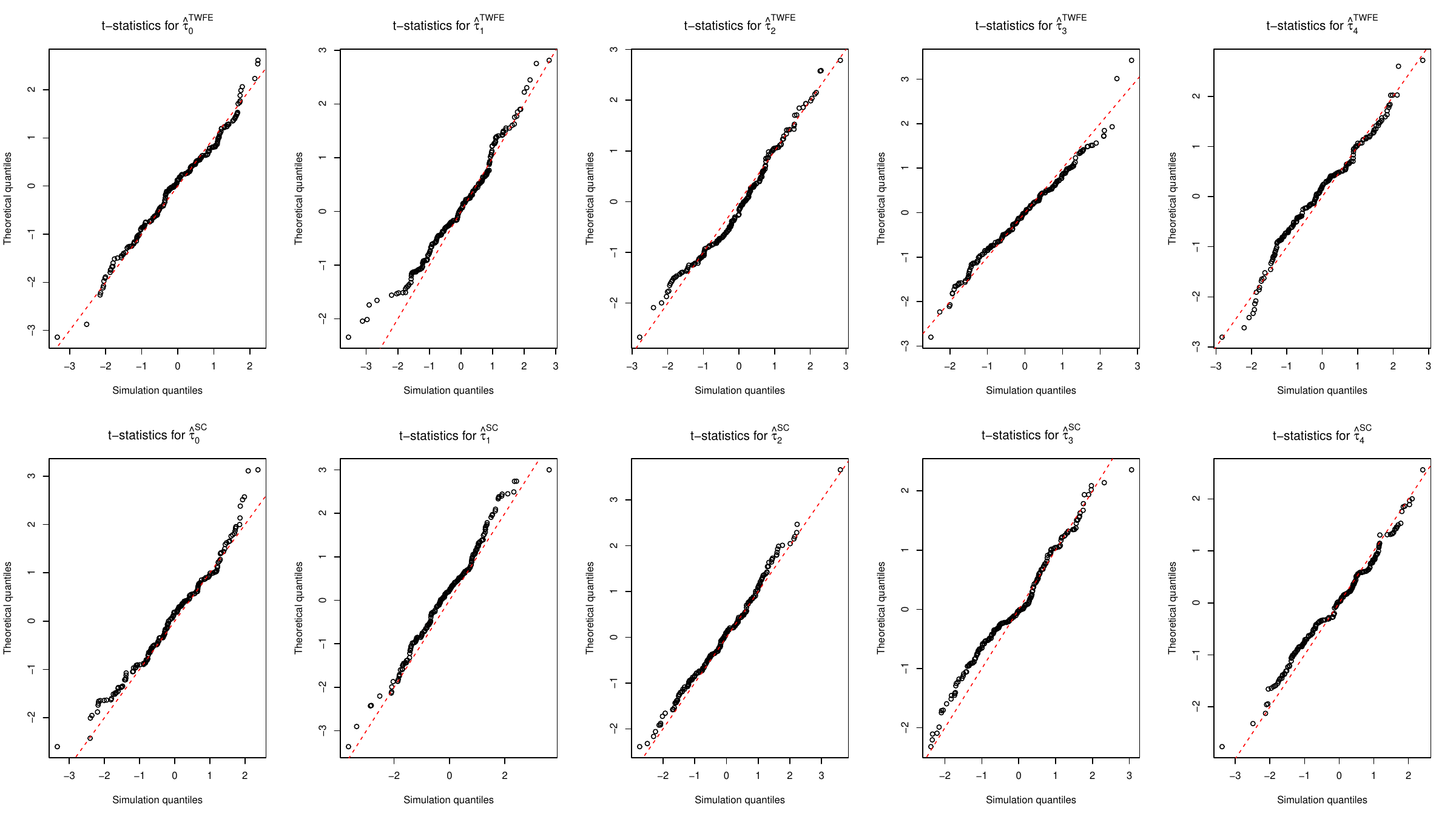}
    \caption{\footnotesize Computations based on $B = 200$ simulations with RW design. Each simulation has $n = 400$ units, $T_0 = 8$ pre-treatment periods, $K = 5$ treatment periods. The simulation parameters are reported in Appendix \ref{ap:pars}. First row: QQ plots for $t$-statistics based on the TWFE estimator; second row: QQ plots for $t$-statistics based on the SC estimator. Variance for each estimator is computed using $100$ bootstrap samples. }
    \label{fig:inf_rw_model}
\end{figure}

\begin{figure}[t!]
    \centering
    \includegraphics[scale = 0.40]{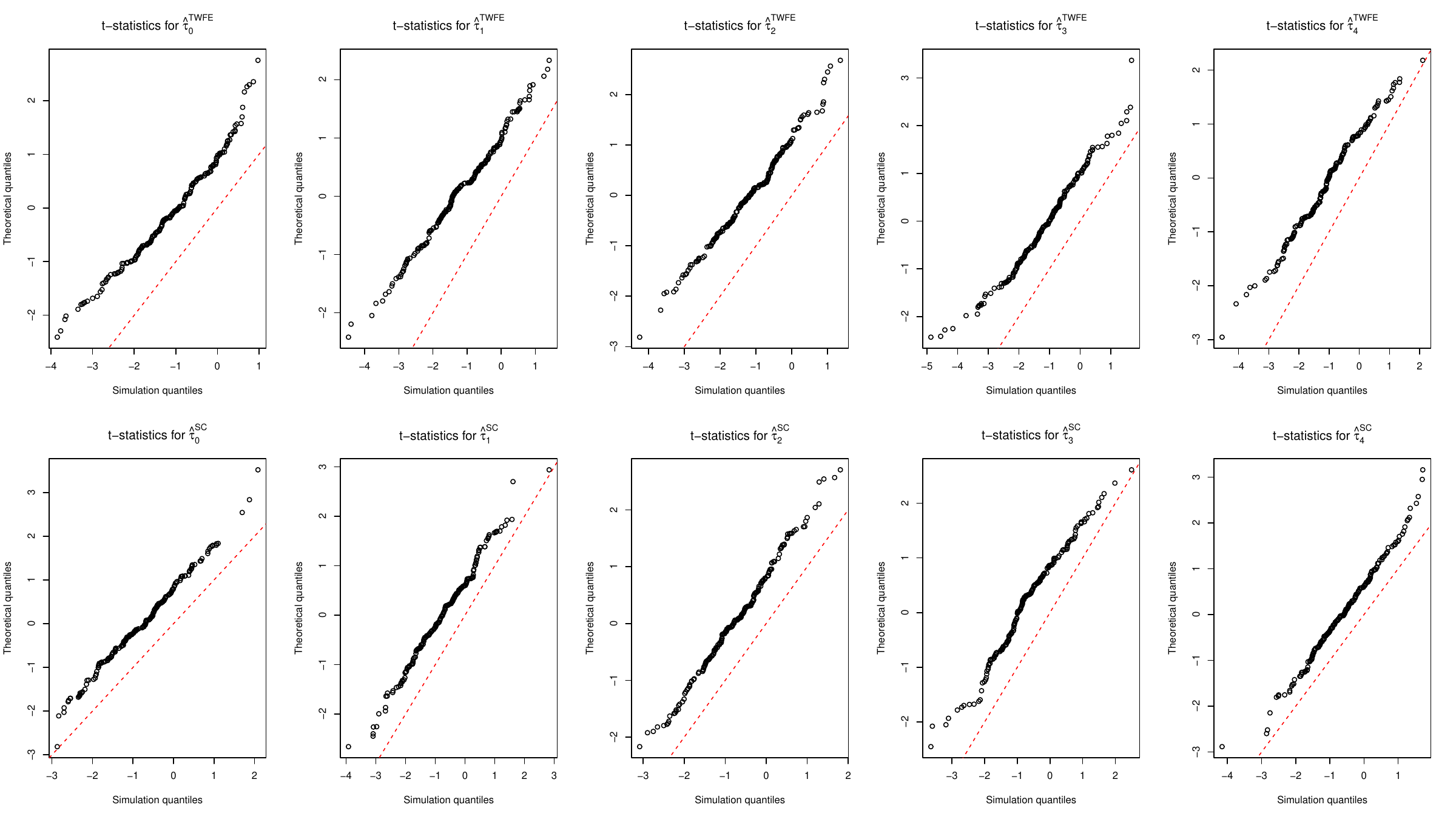}
    \caption{\footnotesize Computations based on $B = 200$ simulations with mixture design. Each simulation has $n = 400$ units, $T_0 = 8$ pre-treatment periods, $K = 5$ treatment periods. The simulation parameters are reported in Appendix \ref{ap:pars}. First row: QQ plots for $t$-statistics based on the TWFE estimator; second row: QQ plots for $t$-statistics based on the SC estimator. Variance for each estimator is computed using $100$ bootstrap samples. }
    \label{fig:inf_mixt_model}
\end{figure}

\begin{table}[t!]
\centering
\setlength{\extrarowheight}{0.2cm}
\begin{tabular}{|l|r|r|r|r|r|}
\hline\hline
\multirow{2}{*}{} & \multicolumn{5}{c|}{Effect in treatment period}\\\cline{2-6}
& 0 & 1 & 2 & 3 & 4 \\
\hline\hline
 \multicolumn{6}{|c|}{RW design}\\\cline{2-6}\hline 
CI based on $\hat \tau_{k}^{TWFE}$ & 0.94 & 0.94 & 0.94 & 0.95 & 0.96 \\
CI based on $\hat \tau_{k}^{SC}$ & 0.94 & 0.93 & 0.93 & 0.94 & 0.94 \\
\hline\hline
 \multicolumn{6}{|c|}{Mixture design}\\\cline{2-6}\hline 
CI based on $\hat \tau_{k}^{TWFE}$ & 0.80 & 0.77 & 0.77 & 0.78 & 0.79 \\
CI based on $\hat \tau_{k}^{SC}$ & 0.90 & 0.86 & 0.90 & 0.92 & 0.93 \\
\hline\hline
\end{tabular}
\caption{\footnotesize Coverage rates for $95\%$ confidence intervals based on $B =200$ simulations with RW and mixture designs. Each simulation has $n = 400$ units, $T_0 = 8$ pre-treatment periods, and $5$ treatment periods. The simulation parameters are reported in Appendix \ref{ap:pars}. In each block the first row: coverage rates based on $\hat \tau_{k}^{TWFE}$; the second row: coverage rate based on  $\hat \tau_{k}^{SC}$. Confidence intervals are constructed using \eqref{eq:conf_int}.}
\label{table:coverage_rw_mix}
\end{table}

\end{document}